\documentclass[letterpaper,12pt]{article}
 
\newcommand{\href}[1]{#1} 

\usepackage{amsmath,amssymb,amstext} 
\usepackage[pdftex]{graphicx} 
\usepackage{amsmath,amssymb,amstext,amsthm,amsfonts}
\usepackage{dsfont}
\usepackage[pdftex]{graphicx}
\usepackage{caption}
\usepackage[x11names]{xcolor}
\usepackage{dcolumn}
\usepackage{bm}
\usepackage{float}
\usepackage{multirow}
\usepackage[round]{natbib}   
\usepackage{listings} 
\usepackage{xtab}
\usepackage{authblk}

\usepackage{algorithm} 
\usepackage{algorithmicx, algpseudocode}

\usepackage[pdftex,pagebackref=false]{hyperref} 

\definecolor{background-color}{gray}{0.95}
\definecolor{steelblue}{rgb}{0.27, 0.51, 0.71}
\definecolor{brickred}{rgb}{0.8, 0.25, 0.33}
\definecolor{bluegray}{rgb}{0.4, 0.6, 0.8}
\definecolor{amethyst}{rgb}{0.6, 0.4, 0.8}
\hypersetup{
    plainpages=false,       
    unicode=false,          
    pdftoolbar=true,        
    pdfmenubar=true,        
    pdffitwindow=false,     
    pdfstartview={FitH},    
    pdftitle={Pooled\ Centrality},    
    pdfauthor={Chris Salahub},    
    pdfsubject={Statistics},  
    pdfnewwindow=true,      
    colorlinks=true,        
    linkcolor=steelblue,         
    citecolor=brickred,        
    filecolor=magenta,      
    urlcolor=cyan           
}

\graphicspath{.}

\newcommand{\code}[1]{\texttt{#1}}
\newcommand*{\Rnsp}{\textsf{R}}

\algblock{Input}{EndInput}
\algnotext{EndInput}

\newtheorem{definition}{Definition}
\newtheorem{theorem}{Theorem}
\newtheorem{lemma}{Lemma}
\newtheorem{proposition}{Proposition}

\newcommand{\ve}[1]{\mathbf{#1}}           
\newcommand{\tr}[1]{{#1}^{\mkern-1.5mu\mathsf{T}}}              

\newcommand{\abs}[1]{\lvert{#1}\rvert}              

\newcommand{\biggiven}{~\vline~}

\newcommand{\genpool}{g}
\newcommand{\ordpool}[2]{ord \left ( #1 ; #2 \right )}
\newcommand{\stopool}{Sto}
\newcommand{\gampool}{gam}
\newcommand{\tippool}{Tip}
\newcommand{\fispool}{Fis}
\newcommand{\pearpool}{Pea}
\newcommand{\chipool}[2]{chi \left ( #1 ; #2 \right )}
\newcommand{\hrpool}[2]{HR \left ( #1 ; #2 \right)}
\newcommand{\centquot}{q}
\newcommand{\prevalence}{\eta}

\DeclareMathOperator*{\argmin}{arg\,min}

\newcommand*{\intersect}{\cap}

\newcommand{\field}[1]{\mathbb{#1}}
\newcommand{\Reals}{\field{R}}

\title{Balancing central and marginal rejection when combining independent significance tests}
\author[1]{Chris Salahub}
\author[1]{Wayne Oldford}
\affil[1]{{\footnotesize Department of Statistics and Actuarial Science, University of Waterloo}}
\affil[]{{\footnotesize \texttt{\{chris.salahub, rwoldford\}@uwaterloo.ca}}}

\begin{document}

\maketitle

\begin{abstract}
  A common approach to evaluating the significance of a collection of $p$-values combines them with a pooling function, in particular when the original data are not available. These pooled $p$-values convert a sample of $p$-values into a single number which behaves like a univariate $p$-value. To clarify discussion of these functions, a telescoping series of alternative hypotheses are introduced that communicate the strength and prevalence of non-null evidence in the $p$-values before general pooling formulae are discussed. A pattern noticed in the UMP pooled $p$-value for a particular alternative motivates the definition and discussion of central and marginal rejection levels at $\alpha$. It is proven that central rejection is always greater than or equal to marginal rejection, motivating a quotient to measure the balance between the two for pooled $p$-values. A combining function based on the $\chi^2_{\kappa}$ quantile transformation is proposed to control this quotient and shown to be robust to mis-specified parameters relative to the UMP. Different powers for different parameter settings motivate a map of plausible alternatives based on where this pooled $p$-value is minimized.
\end{abstract}

\section{Introduction} \label{c:pooledPvals:intro}

When presented with a collection of $p$-values, a natural question is whether they constitute evidence as a whole against the null hypothesis that there are no significant results. The multiple testing problem arises because answering this question requires different analysis than univariate $p$-values. A univariate threshold applied to all $p$-values, for example, will no longer control the type I error at the level of the threshold. A common approach to control the type I error (often called the family-wise error rate in this context) is to use a function to combine the $p$-values into a single value which behaves like a univariate $p$-value.

Explicitly, consider a collection of $M$ independent test statistics $\ve{t} = \tr{(t_1, \dots, t_M)}$ having $p$-values $\ve{p} = \tr{(p_1, \dots, p_M)}$ for the null hypotheses $H_{01}$, $H_{02}$, $\dots$, $H_{0M}$ -- for example, $\chi^2$ tests for the association of $M$ individual genes with the presence of a disease where each $H_{0i}$ asserts no association.
Assessing the overall significance of $\ve{p}$ while controlling the family-wise error rate (FWER) at the outset of analysis is common practice in meta-analysis and big data applications (\citealp{heardrubin2018choosing, wilson2019harmonic}).
The FWER is the probability of rejecting one or more of $H_{01}, \dots, H_{0M}$ when all are true, equivalent to the type I error of the joint hypothesis
\begin{equation*}
  H_0 = \intersect_{i = 1}^M H_{0i}.
\end{equation*}
To emphasize the null distributions, $p_i \sim U = Unif(0,1)$ for all $i \in \{1, \dots, M\}$, this is often written
\begin{equation*}
  H_0 : p_1, p_2, \dots, p_M \overset{\mathrm{iid}}{\sim} U.
\end{equation*}

To test $H_0$, a statistic $l(\ve{p}): [0,1]^M \mapsto \Reals$ of the $p$-values with a distribution that is known or easily simulated under $H_0$ can be computed. If $l(\ve{p})$ has cumulative distribution function (CDF) $F_l(l)$ under $H_0$, then $l(\ve{p})$ admits $\genpool(\ve{p}) = 1 - F_l(l(\ve{p})) \sim Unif(0,1)$ such that rejecting $H_0$ when $\genpool(\ve{p}) \leq \alpha$ controls the FWER at level $\alpha$.\footnote{Note that the use of the CDF in $\genpool(\ve{p}) = 1 - F_l(l(\ve{p}))$ implies that $\genpool(\ve{p})$ is identical for any statistic that is a monotonic transformation of $l(\ve{p})$.} $\genpool(\ve{p})$ therefore summarizes the evidence against $H_0$ in a statistic which behaves like a univariate $p$-value: its magnitude is inversely related to its significance.

If we want $\genpool(\ve{p})$ to additionally have convex acceptance regions like a univariate $p$-values, it should be continuous in each argument and monotonically non-decreasing, i.e. $\genpool(p_1, \dots, p_M) \leq \genpool(p_1^{*}, \dots, p_M^{*}) \leftrightarrow p_1 \leq p_1^{*}, \dots, p_M \leq p_M^{*}$. Functions failing these can behave counter-intuitively, as they may accept $H_0$ for small $p_i$ only to reject as $p_i$ increases for some margin $i$. Finally, if there is no reason to favour any margin, $\genpool$ should be symmetric in $\ve{p}$. The term evidential statistic refers to $\genpool(\ve{p})$ meeting these criteria generally (\citealp{goutisetal1996assessing}), and when testing $H_0$ they are called pooled $p$-values. There is no lack of pooled $p$-value proposals, including the statistics of \cite{tippett1931methods}, \cite{fisher1932statistical}, \cite{pearson1933method}, \cite{stoufferetal1949american}, \cite{mudholkar1977logit}, \cite{heardrubin2018choosing}, and \cite{cinarviechtbauer2022poolr}.

As all of these methods have convex acceptance regions and control the FWER at $\alpha$ under the rule $\genpool(\ve{p}) \leq \alpha$, statistical power against alternative hypotheses is often used to distinguish them. Ideally, one among them would be uniformly most powerful (UMP) against a very broad alternative but this is not possible because of the generality of $H_0$. Indeed, \cite{birnbaum1954combining} proves that if all $f_i$ are strictly non-increasing so that $p_i \sim f_i$ is biased to small values when $H_{0i}$ is false, then there is no UMP test against the negation of $H_0$,
\begin{equation*}
  H_{1} = \neg H_0 : p_1 \sim f_1, p_2 \sim f_2, \dots, p_M \sim f_M
\end{equation*}
where $f_i \neq U$ for at least one $i \in \{1, \dots, M\}$. As the simulation studies in \cite{westberg1985combining}, \cite{loughin2004systematic}, and \cite{kocak2017meta} readily demonstrate, the number of false $H_{0i}$ and the non-null distributions $f_i$ together specify the unique most powerful test. For the particular case of testing $H_0$ against $H_1$ with $f_1 = f_2 = \dots = f_M = Beta(a, b)$ for $a \in (0,1]$ and $b \in [1, \infty)$, the Neyman-Pearson lemma proves that the pooled $p$-value $\hrpool{\ve{p}}{w}$ induced by the statistic
\begin{equation} \label{MC:eq:hrumptest}
   l_{HR}(\ve{p}; w) = w \sum_{i = 1}^M \ln p_i - (1 - w) \sum_{i = 1}^M \ln ( 1 - p_i )
\end{equation}
with $w = (1 - a)/(b - a) \in [0,1]$ is uniformly most powerful (UMP) (\citealp{heardrubin2018choosing}).

Though $\hrpool{\ve{p}}{(1 - a)/(b - a)}$ is UMP against $H_1$ for $f_1 = \dots = f_M = Beta(a,b)$, it is rarely assumed that $f_1 = \dots = f_M$ in the pairwise search for variables. Rather, some of these are assumed to be non-uniform while others are assumed null. A discussion of this setting therefore requires measures of the prevalence and strength of evidence against $H_0$, captured by a series of telescoping alternative hypotheses that bridge the gap between $H_1$ and the setting where $\hrpool{\ve{p}}{w}$ is UMP.

Section \ref{sec:strAndPrev} introduces the measures of strength and prevalence used in this paper, as these are required to understand central and marginal rejection later. Prevalence is measured by $\prevalence$, the proportion of non-null tests, while strength is measured by the Kullback-Leibler divergence. Assuming that non-null tests come from a restricted beta family, the power of a UMP method for particular beta distributions is investigated for different values of the prevalence and strength in Section \ref{sec:power}. A pattern of high power for either strong evidence in a few tests or weak evidence in many tests is noticed and developed into a framework for choosing pooled $p$-values in Section \ref{rejlev}. Following the necessary definitions to develop this framework, including the concepts of central and marginal rejection, it is proved that the tendency to reject concentrated evidence is always less than diffuse evidence, allowing for the definition of a coefficient of the preference of a pooling function to diffuse evidence.

Section \ref{sec:chi} proposes a pooling function based on the $\chi^2$ quantile transformation which controls this preference through its degrees of freedom. It is proven that large degrees of freedom give a pooling function which prefers diffuse evidence while small degrees of freedom prefer concentrated evidence. In a simulation study, this proposal is shown to nearly match the UMP when correctly specified and is more robust to errors in specification. These conclusions are extended in Section \ref{sec:chiIdentifying} where a sweep of parameter values is used to identify the the most powerful choice for a given sample and suggest a region of most plausible alternative hypotheses within the framework in light of it.

\section{Measuring the strength and prevalence of evidence} \label{sec:strAndPrev}

When proving that no UMP exists for the general hypothesis $H_1 = \neg H_0$, \cite{birnbaum1954combining} provides a couple of two-dimensional examples. Though these are demonstrative, they are not instructive for the discussion of tests generally. To the same conclusions, each of \cite{westberg1985combining}, \cite{loughin2004systematic}, and \cite{kocak2017meta} simulate a variety of populations with differing proportions of $\ve{p}$ generated under $U$ or some alternative distribution. We begin by defining a telescoping series of alternative hypotheses which capture the settings explored in these empirical investigations.

\subsection{Telescoping alternatives}

Starting at $H_1$, assume $H_{0i}$ is false only for $i \in J \subset \{1, \dots, M\}$ and quantify the proportion of non-null hypotheses by $\prevalence = \abs{J}/M$. This implies an alternative hypothesis
\begin{equation*}
  H_{2} : p_i \sim \begin{cases} f_i \neq U & \text{ if } i \in J, \\
    U & \text{ if } i \notin J.
  \end{cases}
\end{equation*}
As no distinctions between $H_{01}, \dots, H_{0M}$ are made in $H_0$, $\prevalence$ captures the prevalence of evidence against $H_0$ without loss of generality. If it is additionally assumed that all $i \in J$ have the same alternative distribution $f \neq U$, this gives the alternative hypothesis
\begin{equation*}
  H_{3} : p_i \sim \begin{cases} f & \text{ if } i \in J, \\
    U & \text{ if } i \notin J.
  \end{cases}
\end{equation*}
Finally, in the particular case where $\prevalence = 1$, $\abs{J} = M$ and 
\begin{equation*}
  H_4 : p_1, p_2, \dots, p_M \overset{\mathrm{iid}}{\sim} f \neq U
\end{equation*}
is obtained. Though restricted compared to $H_1$, this alternative makes sense for meta-analysis or repeated experiments, where we could assume all $p_i$ are independently and identically distributed when $H_{0}$ is false.

$H_4$ was distinguished from $H_1$ as early as \cite{birnbaum1954combining} (there called $H_A$ and $H_B$ respectively), but no exploration of intermediate possibilities was considered. All of \cite{westberg1985combining}, \cite{loughin2004systematic}, and \cite{kocak2017meta} explore factorial combinations of $\prevalence$ and $f$, and so use instances of $H_3$ for their investigations. When testing $H_4$ against $H_0$, \cite{heardrubin2018choosing} prove $\hrpool{\ve{p}}{w}$ is the UMP pooled $p$-value if $f$ is from a constrained beta family. By stating clearly $H_1 \supset H_2 \supset H_3 \supset H_4$, a framework for alternative hypothesis is created that contextualizes and relates these previous results. Additionally, the parameter $\prevalence$ under $H_3$ naturally measures the prevalence in $\ve{p}$ of evidence against $H_0$.

\subsection{Measuring the strength of evidence} \label{strPre:KLdiv}

Intuitively, if $f$ has a density highly concentrated near zero then it provides ``strong'' evidence against $H_0$. This is because $p$-values generated by $f$ will tend to be smaller than those following $U$ and therefore will be rejected more frequently for any $\alpha$. Any measure of the strength of evidence in $f$ should therefore increase as the magnitude of $f$ for small values increases.

This relatively simple criterion is challenging to apply to $H_2$. Every $f_i$ for $i \in J$ may be distinct and a single value characterizing their multiple, potentially very different, departures from $U$ introduces ambiguity. Taking a mean of measures, for example, conflates different instances of $H_2$. If $J = \{1,2\}$, the mean strength of evidence cannot distinguish strong evidence in $f_1$ with weak evidence in $f_2$ from moderate evidence in both. This difficulty is avoided if all $f_i$ are equal, i.e. if $H_3$ is chosen as the alternative hypothesis. Indeed, this choice is common in previous empirical investigations.

\cite{westberg1985combining} generates $\ve{p}$ by testing the difference in means of two simulated normal samples, and measures the strength of evidence by the true difference in means between the generative distributions. This is reasonable, but limits us to tests comparing population parameters and requires assumptions on $\ve{t}$ (the tests generating $\ve{p}$). Considering $f$ directly, \cite{loughin2004systematic} takes $p_i \overset{\mathrm{iid}}{\sim} Beta(a,b)$ for all $i \in J$, restricts $a = 1 \leq b$ so that $f$ is non-increasing, and measures the strength of evidence with one minus the median of $f$: $1 - 0.5^{1/b}$.\footnote{\cite{kocak2017meta} also uses $p_i \sim Beta(a,b)$ but takes the broader $0 < a \leq b$ and does not attempt to measure the strength of $f$'s departure from $U$.} This measure of strength is limited to $Beta(1,b)$, though it does achieve the intuitive ordering desired. A more general measure that applies to any $f$ is the Kullback-Leibler (KL) divergence, given by
$$D(p, q) = \int_{\mathcal{X}} p(x) \ln \left ( \frac{p(x)}{q(x)} \right ) dx$$
for $q(x)$ to density $p(x)$ with mutual support on $\mathcal{X}$.

Widespread application of the KL divergence in information theory and machine learning aside, one interpretation of this measure suits pooled hypothesis testing nicely. \cite{joyce2011kullback} describes the KL divergence as the extra information encoded in $q(x)$ when expecting $p(x)$. The explicit assumption underlying the pooled test of $H_0$ is that $f_i = U$ for all $i \in \{1, \dots, M\}$, which gives a natural expected density $q(x) = U(x)$. Furthermore, this density is, in some sense, minimally informative: no region of $[0,1]$ is distinguished from any other by $U$. Any additional information which discriminates particular regions of $[0,1]$, in particular values near $0$, will help inform rejection.

\subsection{Choosing a family for the alternative distribution} \label{familyForAlt}

The beta family of distributions is appealing as a model for the alternative distribution $f$ under $H_3$ for two main reasons. First, it has the same support as $U$ without the need for adjustment. Second, a wide variety of different density shapes can be achieved by changing its two parameters and it has a non-increasing density whenever $a \leq 1 \leq b$. This latter quality makes it ideal to model alternative $p$-value distributions under the assumption that $p_i$ is biased to small values when $H_0$ is false. These features are likely why it is commonly used in the literature (e.g. \cite{loughin2004systematic}, \cite{kocak2017meta}). More significantly, the Neyman-Pearson lemma proves that $\hrpool{\ve{p}}{w}$ is UMP for $p_1, \dots, p_M \overset{\mathrm{iid}}{\sim} Beta(a, 1/w + a(1 - 1/w))$ when $0 \leq w,a \leq 1$, and so a best-case benchmark for power exists to compare to other tests (\citealp{heardrubin2018choosing}). Were another distribution chosen, this important reference point would be absent.

When $f = Beta(a,b)$ the KL divergence also has a relatively simple expression.
Letting $u(x) = 1$ be the uniform density and $f(x) = \frac{\Gamma(a + b)}{\Gamma(a) \Gamma(b)} x^{a - 1}(1 - x)^{b - 1}$ be the beta density with parameters $a$ and $b$, $D(u, f)$ is given by
$$D(u, f) = - \int_{0}^1  \ln f(x) dx = a + b + \ln \left ( \frac{\Gamma(a) \Gamma(b)}{\Gamma(a + b)} \right ) - 2.$$
For the case of the beta densities where $\hrpool{\ve{p}}{w}$ is UMP, i.e. $a \leq 1 \leq b$, we can express this in terms of $a$ and $w = (1 - a)/(b - a)$:
\begin{equation} \label{eq:klbeta}
  D(u,f) := D(a, w) = 2a + \frac{1 - a}{w} + \ln \left ( \frac{\Gamma(a) \Gamma \left ( \frac{1}{w} + a \left [ 1 - \frac{1}{w} \right ] \right )}{\Gamma \left (2a + \frac{1 - a}{w} \right )} \right ) - 2.
\end{equation}
This is a less convenient expression, but provides a direct link between the strength of evidence $D(a, w)$ and the UMP test against $H_4$ by the shared parameter $w$. Interestingly, though the strength depends on both $a$ and $w$, the UMP test only depends on the latter.

To visualize the strength of evidence provided by different beta distributions, shaded inset densities for different choices of $w$ and the KL divergence $D(a,w)$ are displayed in Figure \ref{fig:klDivs} placed at $\ln(w)$, $\ln D(a, w)$.\footnote{The densities of these inset plots were determined for each $w$, $D(a,w)$ pair by finding the corresponding $a$ value numerically using Equation (\ref{eq:klbeta}). This is the cause of the irregular plots in the upper right corner: when $w$ is large enough the required $a$ to obtain a set $D(a,w)$ is too small to be represented as a floating point double alongside $w \approx 1$. This is inconvenient, but these cases correspond to densities that are effectively degenerate at zero in any case.} When $a = w = 1$, $f = u$ so $D(a, w) = 0$. Decreasing either of $a$ or $w$ from 1 causes a larger magnitude of $f$ near zero, with the limiting case $a = w = 0$ corresponding to a degenerate distribution at zero. This general trend in shapes is seen in Figure \ref{fig:klDivs}, decreasing $w$ or increasing $D(a, w)$ increases the concentration of $f$ near zero.

\begin{figure}[!h]
  \begin{center}
    \includegraphics{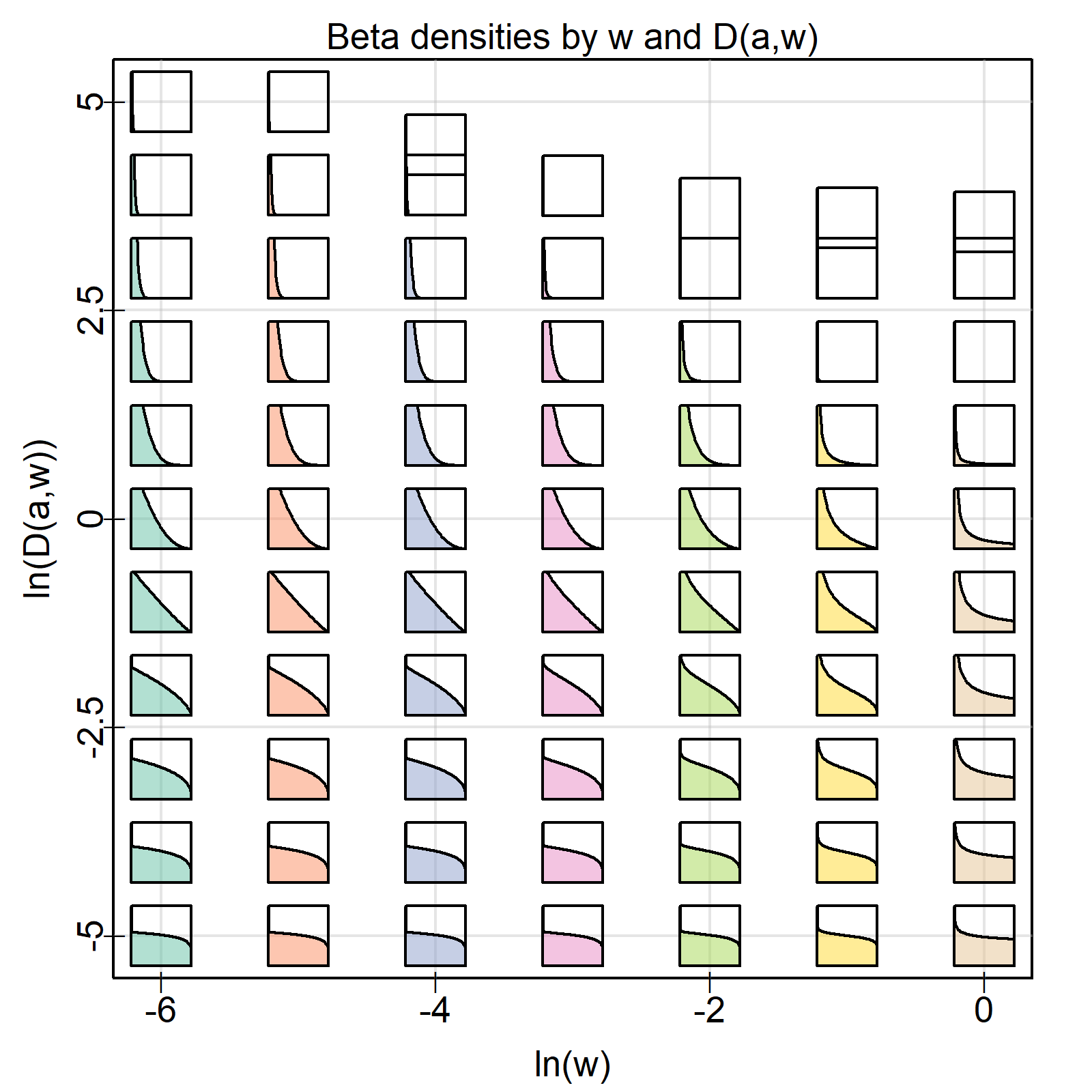}
  \caption{Densities and log KL divergences of $Beta(a, 1/w + a(1 - 1/w))$ from $U$ by $w$. Insets display the densities over $[0,1]$ horizontally and $[0,2]$ vertically and are centred at the $\ln(w)$, $\ln D(a,w)$ coordinates corresponding to the density. These densities range from nearly vertical at 0 when $D(a,w) \approx e^5$ to nearly uniform when $D(a,w) \approx e^{-5}$.}
  \label{fig:klDivs}
  \end{center}
\end{figure}

This suggests a limitation in the KL divergence. Though the ordering of beta densities by $D(a,w)$ generally conforms to the intuitive rule -- larger divergences correspond to inset densities with greater magnitude near zero -- the parameter $w$ is still relevant to the shape. The KL divergence does not distinguish between departures from uniform near 1 and near 0, despite their relevance for rejection when, for example, rejecting the null hypotheses of $p$-values below a threshold. This is particularly obvious in the final row of inset plots in Figure \ref{fig:klDivs}. When $\ln(w) = -6$, the density for $\ln D(a,w) = -5$ is mostly flat, with a slight increase in density near zero and a large decrease near one. When $\ln(w) = 0$, however, a much larger spike in the density near zero is present. 

Nonetheless, the ordering on beta densities imposed by the KL divergence is still very informative. It classifies, generally, which densities are biased to small values.  Therefore, $D(a,w)$ provides a convenient measure of the strength of evidence contained in $f$ under the alternative hypothesis, and has computationally convenient form for the case of interest where $f = Beta(a,b)$.

\section{Pooled $p$-values} \label{sec:pooled}

Having explored the possible alternatives to $H_0$ and some specific instances of these alternatives, we can now focus on the pooled $p$-values meant to test $H_0$ against these alternatives.
Recall that any pooled $p$-value $\genpool(\ve{p})$ is derived from the null distribution of a corresponding statistic $l(\ve{p})$. The statistics underlying pooled $p$-values are of two basic kinds, based either on the $k^{th}$ order statistic $p_{(k)}$ of $\ve{p}$ (e.g. \cite{tippett1931methods} and \cite{wilkinson1951statistical}) or on transformations of each $p_i$ using some quantile function $F^{-1}(p)$ (e.g. \cite{fisher1932statistical}, \cite{pearson1933method}, \cite{stoufferetal1949american}, \cite{lancaster1961combination}, \cite{edgington1972additive}, \cite{mudholkar1977logit}, \cite{heardrubin2018choosing}, \cite{wilson2019harmonic}, and \cite{cinarviechtbauer2022poolr}). The former case takes the general form
\begin{equation} \label{MC:eq:orderpool}
  \ordpool{\ve{p}}{k} = \sum_{l = k}^M {M \choose l}p_{(k)}^l (1 - p_{(k)})^{M-l},
\end{equation}
and the latter
\begin{equation} \label{MC:eq:quantilePval}
  \genpool(\ve{p}) = 1 - F_M \left ( \sum_{i = 1}^M c_i F^{-1} ( 1 - p_i ) \right )
\end{equation}
where $c_1, \dots, c_M \in \Reals$ are known constants, typically $c_1 = \dots c_M = 1$. Equation (\ref{MC:eq:orderpool}) gives the pooled $p$-value based on $l_{Ord}(\ve{p};k) = p_{(k)}$, as with $\tippool(\ve{p}) = \ordpool{\ve{p}}{1} = 1 - (1 - p_{(1)})^M$ (\citealp{tippett1931methods}), while different choices of $F$ and $F_M$ in Equation (\ref{MC:eq:quantilePval}) give the pooled $p$-value based on $l(\ve{p}) = \sum_{i = 1}^M F^{-1}(1 - p_i)$.

Obvious choices include the normal and gamma families, as these are closed under addition. For example, letting $\Phi$ be the $N(0,1)$ CDF and choosing $F(x) = \Phi(x)$ and $F_M(x) = \Phi(x/\sqrt{M})$ gives
\begin{equation} \label{MC:eq:stoufferpool}
  \stopool(\ve{p}) = 1 - \Phi \left ( \sum_{i = 1}^M \Phi^{-1}(1 - p_i)/\sqrt{M} \right)
\end{equation}
based on $l_{Sto}(\ve{p}) = \sum_{i = 1}^M \Phi(1 - p_i)$ from \cite{stoufferetal1949american}. Letting $G_{k, \theta}(x)$ be the CDF of the gamma distribution with shape parameter $k$ and scale parameter $\theta$, taking $F(x) = G_{k,\theta}(x)$ and $F_M(x) = G_{Mk,\theta}(x)$ gives
\begin{equation} \label{MC:eq:gammaPool}
  \gampool(\ve{p}) = 1 - G_{Mk, \theta} \left ( \sum_{i = 1}^M G_{k, \theta}^{-1}(1 - p_i) \right )
\end{equation}
based on $l_{gam}(\ve{p}) = \sum_{i = 1}^M G_{k, \theta}^{-1}(1 - p_i)$. $\gampool$ requires the choice of a parameter $\theta$; choosing $\theta = 1$ gives the gamma method from \cite{zaykinetal2007combining} while $k = 1$ and $\theta = 2$ gives Fisher's method from \cite{fisher1932statistical}.\footnote{Should the $p$-values be weighted for some reason, the gamma distribution also allows more stable weighting than the constants $c_1, \dots, c_M$ in Equation (\ref{MC:eq:quantilePval}) by analogously giving each $p_i$ an individual shape parameter $k_i$ (or equivalently $\chi^2$ degrees of freedom $\kappa_i$) (\citealp{lancaster1961combination}).}

R. A. Fisher's method deserves some additional consideration alongside an analogous proposal from Karl Pearson around the same time\footnote{\cite{owen2009karl} notes that Pearson's proposal is actually slightly different than has been credited in the literature following \cite{birnbaum1954combining}. Though the use of $\sum_{i = 1}^M \ln(p_i)$ was suggested by Karl Pearson, it was in the context of comparing the value to $\sum_{i = 1}^M \ln(1 - p_i)$ and taking the minimum of the two. \cite{owen2009karl} develops this idea into a series of pooling functions that perform best for concordant or discordant effect estimates in a regression setting.}. Both of $l_{Fis}(\ve{p}) = -2 \sum_{i = 1}^M \ln p_i$ (\citealp{fisher1932statistical}) and $l_{Pea}(\ve{p}) = -2 \sum_{i = 1}^M \ln (1 - p_i)$ (\citealp{pearson1933method}) were originally proposed only as computational tricks for the distribution of $\prod_{i = 1}^M p_i$ (\citealp{wallis1942compounding}), but are also quantile transformations based on the $\chi_{2}^2$ distribution. Let
\begin{equation} \label{MC:eq:chiF}
  F_{\chi}(x; \kappa) = \int_0^x \frac{1}{2^{\kappa/2} \Gamma(\kappa / 2)} t^{\kappa/2 - 1} e^{-t/2} dt,
\end{equation}
be the CDF of the $\chi^2_{\kappa}$ distribution, in particular $F_{\chi}(x; 2) = 1 - e^{-x/2}$.
Therefore, $F^{-1}_{\chi}(1 - p; 2) = -2 \ln p$ and so taking $F(x) = F_{\chi}(x; 2)$ and $F_M(x) = F_{\chi}(x; 2M)$ gives
$$\fispool(\ve{p}) = 1 - F_{\chi} \big ( l_{F}(\ve{p}); 2M \big ) = 1 - F_{\chi} \left ( -2 \sum_{i = 1}^M \ln p_i ; 2M \right ) = 1 - F_{\chi} \left ( \sum_{i = 1}^M F_{\chi}^{-1} (1 - p_i; 2); 2M \right )$$ consistent with Equation (\ref{MC:eq:quantilePval}). In contrast, $l_{Pea}(\ve{p})$ uses lower tail probabilities by taking $F_{\chi}(p_i; 2)$, and so
$$\pearpool(\ve{p}) = F_{\chi} \left ( \sum_{i = 1}^M F_{\chi}^{-1} (p_i; 2); 2M \right )$$
departs from the general quantile transformation equation.

\section{Benchmarking the most powerful test} \label{sec:power}

Alone, $\fispool(\ve{p})$ is preferred to $\pearpool(\ve{p})$, as \cite{birnbaum1954combining} found $\pearpool(\ve{p})$ inadmissible for the alternative hypothesis $H_1$ if the $t_i$ independently follow particular distributions in the exponential family. $\fispool(\ve{p})$, in contrast, was admissible in this setting and is optimal in some sense for others (\citealp{littellfolks1971asymptotic, koziolperlman1978combining}).
Together, the statistics for these two pooled $p$-values are combined in $l_{HR}(\ve{p}; w)  = -\frac{w}{2} l_{Fis}(\ve{p})  + \frac{1 - w}{2} l_{Pea}(\ve{p})$, the UMP statistic under $H_4$ when $f = Beta(a,  1/w + a(1 - 1/w))$ (\citealp{heardrubin2018choosing}). Intuitively, then, $\hrpool{\ve{p}}{w}$ is a test based on a linear combination of the lower and upper tail probabilities of $\ve{p}$ transformed to $\chi^2_2$ quantiles. When $w = 1$ it considers the upper tail alone and when $w = 0$ the lower tail alone. Note that $l_{HR}(\ve{p}; w)$, the statistic, and $\hrpool{\ve{p}}{w}$, the unique corresponding pooled $p$-value, will be used interchangeably thoughout this paper.

Unfortunately, this imbues $\hrpool{\ve{p}}{w}$ with some practical shortcomings. While both $l_{Fis}(\ve{p})$ and $l_{Pea}(\ve{p})$ are $\chi^2_{2M}$ distributed under $H_0$, they are not independent and so their combination in $l_{HR}(\ve{p}; w)$ does not have a closed-form distribution. Approximation as in \cite{mudholkar1977logit} or simulation must be used to determine the $\alpha$ quantiles or visualize their distribution. This means, for example, the kernel density estimates of $l_{HR}(\ve{p}; w)$ by $w$ in Figure \ref{fig:lwdens} required the generation of 100,000 independent simulated samples of the case $p_1, \dots, p_{10} \overset{\mathrm{iid}}{\sim} U$.
Further, $H_4$ is the least general of the telescoping alternative hypotheses $H_1 \supset H_2 \supset H_3 \supset H_4$ and is a less natural choice than $H_3$ if only a subset of tests are thought to be significant. Empirical and theoretical investigations show the most powerful test depends on $\prevalence$ and $f$ under $H_3$, so $\hrpool{\ve{p}}{w}$ may be less exceptional under this more general hypothesis.
Finally, $\hrpool{\ve{p}}{w}$ is only UMP if $w$ is known, which is seldom true in practice.

\begin{figure}[!h]
  \begin{center}
    \includegraphics{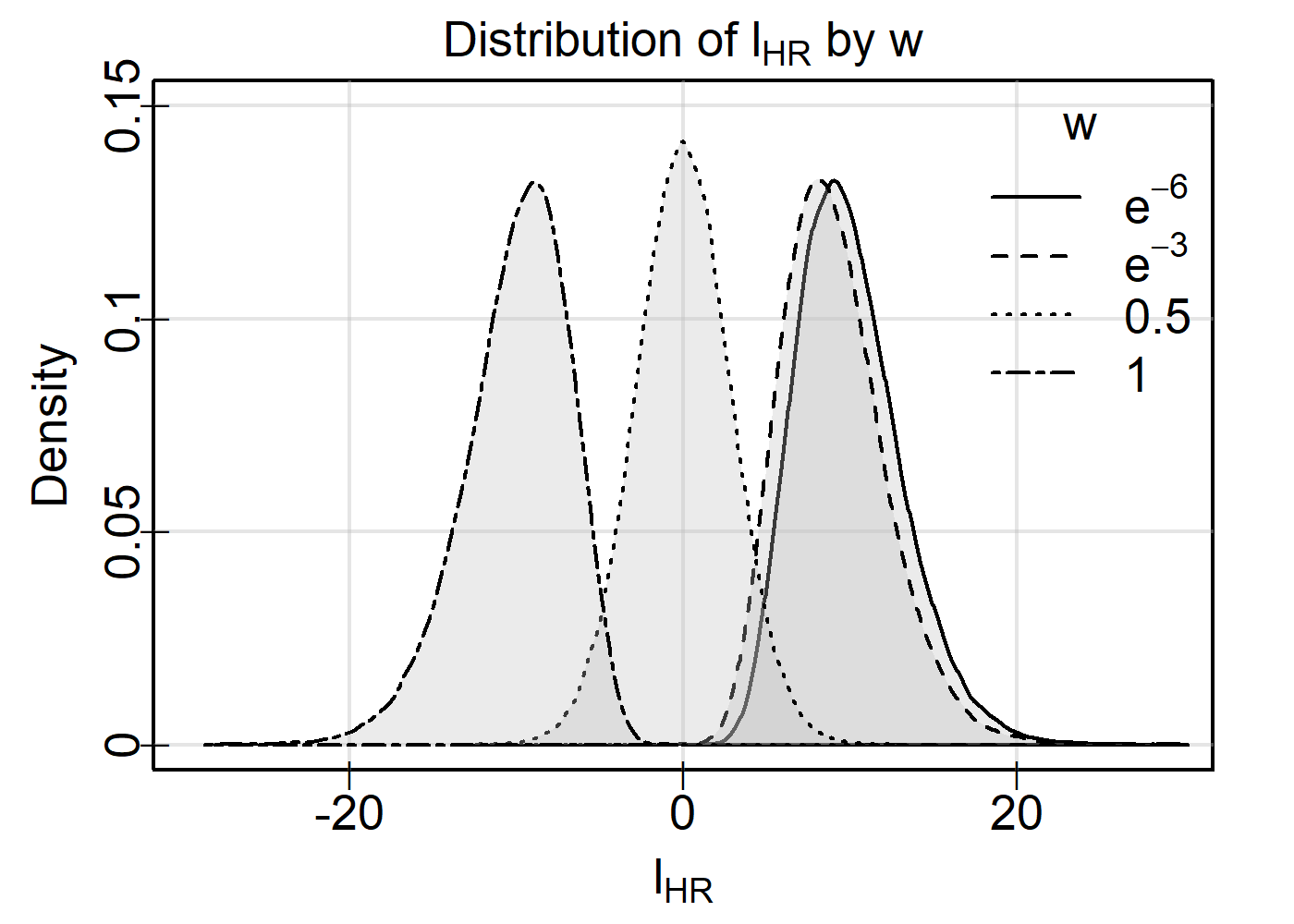}
    \caption{Densities of $l_{HR}$ by $w$ when $M = 10$. Solid lines indicate $w =  e^{-6}$, dashed lines $w = e^{-3}$, dotted lines $w = 1/2$, and dot-dashed lines $w = 1$. Note how $w = 1$ and $w = e^{-6}$ are nearly mirrored distributions skewed away from zero and $w = 1/2$ is symmetric at zero.}
  \label{fig:lwdens}
  \end{center}
\end{figure}

These practical difficulties manifest in two possible errors, assuming $H_4$ when $H_3$ is true and choosing $w$ when the true parameter is $\omega$, and four cases of mis-specification depending on which is present. The power of $\hrpool{\ve{p}}{w}$ under all four cases was investigated by a simulation study at level $\alpha = 0.05$. For both of $H_3$ and $H_4$ and a range of mis-specified $w$, $\hrpool{\ve{p}}{w}$ was applied to factorial combinations of $D(a,w)$, $w$, and $M$ covering their respective ranges. $D(a,w)$ was chosen on the log scale ranging from $-5$ to 5 at 0.5 increments, $w$ was chosen on the log scale at values $-6, -5, \dots, 0$, and the values of $M$ were 2, 5, 10, and 20.

For each of the parameter settings, $l_{HR}(\ve{p}; w)$'s 0.95 quantiles under $H_0$ given $w$ and $M$ are simulated by generating 100,000 independent samples $\ve{p}_i = p_{i1}, \dots, p_{iM} \overset{\mathrm{iid}}{\sim} U$, computing $l_{HR}(\ve{p}_i; w) = l_{HRi}$, and taking the 0.95 quantile of the sequence $l_{HR1}, \dots, l_{HR100,000}$ as the 0.95 quantile of $l_{HR}(\ve{p}; w)$ under $H_0$. Note that the value of $a$ and the case do not impact this simulation, and so these quantiles are used across all $a$ values under both $H_3$ and $H_4$.

Next a Monte Carlo estimate of the probability of rejecting $H_0$ using $l_{HR}(\ve{p}; w)$ (i.e. the power of $l_{HR}(\ve{p}; w)$) is generated for each case. The details of this estimate depend on whether $H_3$ or $H_4$ was used to generate the data and whether $w$ was known or not when choosing $l_w$. To begin, consider the benchmark case when $w$ is known and the data are generated according to $H_4$.

\subsection{Case 1: correct hypothesis and $w$} \label{subsec:h4w}

If the data are generated under $H_4$ and $w$ is chosen correctly, $\hrpool{\ve{p}}{w}$ is UMP and so provides the greatest power of any test. In this case, the probability of rejection for a given $a, w$, and $M$ setting is estimated by generating 10,000 independent samples $\ve{p}_i = p_{i1}, \dots, p_{iM} \overset{\mathrm{iid}}{\sim} Beta(a, 1/w + a(1 - 1/w))$. $l_{HR}(\ve{p}_i; w)$ is computed for each sample and compared to the simulated 0.95 quantile of $l_{HR}(\ve{p}; w)$ under $H_0$\footnote{This corresponds to testing whether $\hrpool{\ve{p}_i}{w} \leq 0.05$.}. If the value is larger than the quantile $H_0$ is correctly rejected and if it is smaller the test incorrectly fails to reject $H_0$. Power is estimated as the proportion of the 10,000 generated samples which correctly lead to the rejection of $H_0$, giving a worst-case standard error less than 0.005 based on the binomial distribution.

\begin{figure}[!h]
  \begin{center}
    \begin{tabular}{cc}
      \includegraphics{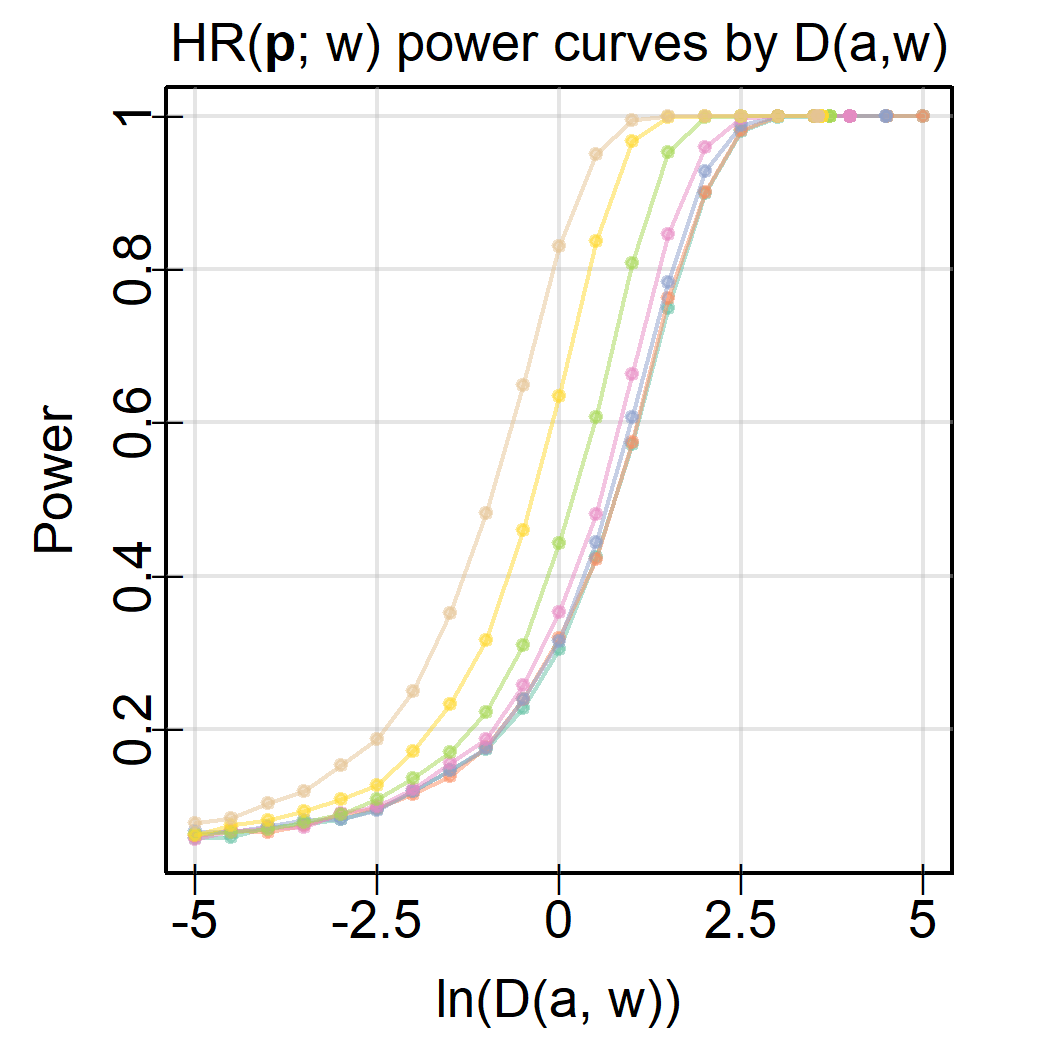} & \includegraphics{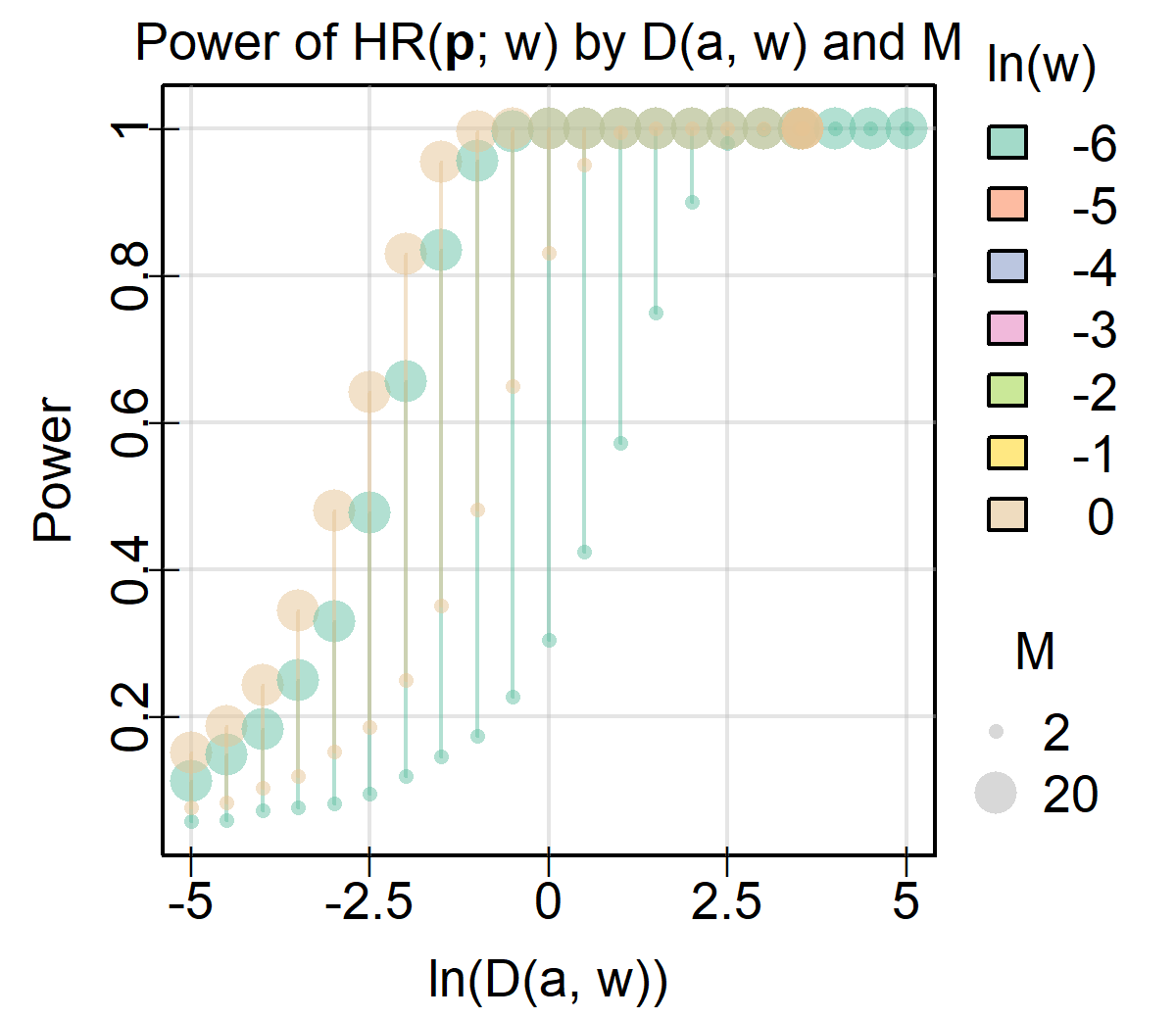} \\
      {\footnotesize (a)} & {\footnotesize (b)} \\
    \end{tabular}
    \caption{Power of $\hrpool{\ve{p}}{w}$ by the KL divergence coloured by $w$ and scaled by $M$ displayed using (a) power curves when $M = 2$ joined by $w$ and (b) lines from $M = 2$ to $M = 20$ for $w = e^{-6}$ and $0$. Increasing either $M$ or the KL divergence increases the power and the greatest rate of change in both occurs when the divergence is in the interval $(e^{-2}, e^2)$.}
  \label{fig:powersByMW}
  \end{center}
\end{figure}

This procedure is applied to all settings of $M$, $w$, and $D(a,w)$ outlined in Section \ref{sec:power}, corresponding to the beta densities of Figure \ref{fig:klDivs}. The beta parameter $a$ was determined for a given $w$ and $D(a,w) = D$ by finding the root of $f(x) = D(x,w) - D$, while the parameter $b$ is given by $1/w + a(1 - 1/w)$. This results in an imbalance in the settings, as $w > e^{-4}$ require $a$ less than the typical floating point minimum value to achieve $D(a,w) = 5$. The impact on the coverage of $D(a,w)$ for each $w$ choice as a result of this, however, was slight as shown by the small gap in plots in Figure \ref{fig:klDivs}.

Figure \ref{fig:powersByMW}(b) shows a scatterplot of the power of $\hrpool{\ve{p}}{w}$ by $D(a, w)$ for every setting when $M = 2$ and $M = 20$, and Figure \ref{fig:powersByMW}(a) shows the power curves of $\hrpool{\ve{p}}{w}$ by $D(a, w)$ coloured by $w$. Generally, power increases in both $M$ (the number of $p$-values) and $D(a,w)$ (the KL divergence), which is unsurprising. Decreasing $D(a,w)$ for $f = Beta(a, 1/w + a(1 - 1/w))$ necessarily gives a density closer to $u$ which is therefore less likely to cause rejection, thus reducing the power. At a certain threshold on $D(a,w)$, $f \approx u$ and so the power will be $\sim \alpha$ for all $D(a, w)$ less than the threshold. Similarly, when $D(a,w)$ is large enough, rejection occurs almost certainly and so the power is constant at one. This suggests that most changes in the power of $\hrpool{\ve{p}}{w}$ occur for moderate levels of evidence; when the evidence is too weak or too strong, all pooled $p$-values will perform equally poorly or well. Regardless of $D(a,w)$, increasing the number of $p$-values makes any distributional differences between $f$ and $u$ more easily detectable, as the whole sample KL divergence is given by $M D(a,w)$. This is why the impact of $M$ on the power in Figure \ref{fig:powersByMW}(b) increases in $D(a,w)$.

An interesting feature of Figure \ref{fig:powersByMW}(a) is the ordering of the lines by decreasing $w$ for essentially every KL divergence. This pattern holds almost everywhere with the exception of several crossings of the lowest power curves. Referring to Figure \ref{fig:klDivs}, increasing $w$ for a given $D(a,w)$ increases the magnitude of the density near zero, thereby increasing the probability of extremely small $p$-values. This difference was noted in Section \ref{strPre:KLdiv}, and causes the expected increase in the power of the UMP.

\subsection{Case 2: correct hypothesis with mis-specified $w$} \label{subsec:h4omega}

Of course, the curves of Figure \ref{fig:powersByMW}(a) are not realistic. In practice, $D(a, w)$ is set by the data generating process and we lack the perfect knowledge of $w$ and $f$ needed to attain them. Suppose that $\ve{p}$ is generated under $H_4$ with $f = Beta(a, \beta)$ and $a \in [0, 1]$, $\beta \in [1, \infty)$ but that $\hrpool{\ve{p}}{w}$ is used instead of the correct $\hrpool{\ve{p}}{\omega}$ with $\omega = (1 - a)/(\beta - a)$. Though $\hrpool{\ve{p}}{w}$ is from the same family of tests, it is no longer UMP and so will not match the power achieved by $\hrpool{\ve{p}}{\omega}$.

The reduction of power from using $\hrpool{\ve{p}}{w}$ when the UMP is $\hrpool{\ve{p}}{\omega}$ for each $a$, $w$, $\omega$, and $M$ setting from Section \ref{subsec:h4w} was determined by generating 10,000 independent samples $p_{i1}, \dots, p_{iM} \overset{\mathrm{iid}}{\sim} Beta(a, 1/\omega + a(1 - 1/\omega))$ and computing $\hrpool{\ve{p}_i}{w}$. The proportion of samples rejected based on the simulated null 0.95 quantiles was recorded as the power of $\hrpool{\ve{p}}{w}$ when data were truly generated with the parameter $\omega$. This procedure was repeated for each of $w = e^{-6}, e^{-3}, 1/2,$ and $1$ for every parameter setting. Figure \ref{fig:MisPowers} displays the results when $M = 2$ for $\omega = e^{-6}$ and $1$. As one setting of $w$ matches $\omega$ in this case, the highest curve displays the power of the UMP.

\begin{figure}[!h]
  \begin{center}
    \includegraphics[scale = 1]{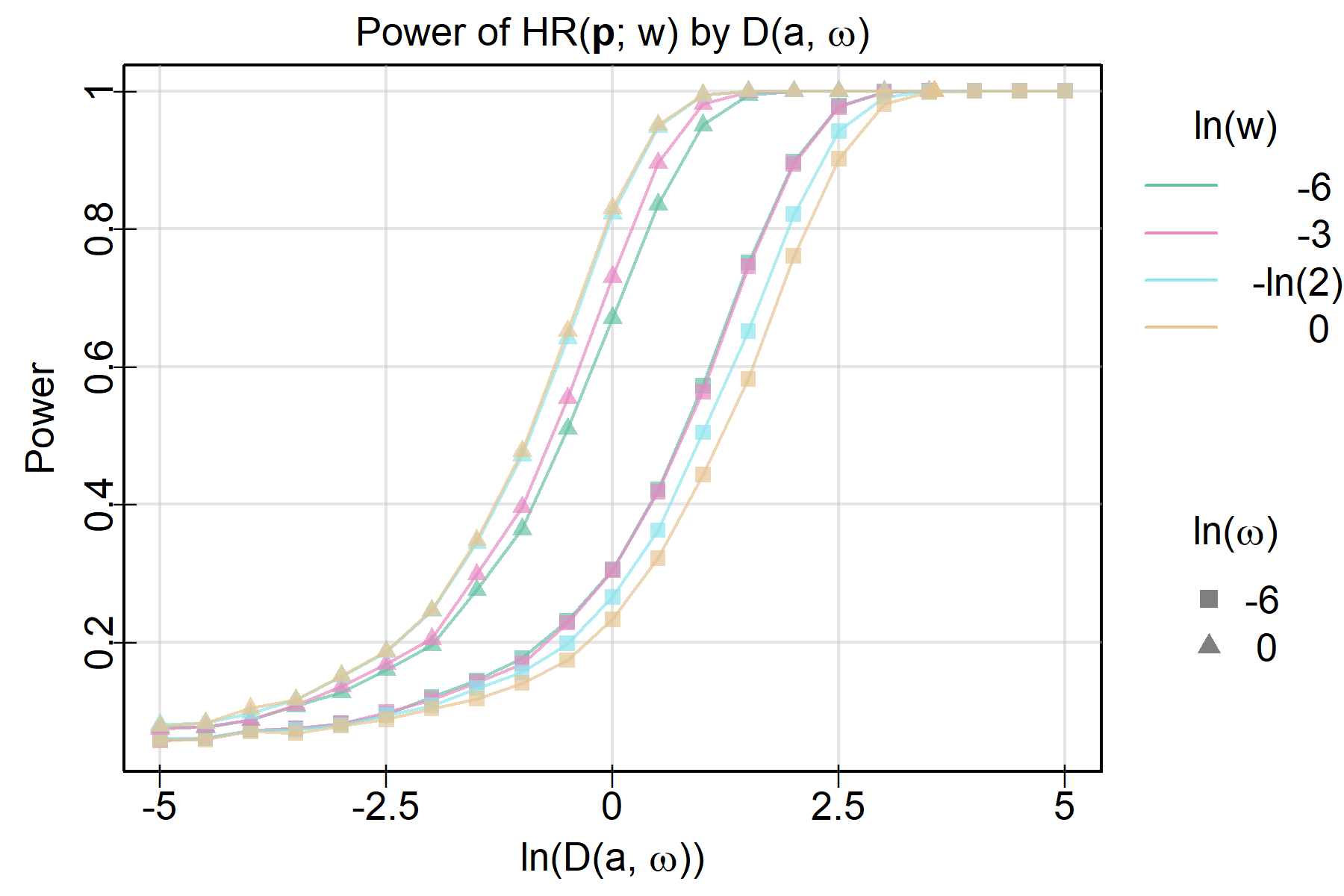}
    \caption{Power curves for $\hrpool{\ve{p}}{w}$ against $D(a, \omega)$ when $M = 2$ with colours giving the value of $w$ and points along the curves giving $\omega$. Mis-specification of $\omega$ has less impact on power than the non-null distribution $f$, but the greater the difference between $w$ and $\omega$, the greater the reduction in power.}
  \label{fig:MisPowers}
  \end{center}
\end{figure}

Two patterns stand out in this plot. First, it is clear that mis-specification impacts the power less than the distribution of $p$-values under $H_4$. Despite the incorrect value $w$ in $\hrpool{}{w}$, the mis-specified curves have the same shape in $D(a, \omega)$ as the UMP $\hrpool{}{\omega}$. In most cases, mis-specification results in only a slight decrease in power.

Second, larger mis-specifications lead to larger decreases in power. The curves for every $w$ are ordered by their distance from $\omega$ for both $\omega = 1$ and $e^{-6}$. When $\omega = 1$, for example, the lines are ordered so $w = 1$ has the greatest power followed closely by $w = 1/2$ and more distantly by $e^{-3}$ and $e^{-6}$. This pattern is reversed when $\omega = e^{-6}$. In both cases, mis-specification has the greatest impact on power for moderate $D(a, \omega)$ while mis-specification has scarcely any impact for large or small values of $D(a,w)$ where all powers converge to 1 or $\alpha$, respectively. It seems prudent, therefore, to choose a middling value such as $w = 1/2$ if $\omega$ is not known but $H_4$ is suspected to avoid the worst impacts of mis-specification when using $\hrpool{\ve{p}}{w}$

\subsection{Case 3: incorrect hypothesis, correctly specified $w$} \label{subsec:h3w}

Assuming all $p_i$ are non-null is not always appropriate, however. The investigations of the previous sections are a useful benchmark and exploration of the impact of mi-specification, but do not address the natural case of $H_3$ when a handful of significant variables are assumed to exist in a host of insignificant ones. We therefore repeat the benchmark experiment of Section \ref{subsec:h4w} under $H_3$ by generating 10,000 independent samples $\ve{p}_i$ of size $M = 10$ with the first $M\prevalence$ distributed according to $f = Beta(a, 1/w + a(1 - 1/w))$ and the latter $M(1 - \prevalence)$ distributed according to $U$ for $\prevalence \in \{0, 0.1, \dots, 1\}$ under each of the settings explored previously. As $\hrpool{\ve{p}}{w}$ is symmetric in all of its arguments, this gives no loss of generality. Figure \ref{fig:H3power} displays contours of the power surface, the proportion of correct rejections of $H_0$, as a function of $\prevalence$ and $D(a,w)$ facetted by $\ln(w)$.

\begin{figure}[!h]
  \begin{center}
    \includegraphics{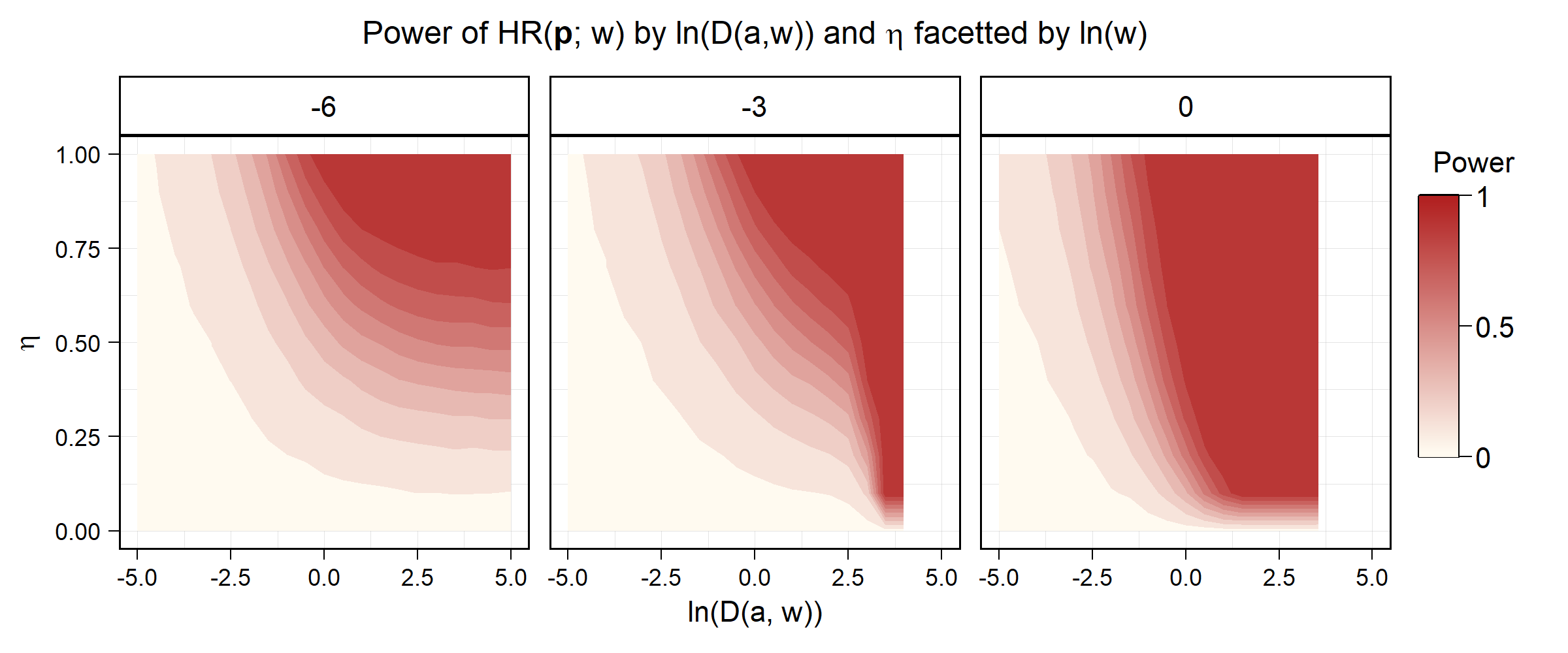}
  \caption{The power of $\hrpool{\ve{p}}{w}$ under $H_3$ by $D(a,w)$ and $\prevalence$ split by $\ln(w)$. Note how the contour for a power of one extends nearly to the bottom of the plot when $w = 1$, but stops near 0.75 when $w = e^{-6}$.}
  \label{fig:H3power}
  \end{center}
\end{figure}

Figure \ref{fig:H3power} places the measure of strength of evidence horizontally and the measure of prevalence vertically. Including $\prevalence = 0$ captures the behaviour of $\hrpool{\ve{p}}{w}$ under $H_0$ along the bottom edge of the power surface and including $\prevalence = 1$ captures its behaviour under $H_4$ along the top edge for a range of KL divergences. In particular, this means the top left part of each subplot corresponds to a generative process for $\ve{p}$ with relatively weak evidence in all $p$-values of $\ve{p}$ and the bottom right part corresponds to a generative process with strong evidence concentrated in a small number of $p$-values. Figures \ref{fig:powersByMW} and \ref{fig:MisPowers} show that rejection occurs at a rate of $\alpha$ once the evidence is weak enough, so the top left corner is less interesting than the top centre, which gives the power for tests of moderate strength spread throughout $\ve{p}$.

When $w$ is small and $l_{HR}(\ve{p}; w) \approx l_{Pea}(\ve{p})/2$, $\hrpool{\ve{p}}{w}$ is relatively weak when strong evidence is concentrated in a few tests. The contours for $w = e^{-6}$ are nearly horizontal for log KL divergences between 2.5 and 5 and the power is nearly identical in the bottom right and the top left corners. On the other hand, when $w \approx 1$ and $l_w(\ve{p}) \approx -l_{Fis}(\ve{p})/2$, $\hrpool{\ve{p}}{w}$ is relatively powerful when evidence is strong and concentrated in a few tests. This is starkly visible when $w = 1$, the power contours are nearly vertical, and the bottom right corner has a power of almost 1. Between these extremes, the power contours display a mix of these seemingly oppositional sensitivities to strength and prevalence.

\subsection{Case 4: when both the hypothesis and $w$ are incorrect} \label{subsec:h3omega}

Finally, consider the most pessimistic case. In the preceding section, $w$ at least matched the non-null distribution $f$ under $H_3$, but this may not always be so. Suppose, instead, everything is mis-specified.
That is, generate $M\prevalence$ $p$-values according to $f = Beta(a, 1/\omega + a(1 - 1/\omega))$ and $M(1 - \prevalence)$ from the uniform distribution and compute $\hrpool{\ve{p}}{w}$ for each of $w = e^{-6}, e^{-3}, 1/2,$ and $1$ for every setting from Section \ref{subsec:h3w} and repeat this generation 10,000 times. The power, given by proportion of rejections, follows an interesting pattern in the strength and prevalence of evidence, which is illustrated by the power contours of $\hrpool{\ve{p}}{1}$ in Figure \ref{fig:H3FisConts} and the difference in power contours in Figure \ref{fig:H3DiffConts}.

\begin{figure}[!ht]
  \begin{center}
    \includegraphics{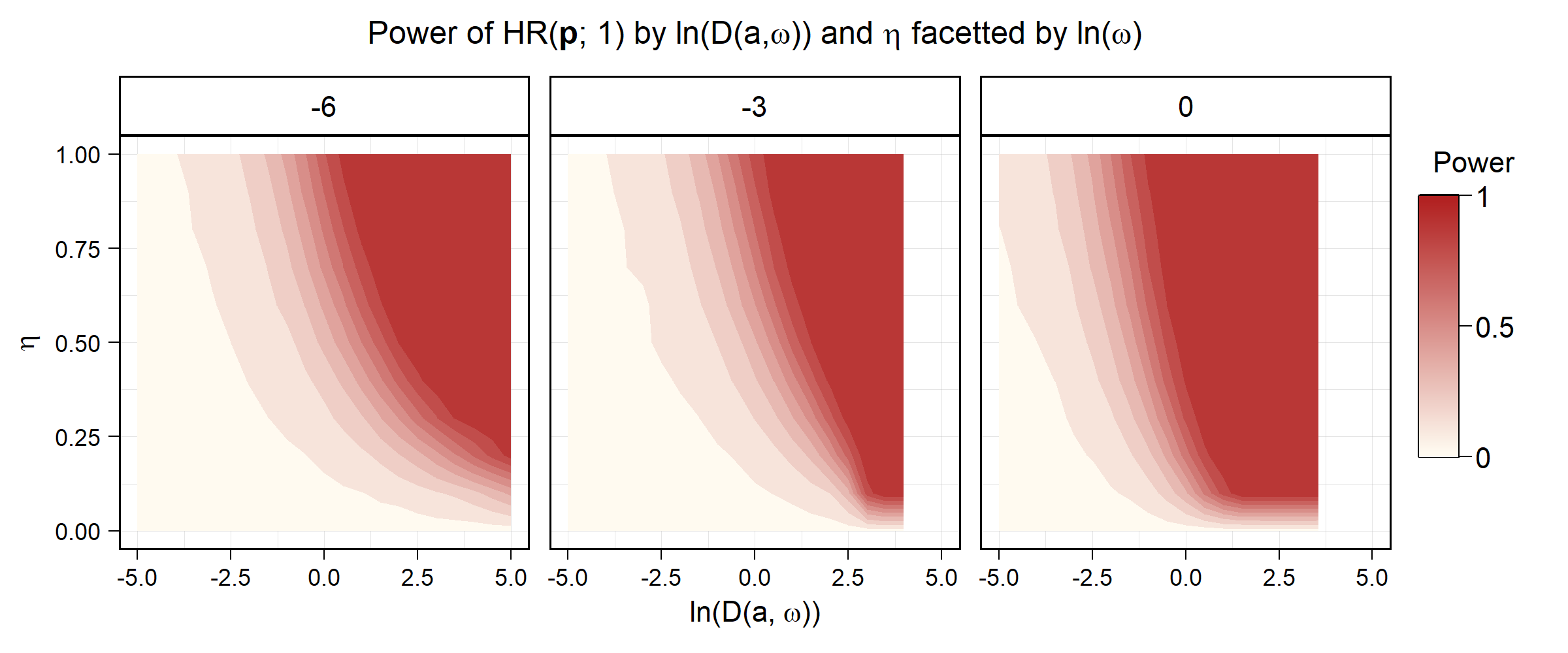}
        \caption{Power contours of $\hrpool{\ve{p}}{1}$ under $H_3$ by $D(a,\omega)$ and $\prevalence$ displayed using the saturation palette of Figure \ref{fig:H3power}.}
    \label{fig:H3FisConts}
  \end{center}
\end{figure}

\begin{figure}[!ht]
  \begin{center}
    \includegraphics{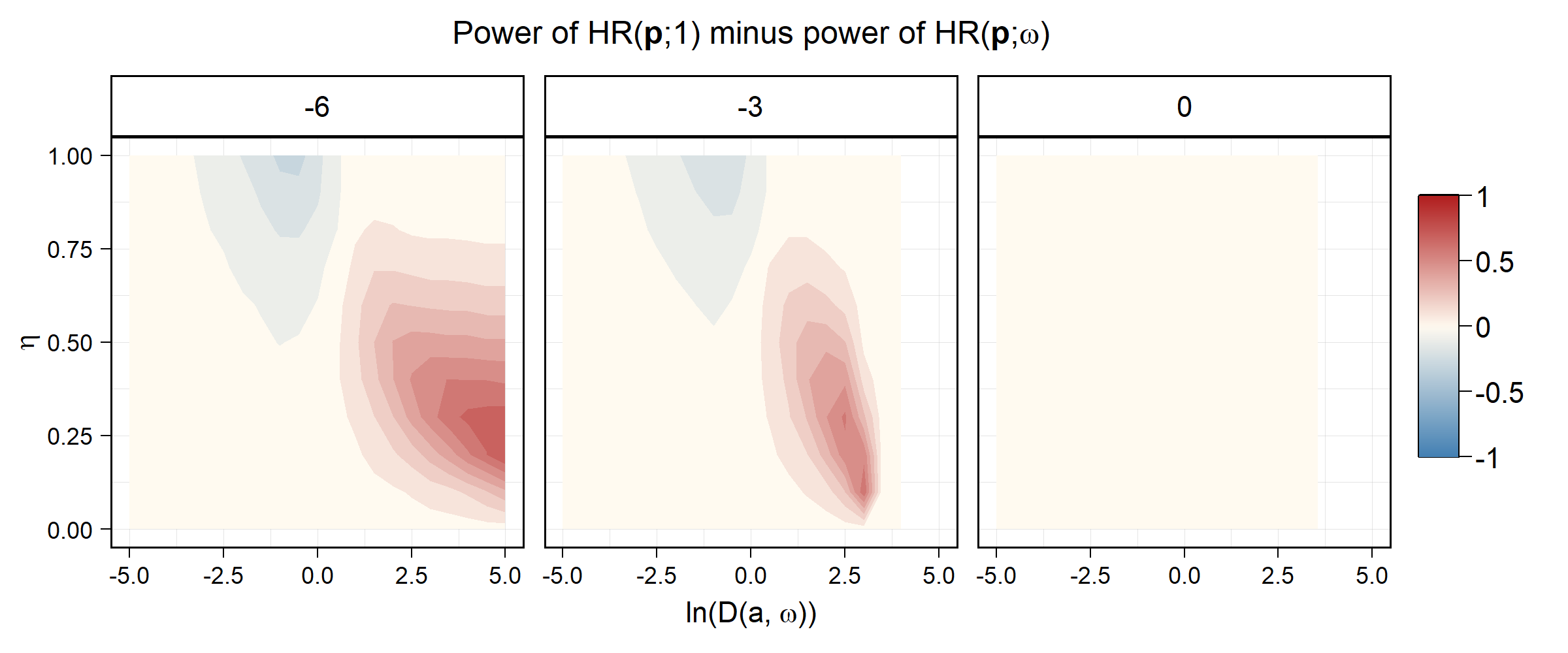}
    \caption{Contours of the difference in power of $\hrpool{\ve{p}}{1}$ and $\hrpool{\ve{p}}{\omega}$ under $H_3$ by $D(a,w)$ and $\prevalence$, displayed using a divergent palette. $\hrpool{\ve{p}}{1}$ is more powerful for strong evidence concentrated in a few tests and less powerful when weak evidence is spread among most tests.}
    \label{fig:H3DiffConts}
  \end{center}
\end{figure}

Compared to the power contours of $\hrpool{\ve{p}}{\omega}$ in Figure \ref{fig:H3power}, the contours of $\hrpool{\ve{p}}{1}$ in Figure \ref{fig:H3FisConts} show greater red saturation -- and thus greater power -- in the lower right corner of every panel. Thus, $\hrpool{\ve{p}}{1}$ has power greater than or equal to that of $\hrpool{\ve{p}}{\omega}$ in this region for every $\omega$. The magnitude of this improvement is unclear, however, due to the difficulty comparing the contours between plots. Figure \ref{fig:H3DiffConts} facilitates the comparison by plotting the difference between the power contours of $\hrpool{\ve{p}}{1}$ and $\hrpool{\ve{p}}{\omega}$ directly for all settings. This more precise plot demonstrates that the most powerful $\hrpool{\ve{p}}{\omega}$ to test $H_3$ against $H_0$ depends on the the strength and prevalence of evidence alone, \emph{not} $\omega$.

Specifically, Figure \ref{fig:H3DiffConts} shows that $\hrpool{\ve{p}}{1}$ is more powerful than the correctly-specified $\hrpool{\ve{p}}{\omega}$ in the lower right corner of the $D(a,\omega)$, $\prevalence$ space and does worse left of centre at the top for both $\omega = e^{-6}$ and $e^{-3}$. In the final panel where $\omega = 1$, $\hrpool{\ve{p}}{1} = \hrpool{\ve{p}}{\omega}$ and so the difference in their powers is zero everywhere. Recalling the interpretation of these regions for these facets, this indicates $\hrpool{\ve{p}}{1}$ is more powerful than $\hrpool{\ve{p}}{\omega}$ at testing $H_3$ against $H_0$ when strong evidence is concentrated in a few tests, but is less powerful when evidence is weaker and spread widely. This pattern holds for every $\omega$, albeit with differences in magnitude. Additionally, it is not symmetric: when $\hrpool{\ve{p}}{1}$ is more powerful, the magnitude of the difference tends to be greater than in regions where it is less powerful.

This parallels \cite{loughin2004systematic}, who found $\fispool(\ve{p})$ more powerful than other alternatives against $H_3$ when strong evidence was concentrated in a few tests. As $l_{Fis}(\ve{p}) \propto l_{HR}(\ve{p}; 1) \implies \fispool(\ve{p}) = \hrpool{\ve{p}}{1}$, this means that Fisher's method is once again relatively powerful for the same setting among the UMP family of tests $\hrpool{\ve{p}}{\omega}$. The consistency of this result in both simulation studies warrants further investigation. To aid in this, marginal and central rejection levels are introduced to capture the tendency of a test to reject weak evidence spread among all tests and strong evidence concentrated in a few, along with a meaningful quotient that combines the them.


\section{Central and marginal rejection levels in pooled $p$-values} \label{rejlev}

Recall that any pooled $p$-value $g(\ve{p})$ behaves like a univariate $p$-value, that is $g(\ve{p}) \in [0,1]$ and under $H_0$, $g(\ve{p}) \sim U$. It may be based on order statistics and use Equation (\ref{MC:eq:orderpool}),
$$\ordpool{\ve{p}}{k} = \sum_{l = k}^M {M \choose l}p_{(k)}^l (1 - p_{(k)})^{M-l},$$
or it may be based on the quantile transformations of Equation (\ref{MC:eq:quantilePval}),
$$g(\ve{p}) = 1 - F_M \left ( \sum_{i = 1}^M F^{-1} ( 1 - p_i ) \right ).$$
In either case, $g(\ve{p})$ must be non-decreasing in every argument if it is to create convex acceptance regions and therefore be admissible\footnote{See \cite{birnbaum1954combining} and \cite{owen2009karl} for a detailed discussion of admissibility and convexity.}, continuous in every $p_i$ if it is to have well-defined rejection boundaries, and symmetric in $p_i$ if no margin is to be favoured. Given these common properties, concepts of marginal and central rejection can be defined in order to describe rejection in the cases of evidence against $H_0$ concentrated in one test and evidence against $H_0$ spread among all tests, respectively.

These correspond to the largest $p$-values for a given FWER $\alpha$ which still result in rejection in two separate cases. The first of these, the marginal level of the pooled $p$-value, is the largest value of the minimum $p$-value which still leads to rejection at level $\alpha$. The second, the central level of the pooled $p$-value, is the largest value which all $p$-values can take simultaneously while still resulting in rejection.

\subsection{Characterizing central behaviour} \label{MC:pool:pc}

The simulations of Section \ref{sec:power} suggest a pooled $p$-value $g(\ve{p})$ which is powerful at rejecting weak evidence spread among all $p$-values will reject $H_0$ when all tests give relatively large $p$-values compared to $\alpha$. Under the rejection rule $g(\ve{p}) \leq \alpha$, this is captured by the largest $p_c$ such that $g(\ve{p}_c) = \alpha$ for  $\ve{p}_c = \tr{(p_c, \dots, p_c)}$. Explicitly:
\begin{definition}[The central rejection level]
  For a pooled $p$-value $g(\ve{p})$, the central rejection level, $p_c$, is the largest $p$-value for which $g(\ve{p}) \leq \alpha$ when $\ve{p} = \tr{(p_c, \dots, p_c)}$. That is
  \begin{equation} \label{MC:eq:centralRejectionDef}
    p_c(g) = \sup \big\{ p \in [0,1]: g(p, p, \dots, p) \leq \alpha \big\}
  \end{equation}
\end{definition}
$p_c$ quantifies the maximum $p$-value shared by all tests which still leads to rejection and therefore is directly related to the power of $g(\ve{p})$ along $\ve{p}_c$. As $g(\ve{p})$ is non-decreasing, rejection occurs for any $\ve{p}$ in the hypercube $[0, p_c]^M$ and so a larger $p_c$ implies rejection for a larger volume of $[0,1]^M$. It also suggests pooled $p$-values smaller than $p_c$ within this hypercube if $g(\ve{p})$ is continuous and monotonic.

This general definition can be applied to the pooled $p$-values of Equations (\ref{MC:eq:orderpool}) and (\ref{MC:eq:quantilePval}) to obtain simple expressions for $p_c$.
\begin{proposition}[The central rejection level of order statistics]
  $p_c(\ordpool{\ve{p}}{k})$ is given by the largest $p \in [0,1]$ which satisfies
  \begin{equation} \label{MC:eq:ordpc}
    \sum_{l = k}^M {M \choose l} p^l (1 - p)^{M-l} \leq \alpha.
  \end{equation}
\end{proposition}
\begin{proof}
  The definition of $p_c$ forces $p_{(1)} = p_{(2)} = \dots = p_{(M)} = p_c$. Therefore $p_{(k)} = p_c$ and so $p_c$ is the largest $p \in [0,1]$ which satisfies $\ordpool{(p, \dots, p)}{k} \leq \alpha$. Expanding $\ordpool{(p, \dots, p)}{k}$ gives Equation (\ref{MC:eq:ordpc}).
\end{proof}
While this cannot generally be solved for a closed-form $p_c$, the particular cases of $\ordpool{\ve{p}}{1} = \tippool(\ve{p})$ and $\ordpool{\ve{p}}{M}$ admit
\begin{equation} \label{MC:eq:tippettpc}
  p_c(\tippool) = 1 - (1 - \alpha)^{\frac{1}{M}}
\end{equation}
and
\begin{equation} \label{MC:eq:maxOrdpc}
  p_c(\ordpool{}{M}) = \alpha^{\frac{1}{M}}
\end{equation}
respectively. Also note that $p_c(\ordpool{}{k})$ defines a constant rejection boundary around the regions of $[0,1]^M$ where $k - 1$ $p$-values are less than or equal to $p_c$, $M-k$ elements are greater, and exactly one is equal to $p_c$. In particular, this means that $p_c(\tippool) = p_c(\ordpool{}{1})$ is constant along each margin, as $M-1$ points are greater than $p_c(\tippool)$ along each margin. Next, consider $p_c$ for the general quantile transformation of Equation \ref{MC:eq:quantilePval}.

\begin{proposition}[The central rejection level of quantile pooled $p$-values]
  Given a pooled $p$-value based on quantile transformations as in Equation (\ref{MC:eq:quantilePval}),
  $$g(\ve{p}) = 1 - F_M \left ( \sum_{i = 1}^M c_i F^{-1} ( 1 - p_i ) \right ),$$
  the central rejection level is given by
  \begin{equation} \label{MC:eq:quantileEqual}
    p_c(g(\ve{p})) = 1 - F \left ( \frac{1}{\sum_{i = 1}^M c_i} F_{M}^{-1}(1 - \alpha) \right )
  \end{equation}
  if $F$ and $F_M$ are continuous CDFs.
\end{proposition}
\begin{proof}
As $F$ and $F_M$ are CDFs , they are monotonically non-decreasing real-valued functions over their ranges. If $F$ and $F_M$ are also continuous, then $F^{-1}$ and $F_M^{-1}$ are continuous. Therefore, we can drop the supremum from Equation \ref{MC:eq:centralRejectionDef} and consider the equality
$$\alpha = g(p_c, p_c, \dots, p_c).$$
Expanding $g$ and solving for $p_c$ implies
\begin{equation*} 
  p_c = 1 - F \left ( \frac{1}{\sum_{i = 1}^M c_i} F_{M}^{-1}(1 - \alpha) \right ).
\end{equation*}
\end{proof}
If the $p$-values are unweighted, then $c_1 = c_2 = \dots = c_M = 1$, $\sum_{i = 1}^M c_i = M$ and the behaviour of $p_c$ depends on the relative growth of $F_M^{-1}$ in $M$. If $F_{M}^{-1}(1 - \alpha)$ grows in $M$ such that $\frac{1}{M} F^{-1}_M (1 - \alpha)$ is unbounded, $p_c$ will go to zero. If, on the other hand, $\lim_{M \rightarrow \infty} \frac{1}{M} F^{-1}_M (1 - \alpha) = c < \infty$, $p_c$ will go to $1 - F(c)$. This provides a general expression for the soft truncation threshold of \cite{zaykinetal2007combining} and suggests interesting asymptotic behaviour for pooled $p$-values based on quantile functions along the line $p_1 = p_2 = \dots = p_M$.

This behaviour can be demonstrated concretely for several quantile transformations. \cite{stoufferetal1949american} takes $F = \Phi$ and $F_M = \sqrt{M} \Phi$\footnote{Though it uses scaling in computation to avoid this, the distribution is the same.} in Equation \ref{MC:eq:quantilePval} to give $\stopool(\ve{p})$ and so
$$\lim_{M \rightarrow \infty} \frac{1}{M} F^{-1}_M (1 - \alpha) = \lim_{M \rightarrow \infty} \frac{1}{\sqrt{M}} \Phi^{-1}(1 - \alpha) = 0$$
for all $\alpha > 0$. This implies
\begin{equation} \label{MC:eq:stoufferEqual}
  \lim_{M \rightarrow \infty} p_c(\stopool) = 1 - \Phi(0) = \frac{1}{2}.
\end{equation}
Indeed, taking $\stopool(\ve{p})$, substituting $p_1 = \dots = p_M = p_c(\stopool)$, and taking the limit gives
$$\lim_{M \rightarrow \infty} 1 - \Phi \left ( \frac{1}{\sqrt{M}} \sum_{i = 1}^M \Phi^{-1} (1 - p_c(\stopool)) \right ) = 1 - \Phi \left ( \lim_{M \rightarrow \infty} \sqrt{M} \Phi^{-1} (1 - p_c(\stopool)) \right ).$$
Evaluating further gives
\begin{equation*}
  \lim_{M \rightarrow \infty} \sqrt{M} \Phi^{-1} (1 - p_c) = \begin{cases}
    -\infty & \text{ when } p_c > \frac{1}{2} \\
    0 & \text{ when } p_c = \frac{1}{2} \\
    \infty & \text{ when } p_c < \frac{1}{2},
    \end{cases}
\end{equation*}
and so $\stopool(\ve{p})$ is either $0$, $\frac{1}{2}$, or $1$ for large $M$ when $p_1 \approx p_2 \approx \dots \approx p_M$. Furthermore, it will reject $H_0$ for \emph{any} FWER level $\alpha$ if $p_1, p_2, \dots, p_M$ are all less than $\frac{1}{2}$ when $M$ is large enough. 

$F_{\chi}(x; \kappa )$ admits similar analysis. By the central limit theorem
$$\lim_{\kappa \rightarrow \infty} \chi^2_{\kappa}  \rightarrow Z \sim N(\kappa, 2\kappa )$$
where $N(\mu, \sigma^2)$ is a normal distribution with mean $\mu$ and variance $\sigma^2$. Therefore, as $M \rightarrow \infty$
\begin{equation*}
  F_{\chi}(x; M\kappa ) \rightarrow \Phi \left ( \frac{x - M\kappa}{\sqrt{2M\kappa}} \right ).
\end{equation*}
This implies that the pooled $p$-value based on the $\chi^2$ quantile has the limiting value
$$\lim_{M \rightarrow \infty} 1 - F_{\chi} \Big ( M F_{\chi}^{-1} \big ( 1 - p_c; \kappa \big ); M\kappa \Big ) = 1 - \lim_{M \rightarrow \infty} \Phi \left ( \frac{M F_{\chi}^{-1} \big ( 1 - p_c; \kappa \big ) - M\kappa}{\sqrt{2M\kappa}} \right )$$
when $p_1 = \dots = p_M = p_c$. As $\Phi$ is absolutely continuous the limit can be taken inside the argument to give
$$1 - \Phi \left ( \lim_{M \rightarrow \infty} \sqrt{\frac{M}{2\kappa}} \left [ F_{\chi}^{-1}(1 - p_c; \kappa) - \kappa \right ] \right ).$$
Now,
\begin{equation*}
  \lim_{M \rightarrow \infty} \sqrt{\frac{M}{2\kappa}} \left [ F_{\chi}^{-1}(1 - p_c; \kappa) - \kappa \right ] = \begin{cases}
    -\infty & \text{ when } p_c > 1 - F_{\chi}(\kappa; \kappa) \\
    0 & \text{ when } p_c = 1 - F_{\chi}(\kappa; \kappa) \\
    \infty & \text{ when } p_c < 1 - F_{\chi}(\kappa; \kappa),
    \end{cases}
\end{equation*}
and so the pooled $p$-value based on the $\chi^2$ quantile transformation behaves similarly to $\stopool(\ve{p})$. It is asymptotically either 1, $\frac{1}{2}$, or 0 when $p_1 \approx p_2 \approx \dots \approx p_M$ depending on whether all are greater than, equal to, or less than $1 - F_{\chi}(\kappa; \kappa)$. Additionally, this implies that
\begin{equation}  \label{MC:eq:cinvieEqLim}
  p_c = 1 - F_{\chi}(\kappa; \kappa)
\end{equation}
for the $\chi^2$ quantile case.
So, although $p_c$ depends on $\kappa$, all $\chi^2$ quantile pooled $p$-values have identical behaviour about their respective $p_c$.


The form of Equation (\ref{MC:eq:quantilePval}) suggests that this result applies generally. Suppose the random variable with CDF $F$ has a mean $\mu$ and variance $\sigma^2$. Under $H_0$, $F^{-1}(1 - p_i)$ for $i = 1, \dots, M$ are independent and identically-distributed realizations of this random variable and so $\sum_{i = 1}^M F^{-1}(1 - p_i)$ is asymptotically normally distributed with mean $M\mu$ and variance $M\sigma^2$ by the central limit theorem. Therefore $F_M \xrightarrow[M \rightarrow \infty]{} \sqrt{M} \sigma \Phi + M \mu$ for the pooled $p$-value based on $F$ and the harsh asymptotic boundary at $p_c$ derived for $\stopool$ will occur for \emph{any} evidential statistic that uses quantile functions.

\subsection{Characterizing marginal behaviour} \label{MC:pool:pr}

Besides the central behaviour of a pooled $p$-value $g(\ve{p})$, the simulations of Section \ref{sec:power} indicate large differences in power occur for the rejection rule $g(\ve{p}) \leq \alpha$ when strong evidence exists in a single test. This is captured by the \emph{marginal rejection level at b}.

\begin{definition}[The marginal rejection level at $b$]
  For a symmetric pooled $p$-value $g(\ve{p})$, the marginal rejection level at $b$, $p_r(g; b)$, is the largest individual $p$-value in $[0, b]$ for which $g(\ve{p}) \leq \alpha$ when all other $p$-values are $b \in [0,1]$. Without loss of generality, define
  \begin{equation} \label{MC:eq:marginalRejectB}
    p_r(g; b) = \sup \big\{ p_1 \in [0,b] : g(p_1, b, \dots, b) \leq \alpha \big\}.
  \end{equation}
  In particular, the marginal value when $b = 1$ is of interest, that is when there is minimal evidence against all hypotheses other than $H_{01}$. Therefore, also define
  \begin{equation} \label{MC:eq:marginalReject}
    p_r(g) = \lim_{b \rightarrow 1} \sup \big\{ p_1 \in [0,b] : g(p_1, b, \dots, b) \leq \alpha \big\}.
  \end{equation}
\end{definition}
Note that symmetry is only necessary to avoid defining marginal rejection levels for each index $i \in \{ 1, \dots, M\}$ separately and that the term \emph{marginal rejection level} refers to Equation (\ref{MC:eq:marginalReject}). If $g$ is non-decreasing in all of its arguments, $p_r(g; b)$ gives the largest value of $p_{(1)}$ that still leads to rejection at $\alpha$ when the evidence in all other $p$-values is bounded at $b$. The most extreme version of this measure is given by $p_r(g)$. By taking $b = 1$, it measures the power of $g$ for evidence in a single test when all other tests provide no evidence against $H_0$, and so the sensitivity of $g$ to evidence in a single test. This leads to a key lemma for $\ordpool{\ve{p}}{1} = \tippool(\ve{p})$.
\begin{lemma}[The marginal rejection level for the minimum statistic] \label{lem:tipmar}
  The marginal rejection level for $g_{Tip}$ has two cases:
  \begin{equation*}
    p_r(\tippool; b) = \begin{cases} b & \text{ for } b < 1 - (1 - \alpha)^{\frac{1}{M}} \\
      1 - (1 - \alpha)^{\frac{1}{M}} & \text{ for } b \geq 1 - (1 - \alpha)^{\frac{1}{M}}.
    \end{cases}
  \end{equation*}
\end{lemma}
\begin{proof}
  Recall that
  $$\tippool(\ve{p}) = 1 - (1 - p_{(1)})^M$$
  is a function of the minimum alone. Rejection occurs when $\tippool(\ve{p}) \leq \alpha$, or rather when $p_{(1)} \leq 1 - (1 - \alpha)^{\frac{1}{M}}$. When $b < 1 - (1 - \alpha)^{\frac{1}{M}}$ and $p_1 \leq b$, all values are below the rejection threshold and so $p_{(1)}$ attains its upper bound. Therefore
  $$p_r(\tippool; b) = \sup \big\{p_1 \in [0,b] : g(p_1, b, \dots, b) \leq \alpha\big\} = b.$$
  When $b \geq 1 - (1 - \alpha)^{\frac{1}{M}}$, rejection will only occur if $p_{(1)}$ is below the rejection threshold at $\alpha$, and so
  $$p_r(\tippool; b) = \sup \big\{p_1 \in [0,b] : g(p_1, b, \dots, b) \leq \alpha\big\} = 1 - (1 - \alpha)^{\frac{1}{M}}.$$
\end{proof}
A direct consequence of Lemma \ref{lem:tipmar} is that $p_r(g_{Tip}) = p_c(g_{Tip})$, which Theorem \ref{MC:thm:pc>pr} proves is uniquely true for $\tippool(\ve{p})$.

As with $p_c$, a larger $p_r$ indicates greater power in a particular region of the unit hypercube. While $p_c$ defines the rejection cube $[0, p_c]^M$, $p_r$ defines the rejection shell $\{\ve{p} \in [0,1]^M : p_{(1)} \leq p_r\}$ with a flat boundary at $p_r$ along each margin. A larger $p_r$ implies a larger shell and therefore a greater volume of $[0,1]^M$ where $H_0$ is rejected and smaller pooled $p$-values within this volume if $g(\ve{p})$ is monotonic. Again, general expressions are provided for $p_r$ for the order statistic and quantile transformation pooled $p$-values.

\begin{proposition}[The marginal rejection level for order statistics]
  For $k \geq 2$, $p_r(\ordpool{}{k}, b) = b$ when $\sum_{l = k}^M {M \choose l} b^l (1 - b)^{M-1} \leq \alpha$ and does not exist otherwise.
\end{proposition}
\begin{proof}
  Recall that
  $$g_{Ord}(\ve{p}; k) = \sum_{l = k}^M {M \choose l}p_{(k)}^l (1 - p_{(k)})^{M-l}$$
  and note Equation \ref{MC:eq:marginalRejectB} forces $p_{(k)} = b$ for all $k > 1$. If $\sum_{l = k}^M {M \choose l} b^l (1 - b)^{M-1} \leq \alpha$, then the supremum of $p_1$ is $b$. On the other hand, if $\sum_{l = k}^M {M \choose l} b^l (1 - b)^{M-1} \geq \alpha$ there is no value of $p_1$ which leads to rejection and so $p_r(\ordpool{}{k}, b)$ does not exist.
\end{proof}

In particular, this implies that $p_r(\ordpool{}{k})$ does not exist for $k \geq 2$, in other words the pooled $p$-value based on $p_{(k)}$ has a value independent of $p_{(1)}$ for $k \geq 2$. So long as $k$ tests are less than a particular bound, $\ordpool{\ve{p}}{k}$ will reject. If fewer than $k$ are below that bound, the values of these small $p$-values are irrelevant.

\begin{proposition}[The marginal rejection level of quantile transformation statistics]
    Given an unweighted evidential statistic based on quantile transformations as in Equation \ref{MC:eq:quantilePval},
    $$g(\ve{p}) = 1 - F_M \left ( \sum_{i = 1}^M F^{-1} ( 1 - p_i ) \right ),$$
     if $F_M$ and $F$ are both continuous then
     \begin{equation*}
       p_r(g;b) = 1 -  F \Big ( F_M^{-1}( 1 - \alpha ) - [M - 1]F^{-1}( 1 - b ) \Big ).
\end{equation*}
 Further, if both are absolutely continuous
\begin{equation} \label{MC:eq:marginalRejAbsCont}
  p_r(g) = 1 -  F \left ( F_M^{-1}( 1 - \alpha ) - [M - 1] \lim_{x \rightarrow 0+} F^{-1}( x ) \right )
\end{equation}
\end{proposition}
\begin{proof}
  Substituting Equation \ref{MC:eq:quantilePval} into Equation \ref{MC:eq:marginalRejectB} gives
       \begin{equation*}
         p_r(g;b)  = \sup \Bigg\{p : 1 - F_M \bigg ( F^{-1} ( 1 - p ) + [M-1] F^{-1} ( 1 - b ) \bigg ) \leq \alpha\Bigg\}.
       \end{equation*}
       As both $F_M$ and $F$ are CDFs, they are non-decreasing, if they are also continuous their inverses exist and the supremum can be dropped to give
     \begin{equation*}
       p_r(g;b) = 1 -  F \Big ( F_M^{-1}( 1 - \alpha ) - [M - 1]F^{-1}( 1 - b ) \Big ).
     \end{equation*}
     If they are both absolutely continuous, then the limit
     $$1 - \lim_{b \rightarrow 1} F \Big ( F_M^{-1}( 1 - \alpha ) - (M - 1)F^{-1}( 1 - b ) \Big ) $$
     can be taken into the argument of $F$ to give Equation \ref{MC:eq:marginalRejAbsCont}.
\end{proof}

Many proposals use absolutely continuous CDFs, so this can be readily applied. $\stopool(\ve{p})$, for example, has
$$p_r(\stopool) = 1 - \Phi \left ( \sqrt{M} \Phi^{-1}(1 - \alpha) - [M - 1] \lim_{x \rightarrow 0+} \Phi^{-1}(x) \right ) = 0$$
for any $\alpha$ and $M$ as $\lim_{x \rightarrow 0} \Phi(x) = - \infty$. Similarly, the proposal by \cite{mudholkar1977logit} has $p_r = 0$, as it uses the logistic distribution which also has $\lim_{x \rightarrow 0} F(x) = -\infty$. This suggests that, for a large enough $p$-value on all remaining tests, no level of evidence in a single test will cause the rejection of $H_0$ for either of these pooled $p$-values; their marginal rejection levels are always 0 for large enough $b$.

\subsection{The centrality quotient} \label{chipool:centquot}

Beyond providing definitions that clarify the power of a pooled $p$-value to detect evidence spread among all tests and evidence in a single test, $p_c$ and $p_r$ as defined in Equations (\ref{MC:eq:marginalReject}) and (\ref{MC:eq:centralRejectionDef}) can be combined into a single value summarizing the relative preference for diffuse or concentrated evidence. First, a key relationship between $p_c$ and $p_r$ is proven.

\begin{theorem}[Order of $p_c$ and $p_r$] \label{MC:thm:pc>pr}
  For a pooled $p$-value $g(\ve{p})$ that is continuous, symmetric, and monotonically non-decreasing in all arguments, $p_c \geq p_r$ if both exist. Furthermore, equality occurs iff $g(\ve{p})$ is constant in $p_k$ for $p_k \neq p_{(1)}$, that is if $g(\ve{p}) = f(p_{(1)})$ is a function of the minimum $p$-value alone.
\end{theorem}
\begin{proof}
  Consider $p_c(g)$ as in Definition \ref{MC:eq:centralRejectionDef}. Then
  $$p_c(g) = \sup \big\{p \in [0,1]: g(p, \dots, p) \leq \alpha \big\}.$$
  Suppose
  $$p_r(g) = \lim_{b \rightarrow 1} \sup \big\{ p_1 \in [0,b]: g(p_1, b, \dots, b) \leq \alpha \big\}$$
  exists. If $g$ is symmetric, $p_r(g)$ captures the marginal rejection level of $g$ in all margins. If $g$ is continuous, then both $p_c(g)$ and $p_r(g)$ lie on the $\alpha$ level surface of $g$. Therefore
  $$g(p_c, \dots, p_c) = \alpha = g(p_r, 1, \dots, 1).$$
  But $g$ is non-decreasing in all of its arguments, so
  $$g(p_c, 1, \dots, 1) \geq g(p_c, p_c, \dots, p_c) = g(p_r, 1, \dots, 1)$$
  and therefore
  $$p_c \geq p_r.$$
  If $p_c = p_r$, then substitute
  $$\alpha = g(p_c, \dots, p_c) = g(p_r, \dots, p_r).$$
  As $g$ is non-decreasing
  $$g(p_r, \dots, p_r) \leq g(p_r, 1, \dots, 1) = \alpha$$
  and so
  $$g(p_r, 1, \dots, 1) = g(p_r, p_r, \dots, p_r).$$
  This implies that the average slope of $g$ over $[p_r, 1]$, equivalently $[p_c,1]$, is zero for all $p_k \neq p_1$. As $g$ is continuous and non-decreasing, this implies that the slope must be zero for every point in this interval for all $p_k \neq p_1$. As $p_k \geq p_1$ for all $p_k \in \ve{p}$ over this region, $p_1 = p_{(1)}$ by definition. By the symmetry of $g$, the same argument holds for every $p_k$. Therefore $g(\ve{p}) = f( p_{(1)})$ for some non-decreasing function $f$.

  To prove the reverse direction note that if $g(\ve{p}) = f( p_{(1)})$, then
  $$\alpha = g(p_c, p_c, \dots, p_c) = g(p_c, 1, \dots, 1)$$
  and so $p_c = p_r$ by the definition of $p_r$ and the continuity of $g$. By symmetry, this same argument holds for any margin.
\end{proof}

Two facts follow directly from this proof. First, Theorem \ref{MC:thm:pc>pr} implies that $p_c = p_r$ only for $\tippool(\ve{p}) = \ordpool{\ve{p}}{1}$ among symmetric, continuous, monotonically non-decreasing $p$-values as $\tippool(\ve{p})$ is the unique pooled $p$-value defined by $p_{(1)}$. A second corollary is the existence of a sensible \emph{centrality quotient} to quantify the balance between central and marginal rejection levels in pooled $p$-values.

\begin{definition}[The centrality quotient]
  Suppose $g$ is a continuous, symmetric, and monotonically non-decreasing pooled $p$-value for which $p_r(g)$ and $p_c(g)$ defined as in Equations (\ref{MC:eq:marginalReject}) and (\ref{MC:eq:centralRejectionDef}) exist, define the centrality quotient
  \begin{equation} \label{MC:eq:centralityQuot}
    \centquot(g) = \frac{p_c(g) - p_r(g)}{p_c(g)}.
  \end{equation}
\end{definition}

Theorem \ref{MC:thm:pc>pr} implies that $\centquot(g) \in [0,1]$ with meaningful bounds. If $\centquot(g) = 0$, $g(\ve{p})$ will reject based on the smallest $p$-values alone, increasing the marginal rejection level as large as possible while staying non-decreasing. Moreover, $\centquot(g) = 0$ implies $g(\ve{p})$ is the pooled $p$-value based on $p_{(1)}$ alone, $\tippool(\ve{p})$. In contrast, when $\centquot(g) = 1$, $g$ cannot reject based on the evidence contained in a single test, instead it requires evidence in many or all tests, for example $\stopool(\ve{p})$. Between these extremes, pooled $p$-values with larger centrality quotients will reject $H_0$ for a larger range of $p_c$ values and a smaller range of $p_r$ values, and so will be more powerful at detecting evidence spread broadly at the cost of power when evidence is concentrated in a small number of $p$-values.

Indeed, increasing $w$ decreases the centrality quotient of $\hrpool{\ve{p}}{w}$. This matches the empirical results obtained in Section \ref{subsec:h3omega} and in particular Figure \ref{subsec:h3omega}, where larger $w$ values provided greater power when the prevalence of evidence was small but the strength of evidence was large and smaller $w$ values gave greater power in the case of weak evidence with high prevalence. For $\alpha = 0.05$, the centrality quotients of a range of $w$ values in $\hrpool{\ve{p}}{w}$ are compared to those of several quantile transformation proposals in Table \ref{MC:tab:centrality} over a range of $M$ values. As predicted by the asymptotic argument at the end of Section \ref{MC:pool:pc}, every method tends towards a centrality of 1 as $M$ increases and each $F_M$ converges to a corresponding normal CDF.

\begin{table}[!ht]
  \begin{center}
    \begin{tabular}{|c|cccc|} \hline
      & \multicolumn{4}{c|}{$M$} \\
     Pooled $p$-value & 2 & 5 & 10 & 20 \\ \hline
     \cite{tippett1931methods} & 0 & 0 & 0 & 0 \\
     \cite{cinarviechtbauer2022poolr} & 0.83 & 0.99 & 1.00 & 1.00 \\
     \cite{stoufferetal1949american} & 1 & 1 & 1 & 1 \\
     \cite{fisher1932statistical} & 0.91 & 1.00 & 1.00 & 1.00 \\
     \cite{mudholkar1977logit} & 1 & 1 & 1 & 1 \\
     \cite{wilson2019harmonic} & 0.49 & 0.79 & 0.90 & 0.95 \\
     $\hrpool{\ve{p}}{e^{-6}}$ & 1.00 & 1.00 & 1.00 & 1.00 \\
     $\hrpool{\ve{p}}{e^{-3}}$ & 1.00 & 1.00 & 1.00 & 1.00 \\
     $\hrpool{\ve{p}}{1}$ & 0.91 & 1.00 & 1.00 & 1.00 \\ 
      \hline
  \end{tabular}
  \caption{Centrality quotients for certain pooled $p$-values.}
  \label{MC:tab:centrality}
  \end{center}
\end{table}

Beyond $\hrpool{\ve{p}}{w}$, other methods show the same relationship between the centrality quotient and regions of relative power in the empirical explorations in \cite{westberg1985combining}, \cite{loughin2004systematic}, and \cite{kocak2017meta}. Pooled $p$-values with larger centrality quotients are more powerful for weak evidence spread among all tests than those with smaller centrality quotients, but are relatively weak against strong evidence concentrated in a few tests. This is suggestive of an inverse relationship between $p_r$ and $p_c$ over different pooled $p$-values, but this is not the case generally. Consider, as a counter-example, $\stopool(\ve{p})$ and the proposal of \cite{mudholkar1977logit}: both have $p_r = 0$ but different values of $p_c$.


\section{Controlling the centrality quotient} \label{sec:chi}

Table \ref{MC:tab:centrality}, and others which could be constructed like it, provide only a limited ability to select a centrality quotient. Most of the proposals have centrality near 1, and all proposals approach 1 as $M$ increases. Rather than choose among these other limited proposals when power is desired in a particular region, this work proposes a family of quantile pooled $p$-values based on $\chi^2_{\kappa}$ which precisely controls the centrality for any $M$.
Following Equation (\ref{MC:eq:quantilePval}), define the $\chi^2_{\kappa}$ quantile pooled $p$-value
\begin{equation} \label{MC:eq:chikpool}
  \chipool{\ve{p}}{\kappa} = 1 - F_{\chi} \left ( \sum_{i = 1}^M F_{\chi}^{-1}(1 - p_i; \kappa); M\kappa \right )
\end{equation}
where $\kappa \in [0, \infty)$ is the the degrees of freedom of the quantile transformation applied to the $p_i$ and doubles as a centrality parameter that sets $\centquot(\chipool{}{\kappa})$ arbitrarily.\footnote{This is similar to the gamma method of \cite{zaykinetal2007combining}, but with a different parameter choice. It is possible the same control of $c$ may be obtained with a general gamma CDF, but sticking to the $\chi^2_{\kappa}$ simplifies the number of parameters from two to one.} This family of pooled $p$-values includes several widely-used previous proposals. Setting $\kappa = 2$ gives $\fispool(\ve{p})$, $\kappa = 1$ gives the proposal from \cite{cinarviechtbauer2022poolr}, taking $\lim_{\kappa \rightarrow \infty} \chipool{\ve{p}}{\kappa}$ gives $\stopool(\ve{p})$, and taking $\lim_{\kappa \rightarrow 0} \chipool{\ve{p}}{\kappa}$ gives $\tippool(\ve{p})$. While the former two are by definition, the latter must be proven. First, prove $\lim_{\kappa \rightarrow \infty} \chipool{\ve{p}}{\kappa} = \stopool(\ve{p})$ by applying the central limit theorem.

\begin{theorem}[Limiting value of $\chipool{\ve{p}}{\kappa})$ as $\kappa \rightarrow \infty$] \label{thm:gchiLarge}
  $$\lim_{\kappa \rightarrow \infty} \chipool{\ve{p}}{\kappa} = \stopool(\ve{p})$$
\end{theorem}
\begin{proof}
Note that Equation (\ref{MC:eq:chikpool}) is always a pooled $p$-value, i.e. has a uniform distribution for any choice of $\kappa$. By the CLT, $\lim_{\kappa \rightarrow \infty} F_{\chi}(x; \kappa) = \Phi(x)$, and so in the limit $\chipool{\ve{p}}{\kappa}$ becomes the pooled $p$-value derived from the sum of standard normal quantile transformations, $\stopool(\ve{p})$.
\end{proof}

The proof for $\chipool{\ve{p}}{0}$ is slightly more involved, and relies on Theorem \ref{MC:thm:pc>pr}.

\begin{theorem}[Limiting value of $\chipool{\ve{p}}{\kappa}$ for $\kappa = 0$] \label{thm:gchiSmall}
  $$\lim_{\kappa \rightarrow 0} \chipool{\ve{p}}{\kappa} = \tippool(\ve{p}) = \ordpool{\ve{p}}{1}$$
\end{theorem}
\begin{proof}
  Theorem \ref{MC:thm:pc>pr} proves that $p_r = p_c$ for a pooled $p$-value if and only if that pooled $p$-value is $\tippool = \ordpool{}{1}$. Therefore, the limit is proven if
  $$\lim_{\kappa \rightarrow 0} p_r \big ( \chipool{}{\kappa} \big ) = \lim_{\kappa \rightarrow 0} p_c \big ( \chipool{}{\kappa} \big )$$
  Expressing these quantities as probability statements gives
  $$p_c\big ( \chipool{}{\kappa} \big ) = P \left ( \chi^2_{\kappa} \geq \frac{1}{M} F_{\chi}^{-1} ( 1 - \alpha; M\kappa ) \right ),$$
  and
  $$p_r\big ( \chipool{}{\kappa} \big ) = P \left ( \chi^2_{\kappa} \geq F_{\chi}^{-1} ( 1 - \alpha; M\kappa ) \right ).$$
  The case of $\chi^2_0$ is a degenerate distribution at 0. That is
  $$F_{\chi^2}(x; 0) = \begin{cases} 0 & x < 0 \\ 1 & x \geq 0. \end{cases}$$
  This can also be seen from the limit of Markov's inequality for the $\chi^2_{\kappa}$ distribution,
  $$P(\chi^2_{\kappa} \geq a) \leq \frac{\kappa}{a},$$
  which goes to zero for any $a > 0$ as $\kappa \rightarrow 0$. As $F_{\chi}(x; \kappa)$ is continuous and monotonically increasing for all $\kappa$, this also implies
  $$F \left ( \frac{\kappa}{\alpha}; \kappa \right ) \geq 1 - \alpha$$
  $$\implies 1 - F \left ( \frac{\kappa}{\alpha}; \kappa \right ) \leq \alpha$$
  $$\implies P \left ( \chi^2_{\kappa} \geq \frac{\kappa}{\alpha} \right ) \leq \alpha$$
  $$\implies F^{-1}_{\chi}(1 - \alpha; \kappa) \leq \frac{\kappa}{\alpha}.$$
  This bound is not particularly tight, for $\alpha = 0.05$ it only restricts the $0.95$ quantile to be less than 20 times the mean. However, it suffices to evaluate
  \begin{align}
    & \lim_{\kappa \rightarrow 0} \left | \frac{1}{M} F^{-1}_{\chi} (1 - \alpha; M\kappa) - F^{-1}_{\chi} (1 - \alpha; M\kappa) \right | \nonumber \\
    & \nonumber \\
    = & \frac{M - 1}{M} \lim_{\kappa \rightarrow 0} F^{-1}_{\chi} (1 - \alpha; M\kappa) \nonumber \\
    & \nonumber \\
    \leq & \frac{M - 1}{M} \lim_{\kappa \rightarrow 0} \frac{M\kappa}{\alpha} = 0 \nonumber
  \end{align}
  and therefore
  $$\lim_{\kappa \rightarrow 0} \frac{1}{M} F^{-1}_{\chi} (1 - \alpha; M\kappa) = \lim_{\kappa \rightarrow 0} F^{-1}_{\chi} (1 - \alpha; M\kappa)$$
  for any $\alpha > 0$. This implies that $p_c \big ( \chipool{}{\kappa} \big ) = p_r \big ( \chipool{}{\kappa} \big )$ in the limit $\kappa \rightarrow 0$.
\end{proof}

The result of Theorem \ref{thm:gchiSmall} can be understood intuitively using the non-central $\chi^2_0$ of \cite{siegel1979noncentral} with a non-centrality parameter $\lambda$, call it $\chi^2_0(\lambda)$. When $\lambda \rightarrow 0$, $\chi^2_0(\lambda) \rightarrow \chi^2_0$ in distribution, so taking $\chi^2_0(\lambda)$ with small $\lambda$ should provide some sense of how $\chi^2_0$ behaves. Unlike $\chi^2_0$, however, $\chi^2_0(\lambda)$ has a discrete probability mass at 0 for all $\lambda > 0$.
As a result, the quantile function of $\chi^2_0(\lambda)$, $F^{-1}_{\lambda}$, returns zero for any input less than $e^{-\frac{\lambda}{2}}$ and so the terms in the sum
$$\sum_{i = 1}^M F_{\lambda}^{-1}(1 - p_i)$$
are non-zero only for those $i$ where $p_i \leq 1 - e^{-\frac{\lambda}{2}}$. As $\lambda \rightarrow 0$, this sum becomes arbitrarily close to $\chipool{}{0}$ but only the smallest $p$-values contribute. Eventually, only the minimum contributes to the sum, and so $\chipool{}{\kappa} \approx f(p_{(1)})$ for very small $\kappa$ values.

The limits $\lim_{\kappa \rightarrow 0} \chipool{\ve{p}}{\kappa} = \tippool(\ve{p})$ and $\lim_{\kappa \rightarrow \infty} \chipool{\ve{p}}{\kappa} = \stopool(\ve{p})$ are also demonstrated empirically by generating $n_{sim}$ independent realizations of $\ve{p}$ assuming $H_0$ is true. For each vector $\ve{p}_i$, compute $\chipool{\ve{p}_i}{\kappa}$ for a range of $\kappa$, $\stopool(\ve{p}_i)$, and $\tippool(\ve{p}_i)$ and compare $\chipool{\ve{p}_i}{\kappa}$ to the other two pooled $p$-values. Figure \ref{MC:fig:chikTipNorm} shows this pattern for a few $\kappa$ when $M = 5$.

\begin{figure}[htp] 
  \begin{center}
    \begin{tabular}{c}
      \includegraphics[scale = 0.9]{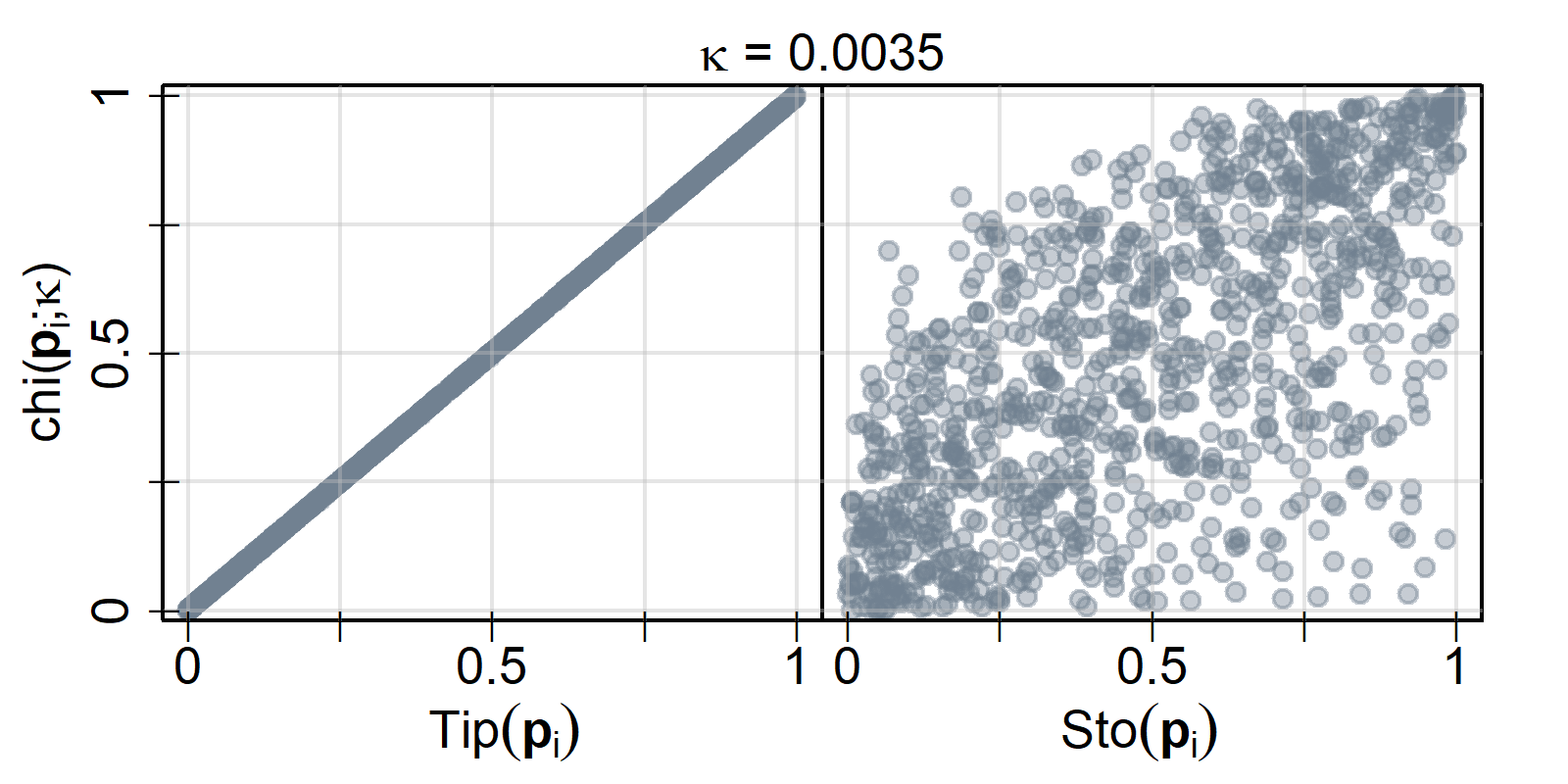} \\
      {\footnotesize (a)} \\
      \includegraphics[scale = 0.9]{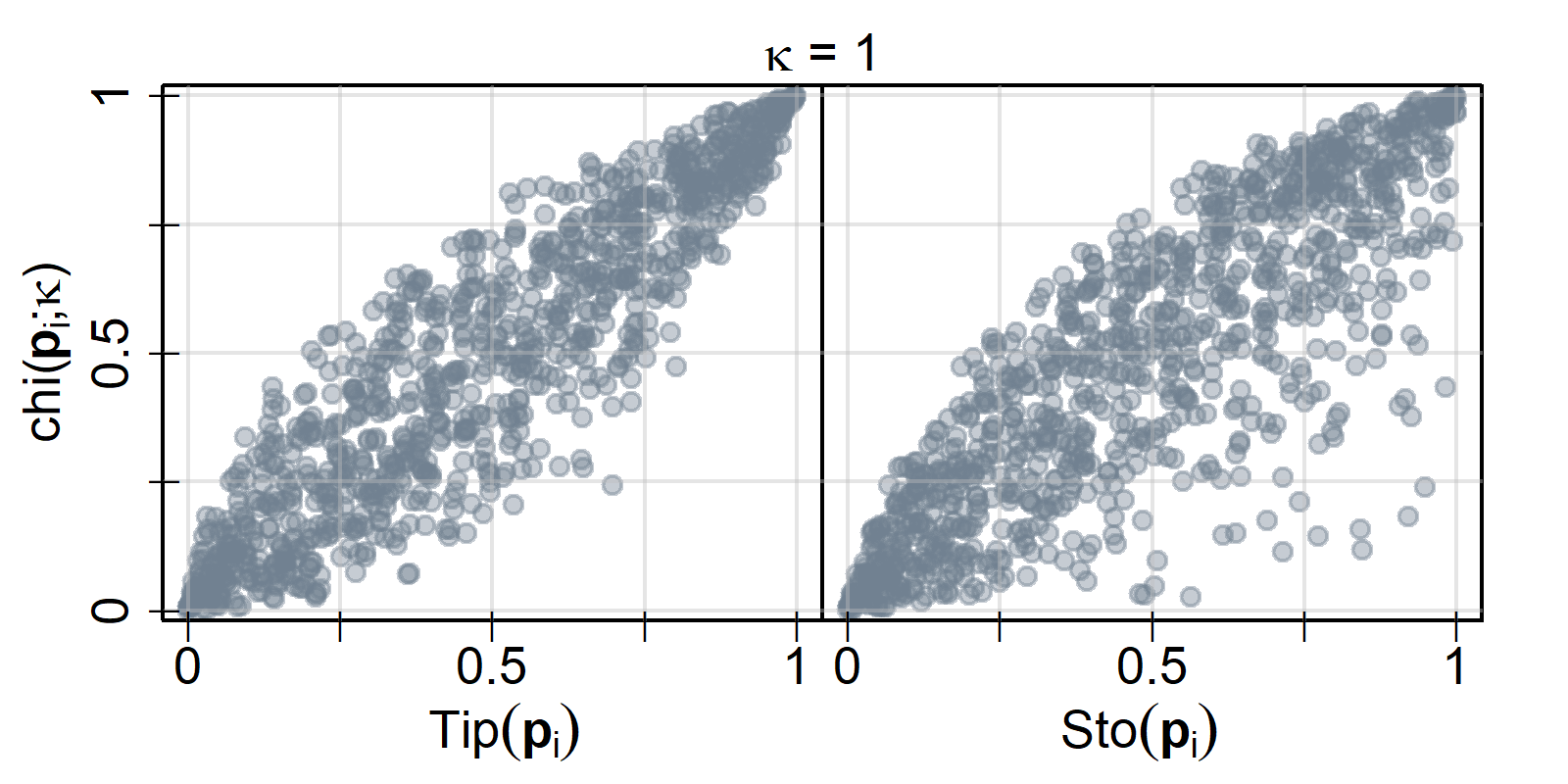} \\
      {\footnotesize (b)} \\
      \includegraphics[scale = 0.9]{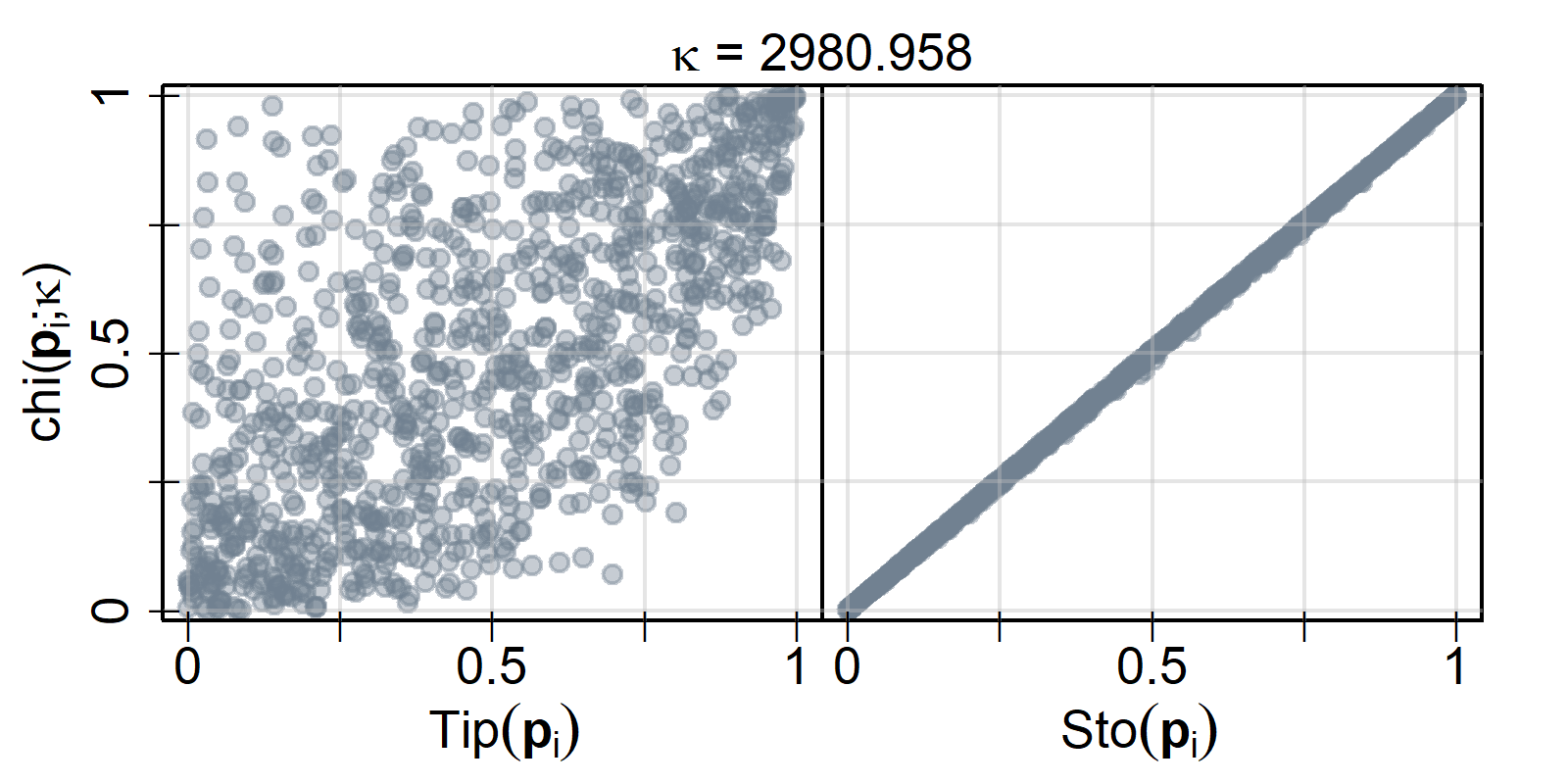} \\ 
      {\footnotesize (c)} \\
    \end{tabular}
    \caption{A comparison of $\chipool{\ve{p}_i}{\kappa}$, $\stopool(\ve{p}_i)$, and $\tippool(\ve{p}_i)$ values for 1000 independently generated $\ve{p}_i \sim Unif([0,1]^5)$ in the case of (a) small $\kappa$, (b) moderate $\kappa$, and (c) large $\kappa$.}
    \label{MC:fig:chikTipNorm}
  \end{center}
\end{figure}

As expected, the agreement between $\chipool{\ve{p}}{\kappa}$ and $\tippool(\ve{p})$ is perfect for small enough $\kappa$, the two functions have identical outputs for $\kappa = 0.0035$ in Figure \ref{MC:fig:chikTipNorm}(a). Similarly, $\chipool{\ve{p}}{\kappa}$ and $\stopool(\ve{p})$ match for large $\kappa$, as when $\kappa \approx 3000$ in Figure \ref{MC:fig:chikTipNorm}(c). Note that the particular values of $\kappa$ where this close agreement occurs will depend on $M$.

Perhaps more interesting is the curved boundary of the points along the top of the plot of $\chipool{\ve{p}}{0.0035}$ against $\stopool(\ve{p})$ in Figure \ref{MC:fig:chikTipNorm}(a), many points populate the lower right corner of this plot but there are none in the upper left. This pattern is mirrored in Figure \ref{MC:fig:chikTipNorm}(c) for the plot of $\chipool{\ve{p}}{2981}$ against $\tippool(\ve{p})$. As $\chipool{\ve{p}}{\kappa}$ is essentially identical to one of $\tippool(\ve{p})$ or $\stopool(\ve{p})$ in these cases, this pattern reflects the relationship between $\tippool(\ve{p})$ and $\stopool(\ve{p})$. By definition, $\tippool(\ve{p})$ considers only $p_{(1)}$, but any of the $p_i$ can impact $\stopool(\ve{p})$. As a result there will be many cases where a small $\tippool(\ve{p})$ occurs despite a large $\stopool(\ve{p})$ because a very small $p_{(1)}$ happens by chance. The reverse is impossible, if $\tippool(\ve{p})$ is large then $p_{(1)}$ is large and therefore all values in $\ve{p}$ are large, suggesting a large $\stopool(\ve{p})$.

\subsection{Choosing a parameter}

In addition to these meaningful limits, there seems to be a monotonically increasing relationship between $\kappa$ and $\centquot(\chipool{}{\kappa})$. Let $\chi^*_{\kappa}(\alpha)$ be the $1 - \alpha$ quantile of the $\chi^2$ distribution with $\kappa$ degrees of freedom, then $\chipool{\ve{p}}{\kappa}$ has the central rejection level
\begin{equation} \label{MC:eq:chiMethPc}
p_c \big ( \chipool{}{\kappa} \big )= 1 - F_{\chi} \left ( \frac{1}{M} F_{\chi}^{-1}(1 - \alpha; M\kappa ); \kappa \right ) = P \left (\chi^2_{\kappa} \geq \frac{1}{M} \chi^*_{M\kappa}(\alpha) \right )
\end{equation}
and the marginal rejection level
\begin{equation} \label{MC:eq:chiMethPr}
p_r \big ( \chipool{}{\kappa} \big ) = 1 - F_{\chi} \bigg ( F_{\chi}^{-1}(1 - \alpha; M\kappa ); \kappa \bigg ) = P \bigg (\chi^2_{\kappa} \geq \chi^*_{M\kappa}(\alpha) \bigg ),
\end{equation}
implying
\begin{eqnarray} 
  \centquot(\chipool{}{\kappa}) & = & \frac{p_c \big ( \chipool{}{\kappa} \big ) - p_r \big ( \chipool{}{\kappa} \big )}{p_c \big ( \chipool{}{\kappa} \big )} \nonumber \\
  & & \nonumber \\
  & = & P \left ( \chi^2_{\kappa} \leq \chi^*_{M\kappa}(\alpha) \biggiven \chi^2_{\kappa} \geq \frac{1}{M} \chi^*_{M\kappa}(\alpha) \right ). \label{MC:eq:chiMethC}
\end{eqnarray}
That is, the centrality quotient of $\chipool{\ve{p}}{\kappa}$ is the conditional probability that a $\chi^2_{\kappa}$ random variable is less than $\chi^*_{M\kappa}(\alpha)$ given that it is greater than $\frac{1}{M} \chi^*_{M\kappa}(\alpha)$.

A better sense of the region corresponding to this conditional probability for $\alpha < 0.5$ is garnered by writing $\chi^*_{M\kappa}$ in terms of the mean of the $\chi^2_{M \kappa}$ distribution, $M\kappa$. Taking an arbitrary remainder function $R_{M\kappa}(\alpha) > 0$ such that $\chi^*_{M\kappa}(\alpha) := M\kappa + R_{M\kappa}(\alpha)$, subsitution gives
$$\centquot(\chipool{}{\kappa}) = P \left ( \chi^2_{\kappa} \leq M\kappa + R_{M\kappa}(\alpha) \biggiven \chi^2_{\kappa} \geq \kappa + \frac{1}{M}R_{M\kappa}(\alpha) \right ),$$
clarifying that $\centquot(\chipool{}{\kappa})$ is a conditional probability on the right tail of the $\chi^2_{\kappa}$ distribution when $\alpha < 0.5$. Making more precise statements about $R_{M\kappa}(\alpha)$ is challenging due to the small values of $\kappa$ which may be chosen for $\chipool{}{\kappa}$. Most approximations of $\chi^2$ tail probabilities and quantiles either break down when the degrees of freedom is less than one or explicitly assume more than one degrees of freedom (\citealp{hawkinswixley1986chitrans, canal2005normal, inglot2010inequalities}).
Nonetheless, the above probability can be computed numerically, as was done for the curves of $\centquot(\chipool{}{\kappa})$ by $\log_{10} ( \kappa )$ for $M$ ranging from 2 to 10,000 in Figure \ref{MC:fig:cgChi}.

\begin{figure}[h!]
  \begin{center}
    \includegraphics{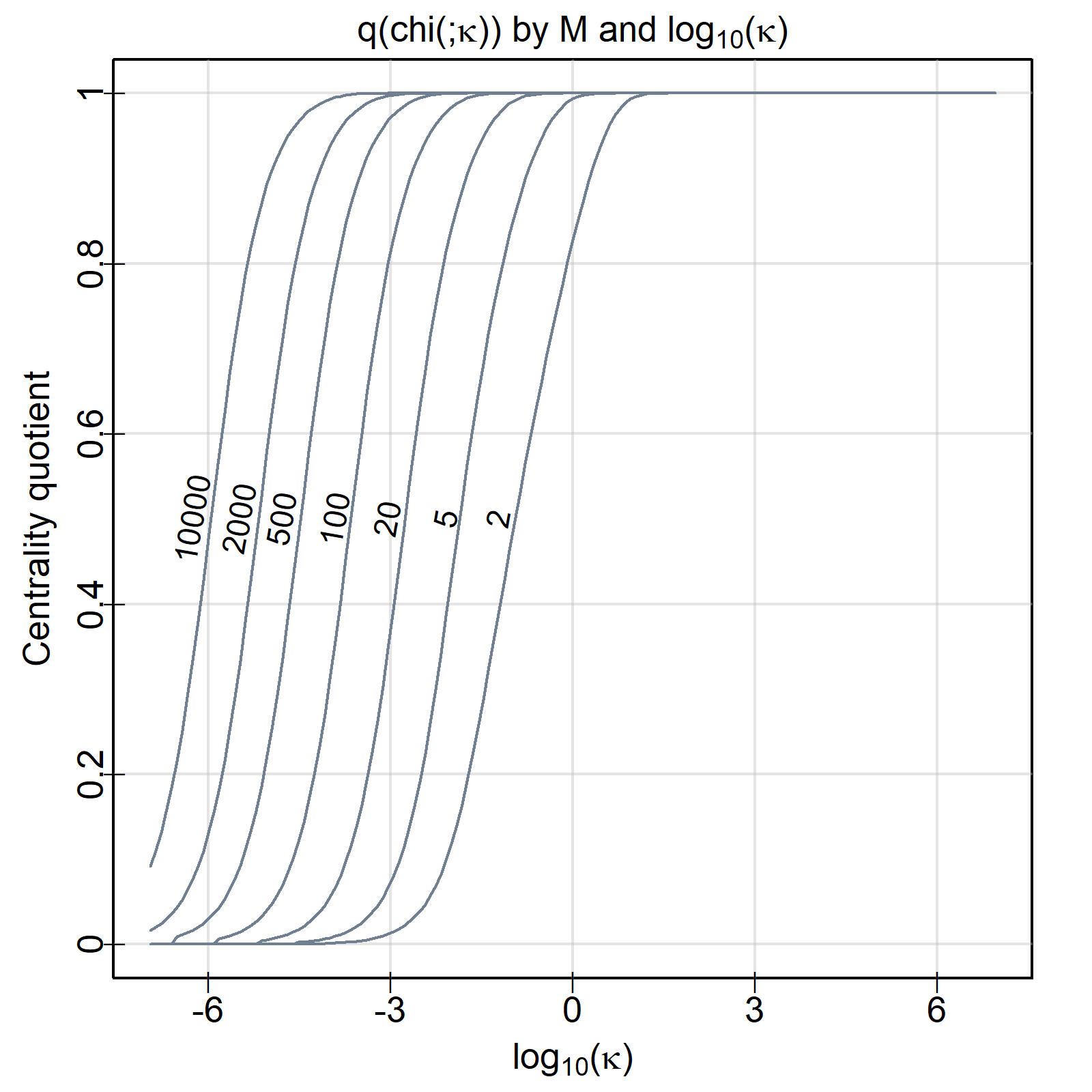}
  \end{center}
  \caption{The centrality quotient of $\chipool{}{\kappa}$ by $\log_{10}(\kappa)$}
  \label{MC:fig:cgChi}
\end{figure}

The curves of $\centquot(\chipool{}{\kappa})$ by $\kappa$ have a consistent sigmoid shape for all $M$. Most of the change in the centrality quotient occurs for values in a three unit range in $\log_{10}( \kappa )$ for any $M$, though the centre of this range decreases as $M$ grows. When $\kappa = 10^{-3}$, for example, the centrality quotient when $M = 100$ is greater than 0.8 while the same $\kappa$ value corresponds with a centrality quotient of less than 0.05 when $M = 2$. Just as with any other pooled $p$-value, increasing $M$ increases the centrality of $\chipool{}{\kappa}$ for a given $\kappa$ as the sum of independent $p$-values becomes more normally distributed by the central limit theorem.\footnote{This can also be understood geometrically. For a pooled $p$-value in $M$ dimensions, the volume of the marginal shell of width $p_r$ is $1 - (1 - p_r)^M$, which approaches 1 for any $p_r > 0$ as $M \rightarrow \infty$. As the total volume of the rejection region is $\alpha$ for the rejection rule $g(\ve{p}) \leq \alpha$, $p_r$ must decrease in $M$ to hold the volume constant.}

In practice, the inverse of the above curves may be of greater interest to control the centrality quotient under $\chipool{\ve{p}}{\kappa}$ rather than simply report it. Figure \ref{MC:fig:cgChi} does allow the selection of $\kappa$ for a given centrality quotient by estimating the $\kappa$ value where the intersection between a curve and a vertical line at $\kappa$ is at the desired quotient, but a table displaying the numerically estimated inverse for evenly-spaced $\kappa$ as in Table \ref{MC:tab:kappaByM} is more precise and straightforward to use. Determining the desired $\log_{10} ( \kappa )$ for a given centrality quotient and $M$ proceeds as for a table of critical values. The user searches down the columns for the $M$ most closely corresponding to the setting at hand, and then searches through that row for the desired column. If $\centquot(\chipool{}{\kappa}) = 1$ or $0$ is desired, the table is unnecessary as $\stopool(\ve{p})$ or $\tippool(\ve{p})$ can be used directly. Unlike with critical value tables, there is no need to be conservative: linear interpolation between the provided $\log_{10}(\kappa)$ values is a reasonable approach to choosing $\kappa$.

\begin{table}[!ht]
  \begin{center}
    \begin{tabular}{|c|rrrrrrrrr|}
      \hline  & \multicolumn{9}{c|}{Centrality quotient} \\
      $M$ & \multicolumn{1}{l}{0.1} &  \multicolumn{1}{l}{0.2} &  \multicolumn{1}{l}{0.3} &  \multicolumn{1}{l}{0.4} &  \multicolumn{1}{l}{0.5} &  \multicolumn{1}{l}{0.6} &  \multicolumn{1}{l}{0.7} &  \multicolumn{1}{l}{0.8} &  \multicolumn{1}{l|}{0.9} \\ \hline
      2 & -2.1 & -1.7 & -1.4 & -1.2 & -0.9 & -0.7 & -0.4 & -0.1 & 0.3 \\
      5 & -2.8 & -2.5 & -2.3 & -2.1 & -1.9 & -1.7 & -1.4 & -1.2 & -0.8 \\
      20 & -3.7 & -3.4 & -3.1 & -2.9 & -2.8 & -2.6 & -2.4 & -2.1 & -1.8 \\
      100 & -4.6 & -4.3 & -4 & -3.8 & -3.7 & -3.5 & -3.3 & -3 & -2.7 \\
      500 & -5.4 & -5.1 & -4.8 & -4.7 & -4.5 & -4.3 & -4.1 & -3.9 & -3.5 \\
      2000 & -6.1 & -5.8 & -5.5 & -5.3 & -5.2 & -5 & -4.8 & -4.6 & -4.2 \\
      10000 & -6.9 & -6.6 & -6.3 & -6.1 & -6 & -5.8 & -5.6 & -5.4 & -5 \\ \hline
    \end{tabular}
  \end{center}
  \caption{$\log_{10} (\kappa)$ values by $\centquot(\chipool{}{\kappa})$ and $M$ to aid in parameter selection for the desired balance of central and marginal rejection.}
  \label{MC:tab:kappaByM}
\end{table}

Using the parameter $\kappa$ of $\chipool{\ve{p}}{\kappa}$, the relative preference of $\chipool{\ve{p}}{\kappa}$ to rejection along the margins or in the centre can be directly controlled. Large $\kappa$ produce a pooled $p$-value which is powerful at detecting evidence spread among all tests, while small $\kappa$ favour the detection of concentrated evidence in a single test with extremes giving the widely-used $\tippool(\ve{p})$ and $\stopool(\ve{p})$. The parameter $\kappa$ orders pooled $p$-values of the $\chipool{\ve{p}}{\kappa}$ family by relative centrality, simplifying the choice of pooled $p$-value and communication of results. Finally, as it is based on Equation (\ref{MC:eq:quantilePval}), it is an exact quantile-based method which does not rely on asymptotic behaviour and which could, hypothetically, be computed by hand with the aid of $\chi^2$ quantile tables.

\subsection{Comparing the chi-squared pooled $p$-value to the UMP benchmark} \label{chipool:chipoolvsHR}

Recall the simulation studies that motivated the exploration of central and marginal rejection levels. After a benchmark power computation, the power of $\hrpool{\ve{p}}{w}$ for $\alpha = 0.05$ was evaluated under $H_4$ for a range of beta alternatives ($f = Beta(a, 1/\omega + a(1 - 1/\omega))$) with KL divergences from uniform ($D(a,\omega)$) spanning $e^{-5}$ to $e^5$. Correct specification of $w$ was important: the larger the magnitude of $w - \omega$, the larger the decrease in power of $\hrpool{\ve{p}}{w}$ from $\hrpool{\ve{p}}{\omega}$. Under $H_3$, mis-specification did not matter at all, the power of $\hrpool{\ve{p}}{w}$ was dictated by the the proportion of false null hypotheses ($\prevalence$) and the strength of evidence against $H_0$ in each non-null hypothesis ($D(a,\omega)$). The parameter $w$ tunes $\hrpool{\ve{p}}{w}$ to favour either weak evidence spread among all tests, or strong evidence in only a few.

Finer selection of this tradeoff is achieved with the parameter $\kappa$ using the $\chipool{\ve{p}}{\kappa}$ family of pooled $p$-values, but controlling centrality is of little use if $\chipool{\ve{p}}{\kappa}$ is not powerful under the settings that motivated their definition. The power of $\chipool{\ve{p}}{\kappa}$ at level $\alpha = 0.05$ for each $\kappa \in \{e^{-8}, e^{-4}, 1, 2, e^4, e^8\}$ was therefore determined under every setting from Section \ref{sec:power} using the same simulated samples generated under $H_4$. Prior to the simulation, it is expected is that large $\kappa$ will be uniformly more powerful than small $\kappa$, as under $H_4$ evidence is spread among all tests. The results confirmed this expectation: the most powerful $\chipool{\ve{p}}{\kappa}$ for all settings under $H_4$ was $\chipool{\ve{p}}{e^{8}} \approx \chipool{\ve{p}}{2981}$. It is compared to the both the UMP and mis-specified $\hrpool{\ve{p}}{w}$ in Figure \ref{fig:chiComparedMis} by adding a dark grey line to Figure \ref{fig:MisPowers}.

\begin{figure}[!ht] 
  \begin{center}
      \includegraphics{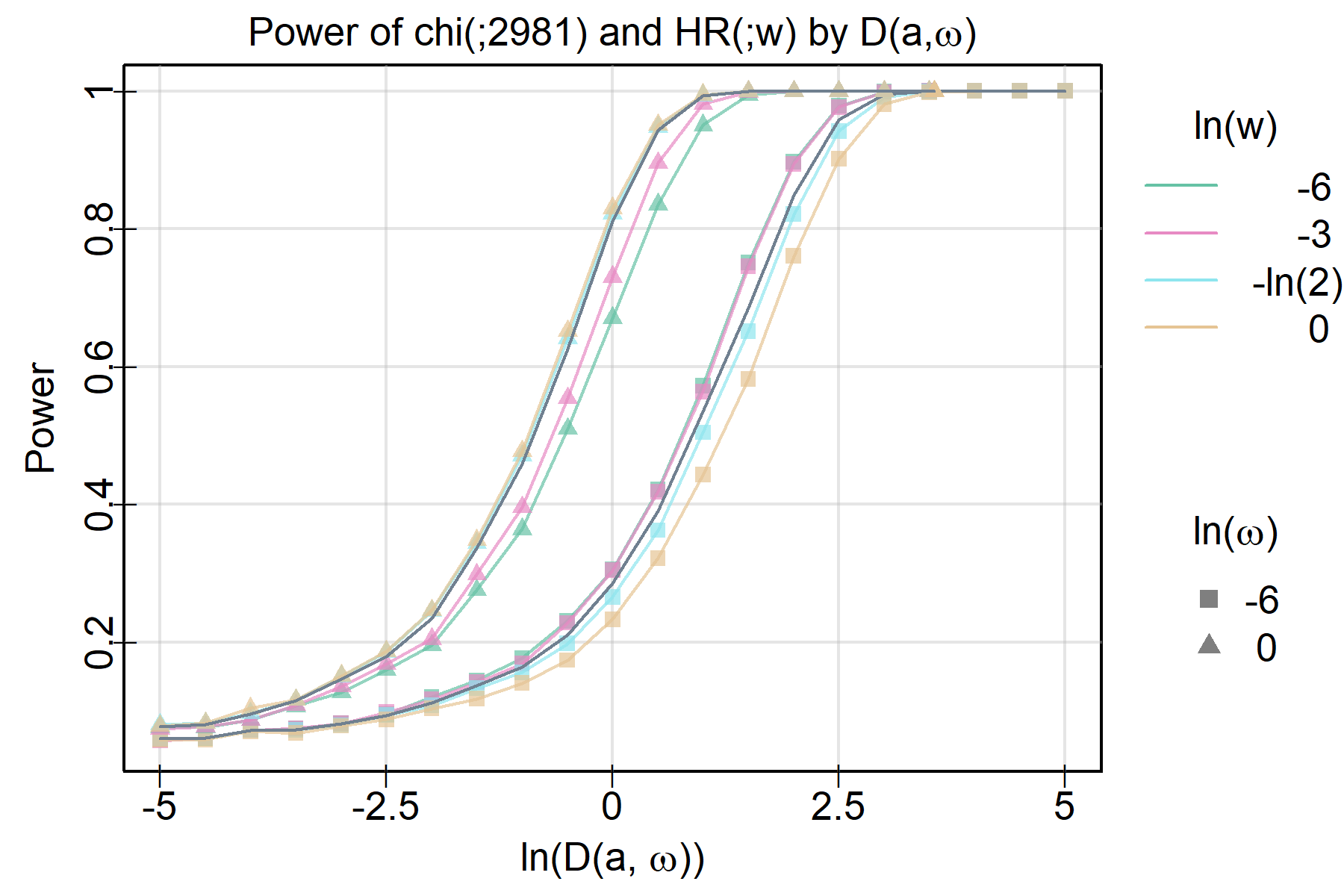}
    \caption{A comparison of $\chipool{\ve{p}}{2981}$ to the UMP and mis-specified $\hrpool{\ve{p}}{w}$ under $H_4$. $\chipool{\ve{p}}{2981}$ nearly UMP power more consistently than any $\hrpool{\ve{p}}{w}$, and so is more robust to $f$.}
    \label{fig:chiComparedMis}
  \end{center}
\end{figure}

$\chipool{\ve{p}}{2981}$ has higher power than most mis-specified $\hrpool{\ve{p}}{w}$ for all settings in this case and is close to the UMP more consistently than any $\hrpool{\ve{p}}{w}$. Only when $w \approx \omega$ does $\hrpool{\ve{p}}{w}$ beat $\chipool{\ve{p}}{2981}$, and so it is less robust to mis-specification of $f$ than $\chipool{\ve{p}}{2981}$. It may therefore be advisable to use $\chipool{\ve{p}}{\kappa}$ with a large $\kappa$ (or simply $\stopool(\ve{p})$) when testing $H_4$ with beta alternatives in the case where $\omega$ is not known, rather than risk the penalty of choosing $w$ wrong when using $\hrpool{\ve{p}}{w}$. This is despite the fact that $\hrpool{\ve{p}}{\omega}$ is UMP for this setting.

For the case where $\ve{p}$ was generated under $H_3$, $\chipool{\ve{p}}{\kappa}$ was again computed for each $\kappa \in \{e^{-8}, e^{-4}, 1, 2, e^4, e^8\}$ over the 10,000 independent samples for each setting of $D(a,\omega)$, $\omega$, and $\prevalence$ with $M = 10$ from Section \ref{sec:power}. Contour plots analogous to Figure \ref{fig:H3DiffConts} showing the differences in power between $\chipool{\ve{p}}{\kappa}$ and $\hrpool{\ve{p}}{1} = \chipool{\ve{p}}{2}$ were generated. The reference $\hrpool{\ve{p}}{1}$ was chosen because it is a test shared by both the $\chipool{\ve{p}}{\kappa}$ and $\hrpool{\ve{p}}{w}$ families.

\begin{figure}[!ht]
  \begin{center}
    \begin{tabular}{c}
      \begin{tabular}{c} \includegraphics{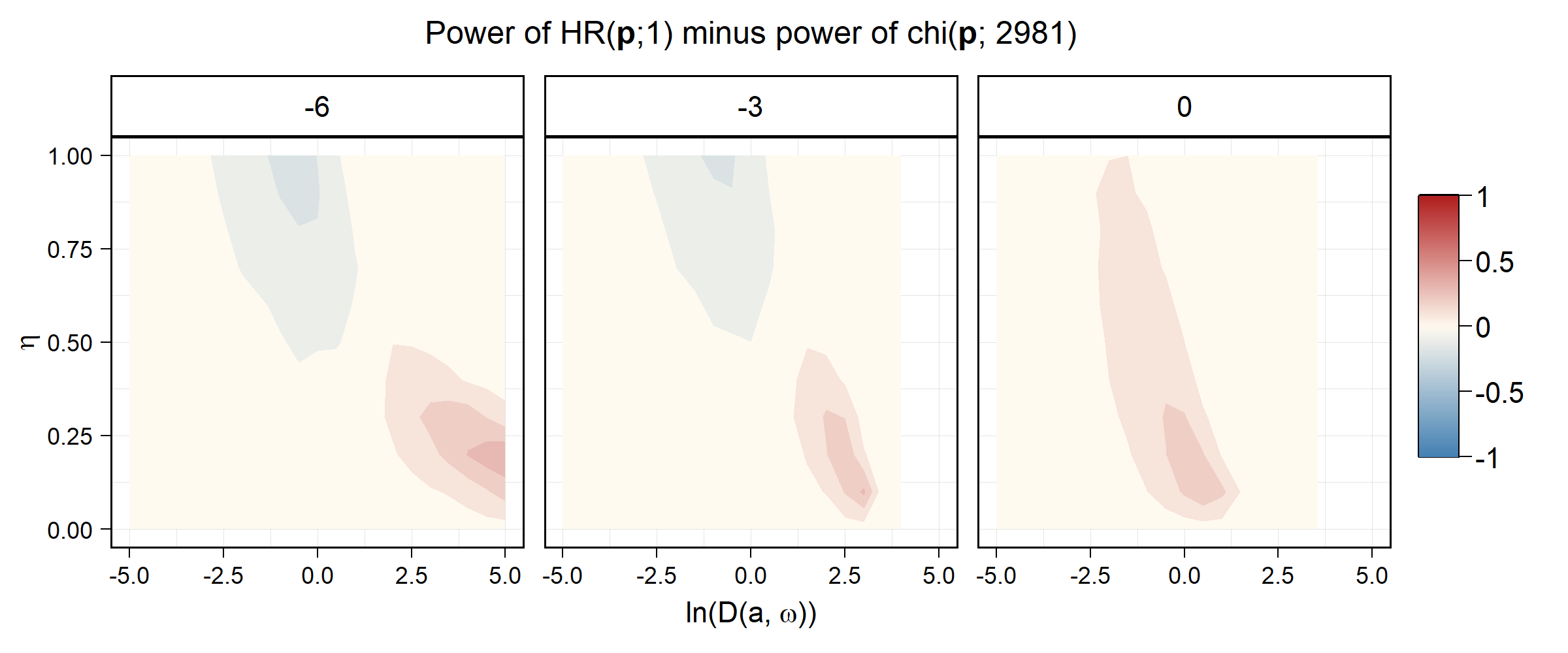} \end{tabular} \\
      {\footnotesize (a)} \\
      \begin{tabular}{c} \includegraphics{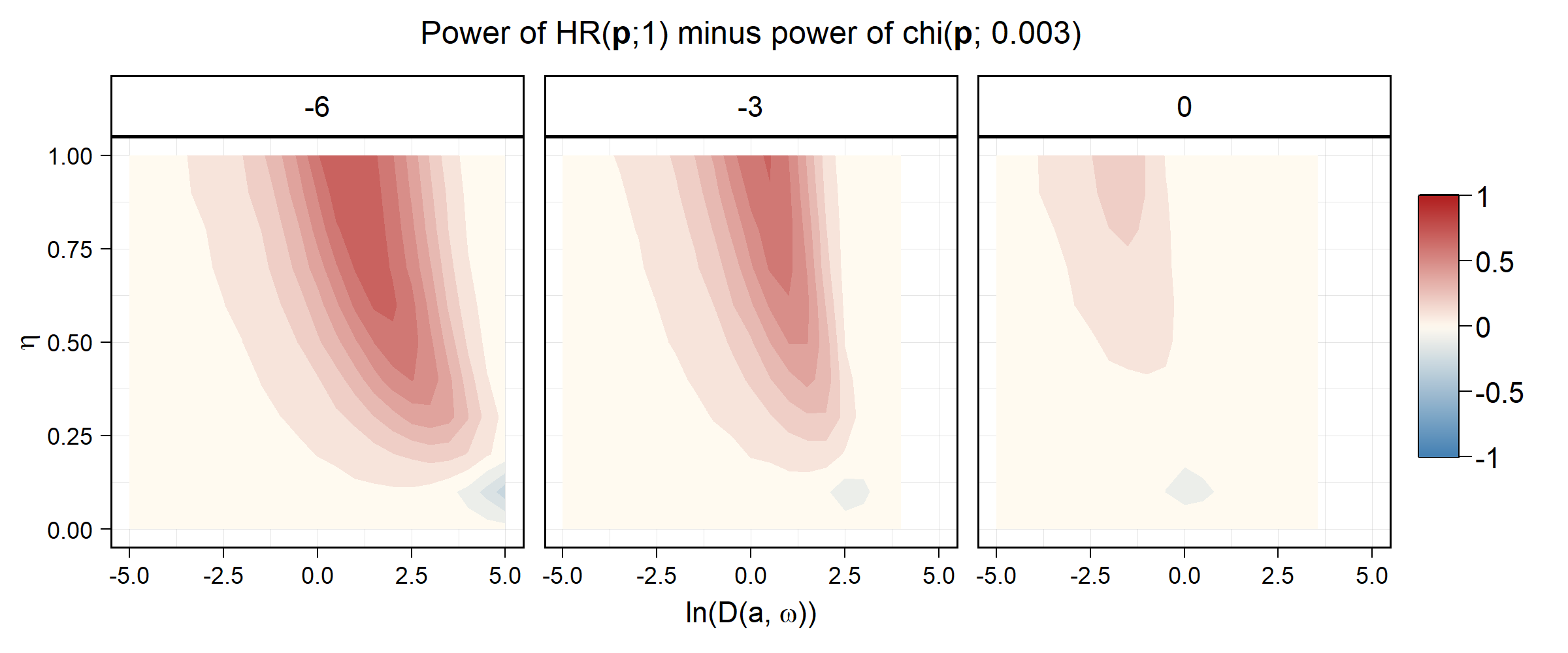} \end{tabular} \\
      {\footnotesize (b)} \\
    \end{tabular}
    \caption{Contours for the power of $\hrpool{\ve{p}}{1} = \chipool{\ve{p}}{2}$ minus (a) $\chipool{\ve{p}}{e^8} \approx \chipool{\ve{p}}{2981}$ and (b) $\chipool{\ve{p}}{e^{-8}} \approx \chipool{\ve{p}}{0.003}$ by $\prevalence$ and $D(a, \omega)$ facetted by $\omega$. Compared to Figure \ref{fig:H3DiffConts}, (a) displays less of a penalty for the case of concentrated evidence while still outperforming $\hrpool{\ve{p}}{1}$ for evidence spread among all tests.}
  \label{fig:BothChiDiffs}
  \end{center}
\end{figure}

The patterns of power for $\chipool{\ve{p}}{\kappa}$ mimic those of $\hrpool{\ve{p}}{w}$: large $\kappa$ favour evidence spread among all tests as do small $w$ in $\hrpool{\ve{p}}{w}$. Despite this similarity, $\chipool{\ve{p}}{2981}$ has higher power when applied to the case of concentrated evidence and so is more robust under $H_3$. This is seen clearly in a comparison of the bottom right corner of Figure \ref{fig:H3DiffConts} to Figure \ref{fig:BothChiDiffs}(a), the former shows a much larger and darker red region than the latter.

The $\chipool{\ve{p}}{\kappa}$ family also extends the range of possible centrality parameters compared to $\hrpool{\ve{p}}{w}$. As $\hrpool{\ve{p}}{1} = \chipool{\ve{p}}{2}$ is one of the boundaries of the $w$ parameter range, no comparable pooled $p$-values to $\chipool{\ve{p}}{\kappa}$ for $\kappa < 2$ exist in the $\hrpool{\ve{p}}{w}$ family. Using $\chipool{\ve{p}}{\kappa}$ therefore gives greater control over the balance of central and marginal rejection than $\hrpool{\ve{p}}{w}$, though it seems exceptionally small $\kappa$ in $\chipool{\ve{p}}{\kappa}$, or equivalently $\tippool(\ve{p})$, should only be used sparingly. Figure \ref{fig:BothChiDiffs}(b) shows that $\chipool{\ve{p}}{0.003}$ loses power almost everywhere compared to $\chipool{\ve{p}}{2}$ in exchange for higher power only in the case of extreme evidence in a single test. Under $H_3$, very small values of $\kappa$ should probably only be used if such a pattern of evidence is strongly suspected.

The $\chipool{\ve{p}}{\kappa}$ family is therefore of interest both practically and theoretically. It provides control over central and marginal rejection under $H_3$ and robustly gives nearly UMP power for large values of $\kappa$ under $H_4$. It has interpretable endpoints which cover a greater range of centrality quotients than $\hrpool{\ve{p}}{w}$ and gives a means of controlling the bias towards central rejection present in all quantile pooled $p$-values as $M$ increases. $\chipool{\ve{p}}{\kappa}$ is a pooled $p$-value with great potential as a practical tool for controlling the FWER when testing $H_0$.

\section{Identifying plausible alternative hypotheses and selecting tests} \label{sec:chiIdentifying}

The link between $\kappa$, the centrality quotient, and relative power in regions of the $D(a, w), \prevalence$ plane under $H_3$ can be exploited to identify alternatives to $H_0$ that could have plausibly generated $\ve{p}$. Rather than selecting a particular $\kappa$ value, can consider all possible $\kappa$ values simultaneously, compute $\chipool{\ve{p}}{\kappa}$ for each, and record
\begin{equation} \label{eq:kappaMinStat}
  \kappa_{\min} = \argmin_{\kappa \in [0, \infty)} \chipool{\ve{p}}{\kappa}
\end{equation}
As each $\kappa$ value is associated with a particular centrality quotient, each $\kappa$ identifies a particular region of relative power against others in the $D(a,w), \prevalence$ plane under $H_3$. At the same time, $\kappa_{\min}$ reports the value of $\kappa$ which produces the smallest pooled $p$-value for $\ve{p}$ and therefore suggests the $\kappa$ value where evidence against $H_0$ is the strongest relative to other $\kappa$ values. As stronger evidence leads to more frequent rejection and higher power when $H_0$ is false, $\kappa_{\min}$ therefore links the evidence present in $\ve{p}$ to a region in $D(a,w), \prevalence$ if we assume $H_3$ is truly used to generate the data with $f = Beta(a, 1/w + a(1 - 1/w))$.

\subsection{Non-increasing beta densities} \label{chipool:noninc}

\begin{figure}[!h] 
  \begin{center}
    \begin{tabular}{cc}
      \includegraphics[scale = 1]{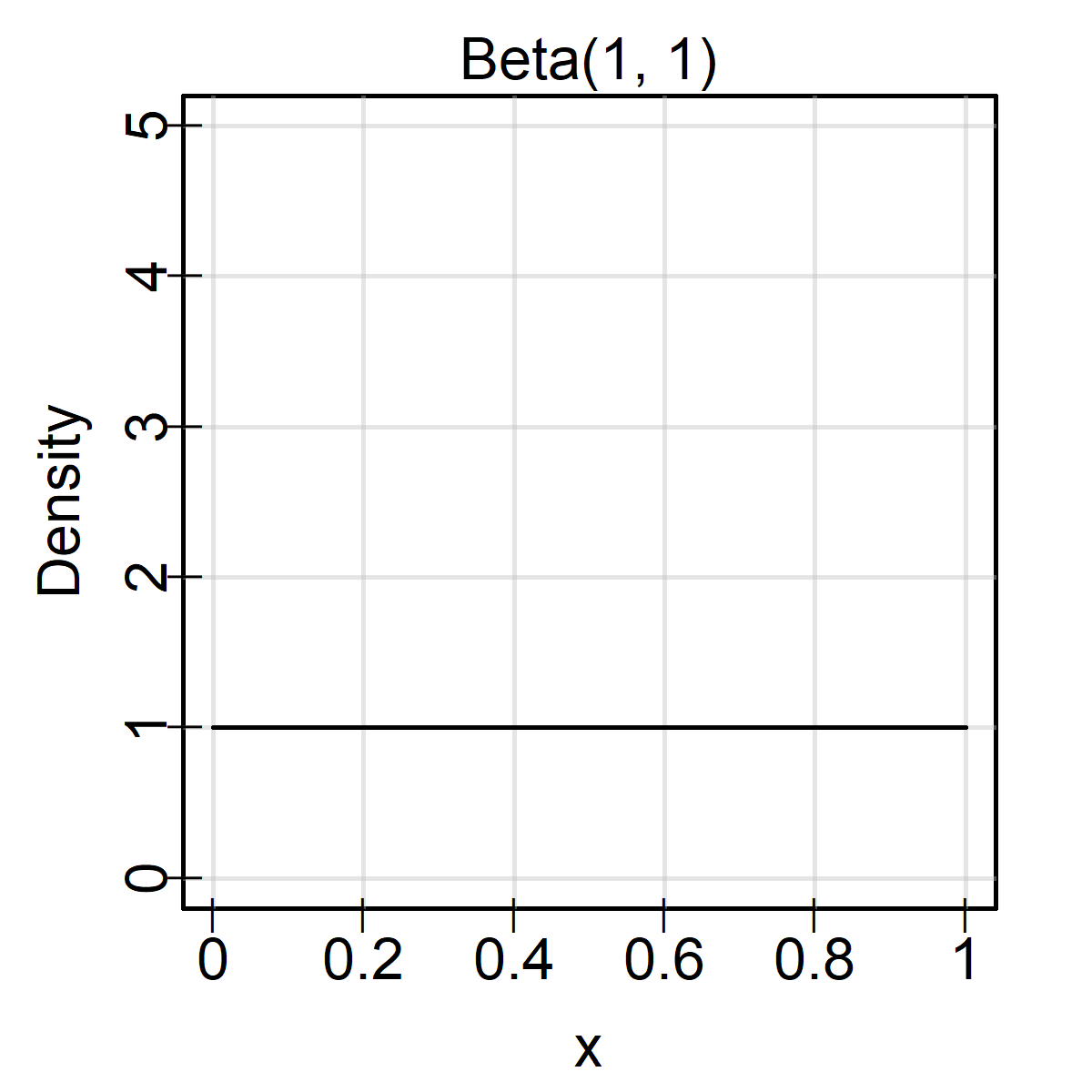} & \includegraphics[scale = 1]{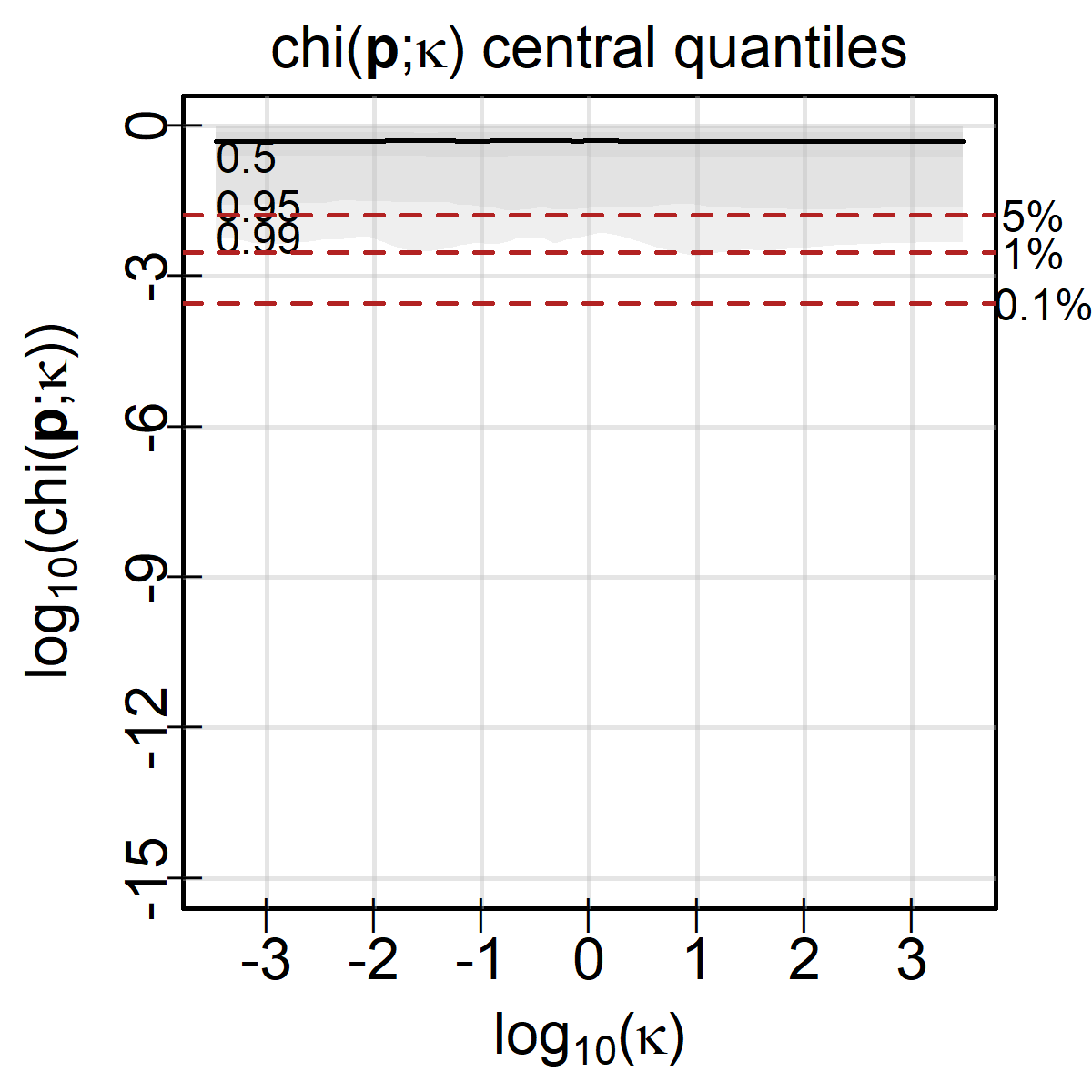} \\
      {\footnotesize (a)} & {\footnotesize (b)}
    \end{tabular}
    \caption{The (a) density and (b) central quantiles and median by $\kappa$ for the null case. All lines are flat and above the null quantiles from the larger simulation as expected.}
    \label{MC:fig:nullBetaQ}
  \end{center}
\end{figure}

Begin with a demonstration of the sweep of $\kappa$ values for previously explored cases by generating curves of $\chipool{}{\kappa}$ by $\kappa$ for different densities under $H_4$. In each of the following, samples of 100 i.i.d. $p$-values from different beta distributions are generated independently 1,000 times and $\chipool{}{\kappa}$ is computed for a sequence of $\kappa$ values chosen uniformly on the log scale. Let the $i^{\text{th}}$ sample be $\ve{p}_i$ and the pooled $p$-value computed using parameter $\kappa_j$ for $\ve{p}_i$ be $\chi_{ij} = \chipool{\ve{p}_i}{\kappa_j}$. To provide context to $\chi_{ij}$, a larger simulation of 100,000 samples was generated under $H_0$ and the minimum of $\chipool{\ve{p}_i}{\kappa_j}$ for the same sequence of $\kappa_j$ values was recorded. Figure \ref{MC:fig:nullBetaQ}(b) displays the median and 0.5, 0.95, and 0.99 central quantiles for $f = Beta(1, 1)$ (equivalent to the null case) alongside the $Beta(1,1)$ density in \ref{MC:fig:nullBetaQ}(a). For reference three dashed red lines at the observed 0.05, 0.01, and 0.001 quantiles of the minimum pooled $p$-value over the 100,000 simulated null cases have been added.

This case shows a flat median curve and flat central quantiles which are all slightly above the corresponding minimum quantiles. The null case performs as expected, $\kappa_{\min}$ would is distributed uniformly over the range of $\kappa$ values and would produce values below the null quantiles at the expected proportions. A contrasting case in shown in Figure \ref{MC:fig:decrBetaQ}, which uses the same layout for an identical simulation carried out when $f = Beta(0.5, 1)$.

\begin{figure}[!h] 
  \begin{center}
    \begin{tabular}{cc}
      \includegraphics[scale = 1]{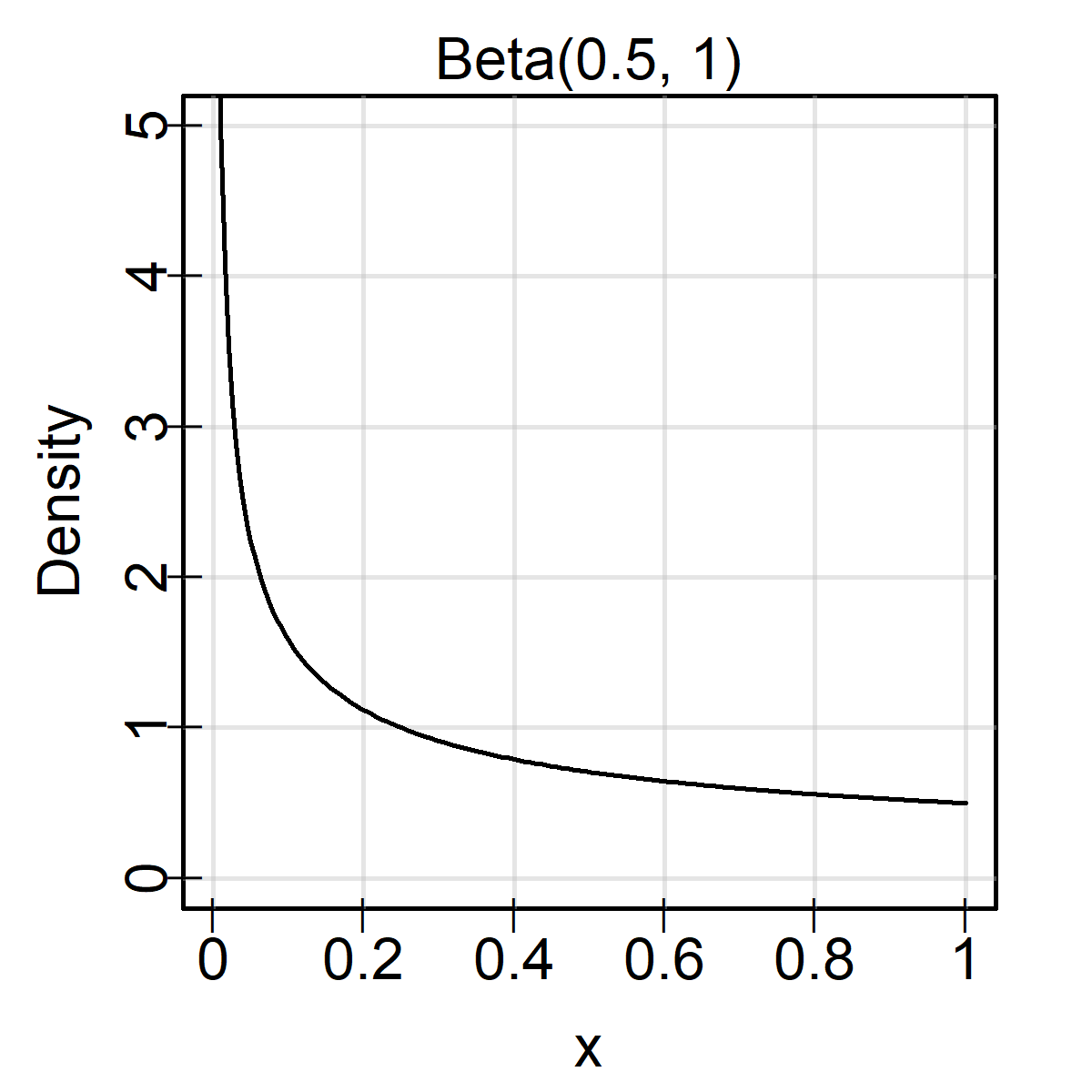} & \includegraphics[scale = 1]{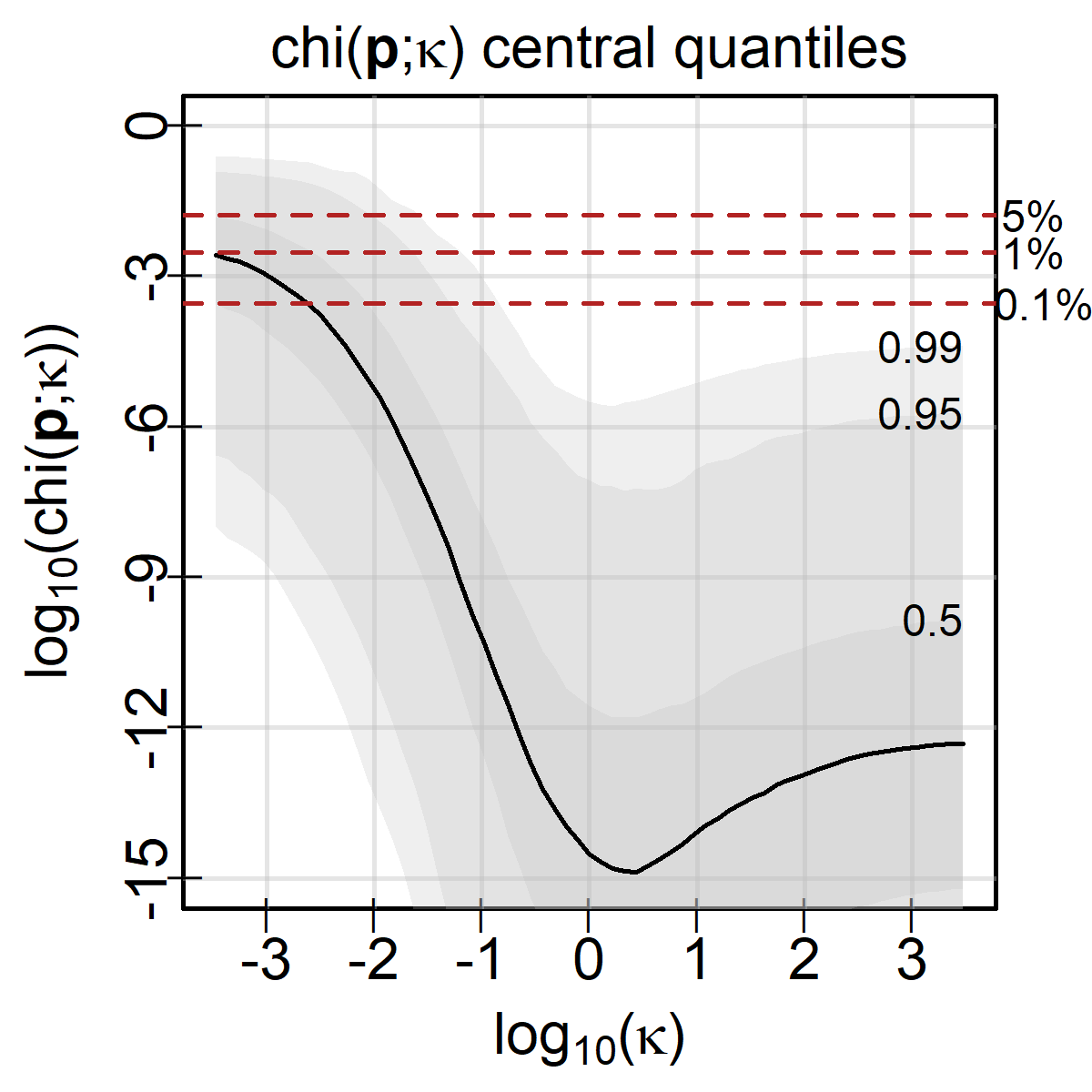} \\
      {\footnotesize a)} & {\footnotesize b)}
    \end{tabular}
    \caption{The (a) density and (b) central quantiles and median by $\kappa$ for $p$-values generated identically and independently from a Beta(0.5, 1) distribution. The minimum around $\kappa = 2$ corresponds to the UMP.}
    \label{MC:fig:decrBetaQ}
  \end{center}
\end{figure}

Displaying the median and the same central quantiles as before, there is a unique minimum at $\kappa = 2$, a lower right end to the curve than the left end, and a generally lower value across its entire length. If one of these curves was observed in practice, $\kappa_{\min} \approx 2$ would be chosen and larger $\kappa$ values may not be fully ruled out. This conclusion would be correct: under $H_4$ with $f = Beta(0.5, 1)$ the UMP pooled $p$-value is $\hrpool{\ve{p}}{w}$ with $w = (1 - a)/(b - a) = 1$ and $\hrpool{\ve{p}}{1} = \chipool{\ve{p}}{2} = \fispool(\ve{p})$. Furthermore, the power investigations in Section \ref{chipool:chipoolvsHR} demonstrate that $\chipool{\ve{p}}{2981} \approx \stopool(\ve{p})$ is nearly as powerful as the UMP for all $w$ under $H_4$. This confirms empirically that the level of the curve of $\chipool{\ve{p}}{\kappa}$ over $\kappa$ corresponds to the relative power of pooled $p$-values in $\chipool{\ve{p}}{\kappa}$ for this beta distribution.

Of course, this conclusion should be expanded to $H_3$ and so mixture of $Beta(0.1, 1)$ and $Beta(1, 1)$ distributions is considered. The first distribution provides strong evidence against the null hypothesis while the second corresponds to the uniform distribution, and so contains no evidence against the null. Mixing these such that the probability of drawing from $Beta(0.1,1)$ is 0.05 and the probability of drawing from the null is 0.95 we are placed in the $D(a,w), \prevalence$ space at $0.3, 0.05$. Simulating this as for the null and $Beta(0.5, 1)$ cases and displaying the central quantiles of $\chi_{ij}$ alongside the mixture density gives Figure \ref{MC:fig:mixBetaQ}.

\begin{figure}[!h] 
  \begin{center}
    \begin{tabular}{cc}
      \includegraphics[scale = 1]{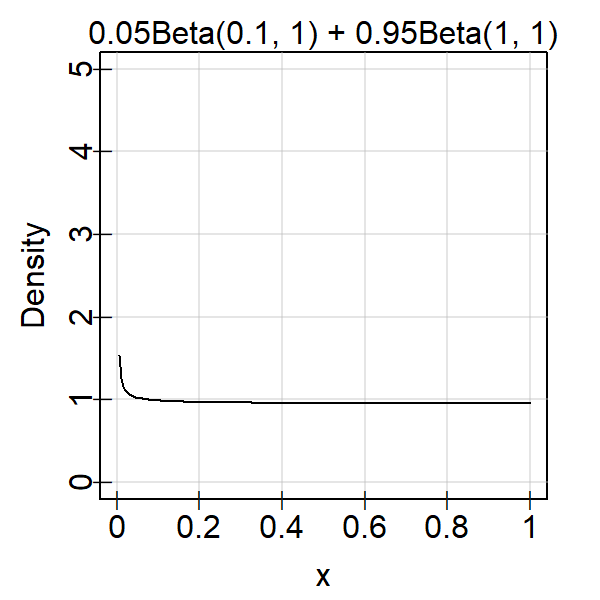} & \includegraphics[scale = 1]{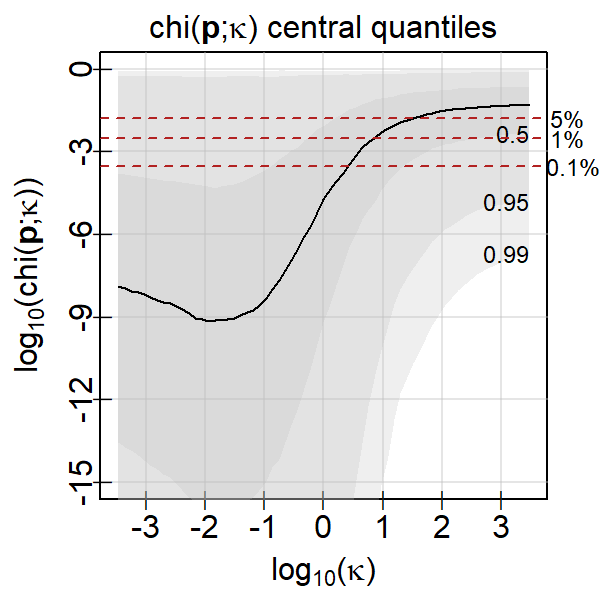} \\
      {\footnotesize (a)} & {\footnotesize (b)}
    \end{tabular}
    \caption{The (a) density and (b) central quantiles and median by $\kappa$ for the mixture 0.05Beta(0.1, 1) + 0.95Beta(1, 1). Small $\kappa$ values provide the smallest pooled $p$-values, and hence power at detecting this alternative.}
    \label{MC:fig:mixBetaQ}
  \end{center}
\end{figure}

The central quantiles are more variable for this case than the unmixed densities because the probability of generating any $p$-values from $Beta(0.1, 1)$ is small and so many samples would have included only null $p$-values. Nonetheless, the median has a unique minimum near $\kappa = 0.01$, and is generally lower for small $\kappa$ than large $\kappa$. Considering the coordinates of this case in the $D(a,w), \prevalence$ plane, this is completely consistent with the earlier investigations of power where small $\kappa$ values were most powerful for strong evidence concentrated in a few tests. In practice, seeing the median curve would cause us to suspect this case correctly.

\subsection{Identifying a region of alternative hypotheses} \label{chipool:regions}

While these one-to-one comparisons between densities and $\kappa$ curves help to demonstrate the link between $\kappa_{\min}$ and the alternative densities used to generate $\ve{p}$, they are not incredibly informative. Given $\ve{p}$ and supposing we generate such a curve of $\chipool{\ve{p}}{\kappa}$ by $\kappa$, we would need to sort through an incredible number of density-curve pairs to identify plausible alternatives corresponding to the curve obtained.


Instead, consider a more automated approach. Given a collection of $p$-values, this generates a curve by sweeping parameter values of $\kappa$, identifies the minimum $\kappa$ values (or any below a particular threshold), and maps these back onto the plane of strength and prevalence depicted in, for example, Figure \ref{fig:BothChiDiffs}. This requires a detailed guide of where each $\kappa$ value is most powerful in the $\prevalence, D(a,w)$ plane so that $\kappa_{\min}$ or the range of significant $\kappa$ values can be placed accurately. Therefore, a simulation was carried out over 20 $\ln(w)$ values evenly spaced from $-6$ to $0$, $\prevalence$ values from 0 to 1 in increments of 1/80, and 80 $\ln D(a,w)$ values evenly spaced between $-5$ and $5$. For each combination, 10,000 samples of 80 $p$-values were generated with $80\prevalence$ following the beta distribution specified by $\ln (w)$ and $\ln D(a,w)$ and $80(1 - \prevalence)$ following the uniform distribution.\footnote{This resolution was not the only one tried, similar experiments were carried out for $M = 10$, $20,$ and $40$ and the only impact of increasing $M$, the number of steps in $\prevalence$, and the number of steps in $\ln D(a,w)$ was increasing resolution of the same patterns. This suggests that these patterns do not depend on the sample size.}

Each of the 10,000 samples then had pooled $p$-values computed over a sweep of 65 $\ln \kappa_i$ values evenly spaced from $-8$ to $8$ and the power at level $\alpha = 0.05$ was computed for the rejection rule $\chipool{\ve{p}}{\kappa_i} \leq 0.05$. The $\kappa_i$ with the greatest power for each combination corresponds to $\kappa_{\min}$ for that combination because rejection is determined by thresholding the pooled $p$-value and so a higher power implies a lower distribution of the pooled $p$-value at a given point. This distribution with a lower location will result in an equal or lower quantile curve to all others. Bivariate discretized Gaussian smoothing is applied to each $\chipool{\ve{p}}{\kappa_i}$ power surface in $\prevalence, D(a,w)$ for each $w$ value in order to obtain a smoothed estimate of the power surface minimally impacted by random binomial noise. This was completed only because none of the investigations carried out indicated discontinuities in the power or distribution of $p$-values by $\kappa_i$.

For each $w$, the power surfaces of every $\chipool{\ve{p}}{\kappa_i}$ in $\prevalence, D(a,w)$ were then compared to the maximum among them point-wise. This is motivated by the simpler case shown in Figure \ref{MC:fig:decrBetaQ}, as several $\kappa_i$ values are often equally powerful for a given setting. Specifically, the comparison was a binomial test of the difference in proportions using a normal approximation at 95\% confidence. A surface was deemed equal to the maximum power at that point if the test failed to reject the null hypothesis of equal proportions.\footnote{For powers $p_1$ and $p_2$ computed over the same number of trials $n$, this computes
  $$z = \frac{\sqrt{n}(p_1 - p_2)}{\sqrt{2p(1-p)}}$$
  where $p = \frac{p_1 + p_2}{2}$ and then compares $z$ to normal critical values. By the CLT, $z \sim N(0,1)$ approximately for large $n$, and as 10,000 simulations are performed for each power estimate, this approximation should be quite accurate.} For each $\kappa$ and $w$, all of this pre-processing gave a matrix in $\prevalence$ and $\ln D(a,w)$ indicating whether $\chipool{\ve{p}}{\kappa_i}$ achieved the maximum power for that combination for every $\kappa_i$. To produce a final summary in $\prevalence$ and $\ln D(a,w)$ alone, these indicators were summed over $w$ for each $\kappa_i$.
Finally, masks were added in the top right and bottom left corners where all methods are equally powerful with powers 1 and 0.05 respectively to make the meaningful patterns more visible.
Figures \ref{MC:fig:kappaAlternatives}(a) - (d) display these sums (counts of cases in $w$ where $\kappa_i$ achieved maximum power) for several $\kappa_i$ in a given $\prevalence, \ln D(a,w)$ region, with guide histograms on each row and column added to quickly indicate the relative marginal frequencies. Each plot is therefore rich with both marginal and joint information on the regions where a particular $\kappa_i$ is most powerful.

\begin{figure}[!h]
  \begin{center}
    \begin{tabular}{cc}
      \includegraphics[scale = 0.9]{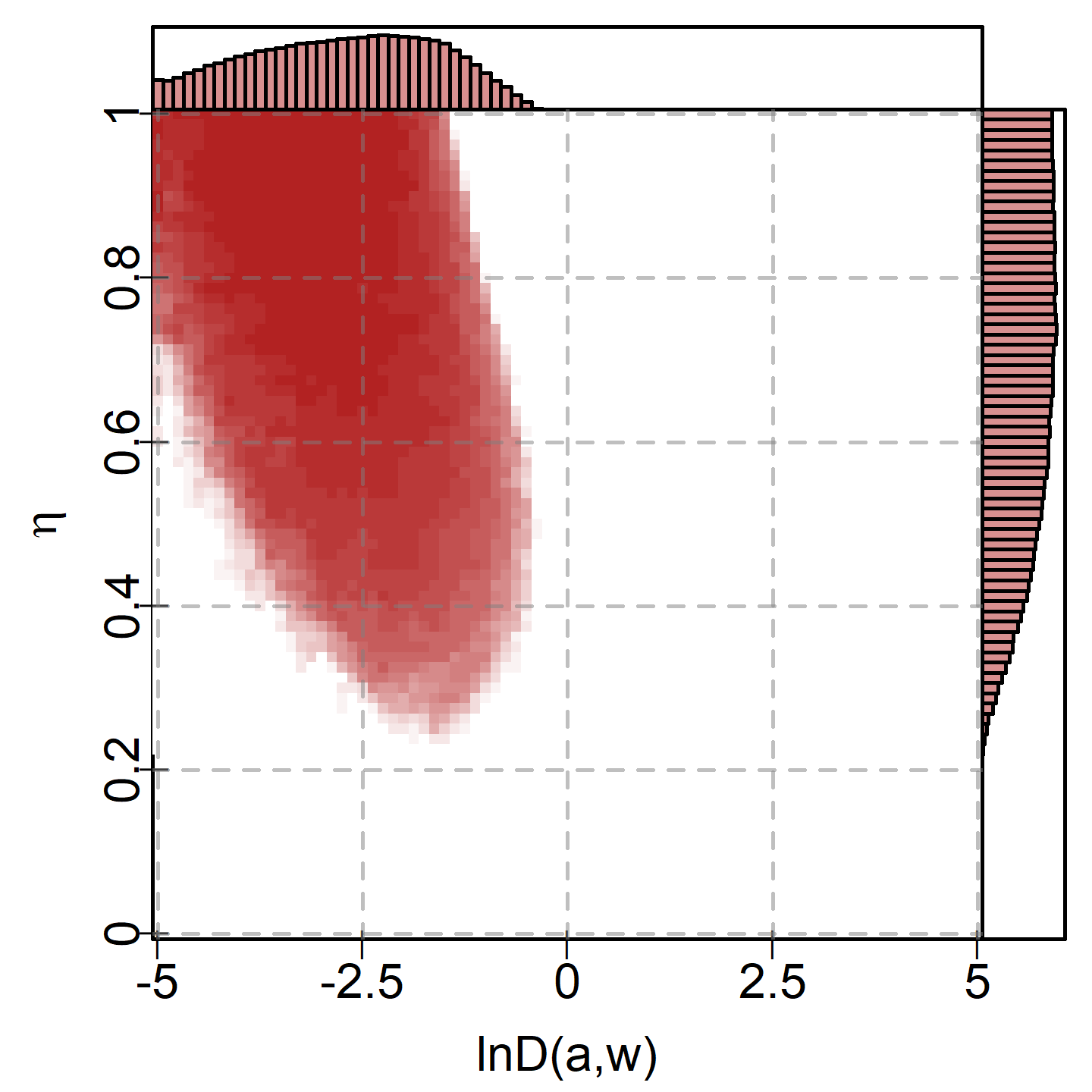} & \includegraphics[scale = 0.9]{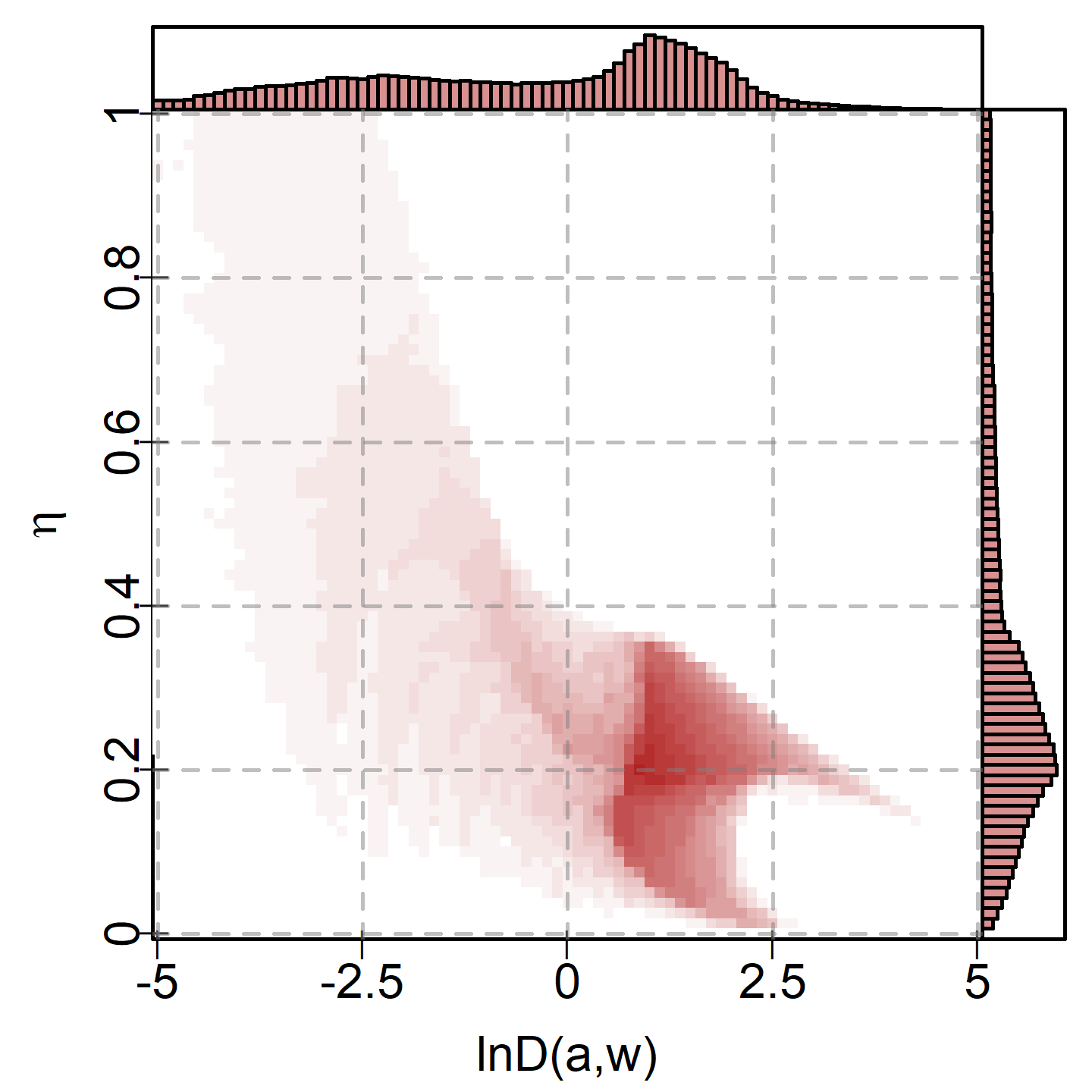} \\
      {\footnotesize a) $\ln \kappa = 8$} & {\footnotesize b) $\kappa = 2$} \\
      \includegraphics[scale = 0.9]{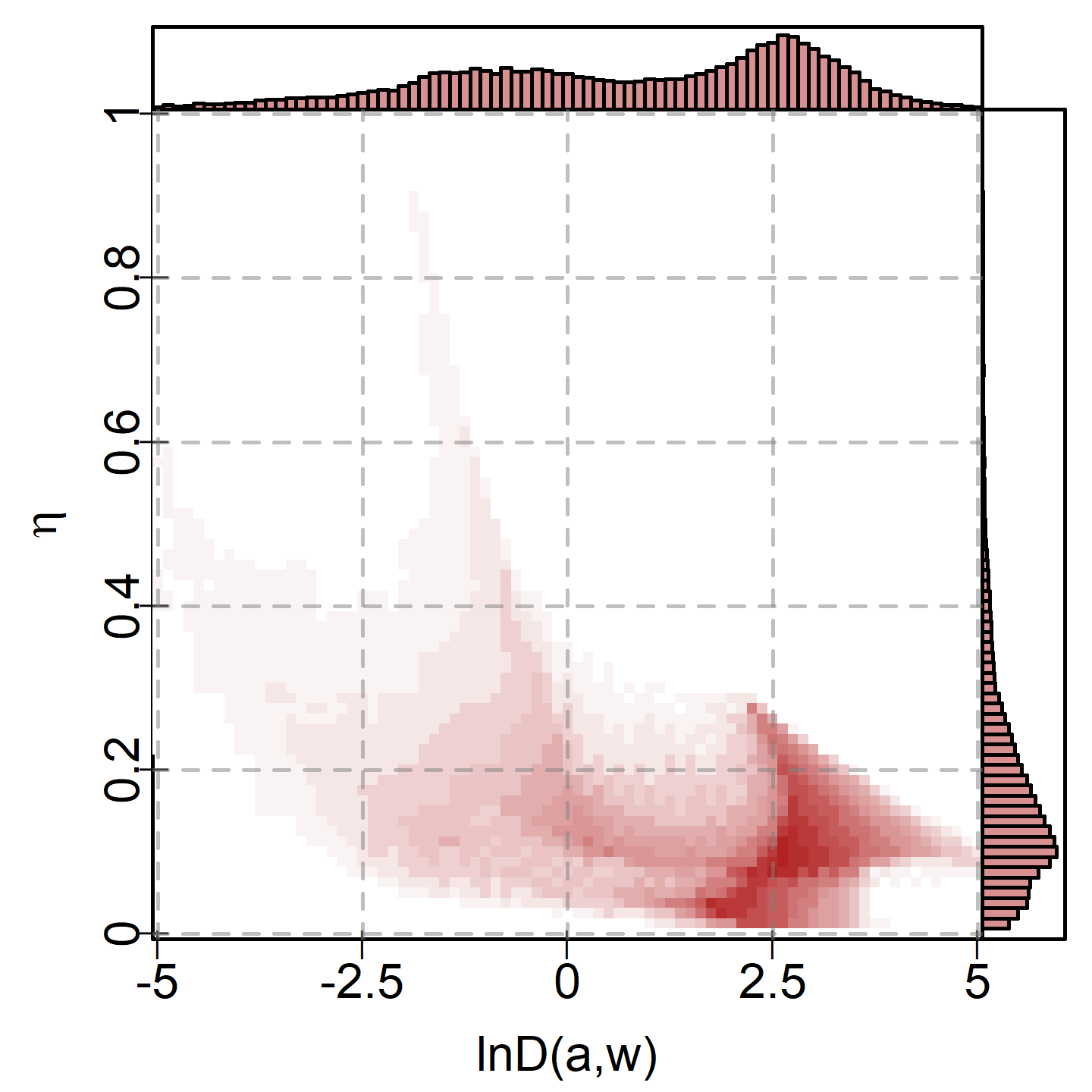} & \includegraphics[scale = 0.9]{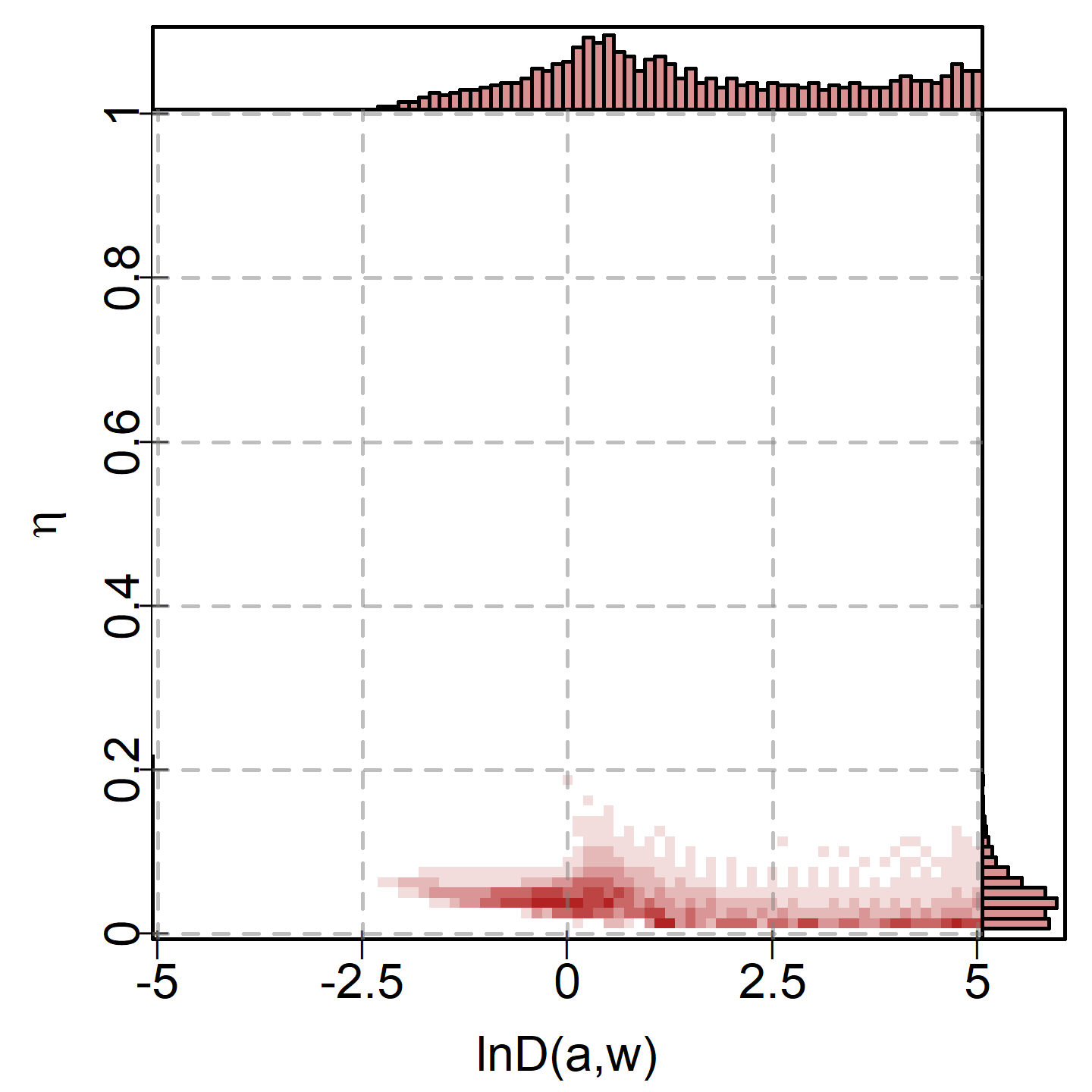} \\
      {\footnotesize c) $\ln \kappa = -1$} & {\footnotesize d) $\ln \kappa = -4$}
    \end{tabular}
    \caption{Likely alternatives for a range of $\kappa$ values. For full coverage of $D(a,w)$, $w$ was chosen uniformly on a log scale.}\label{MC:fig:kappaAlternatives}
  \end{center}
\end{figure}

Consistent with previous investigations, this map shows that the regions where the small $\kappa$ values are most powerful correspond to small $\prevalence$ values. The mode of the histogram of $\prevalence$ values increases steadily in $\kappa$ until it is near one when $\kappa = e^8$. For $\kappa = e^{-4}$, the pooled $p$-value is only most powerful for settings with $\prevalence < 0.1$, such that a minimum of the parameter curve below $e^{-4} \approx 0.02$ indicates a small minority of tests are significant.

Given that these plots display counts of cases where a particular $\chipool{\ve{p}}{\kappa}$ is most powerful, choosing to select $w$ evenly-spaced on the log scale inadvertently places greater weight on small values of $w$ and under-samples large values in order to achieve more complete coverage of $D(a,w)$. Even for moderate $w$ floating point representation limits prevent the computation of $a$ and $b$ for the beta distributions with large KL divergences. This complete coverage of the strength of evidence is one of two possible perspectives, with the other focused on even exploration of the parameter space. For this parameter-based perspective, the exact same procedure was performed but with 20 $w$ values selected at even increments from 0.05 to 1. Figures \ref{MC:fig:kappaAlternativesUnifW}(a) - (d) present the same heatmaps as Figure \ref{MC:fig:kappaAlternatives} from this perspective.

\begin{figure}[!h]
  \begin{center}
    \begin{tabular}{cc}
      \includegraphics[scale = 0.9]{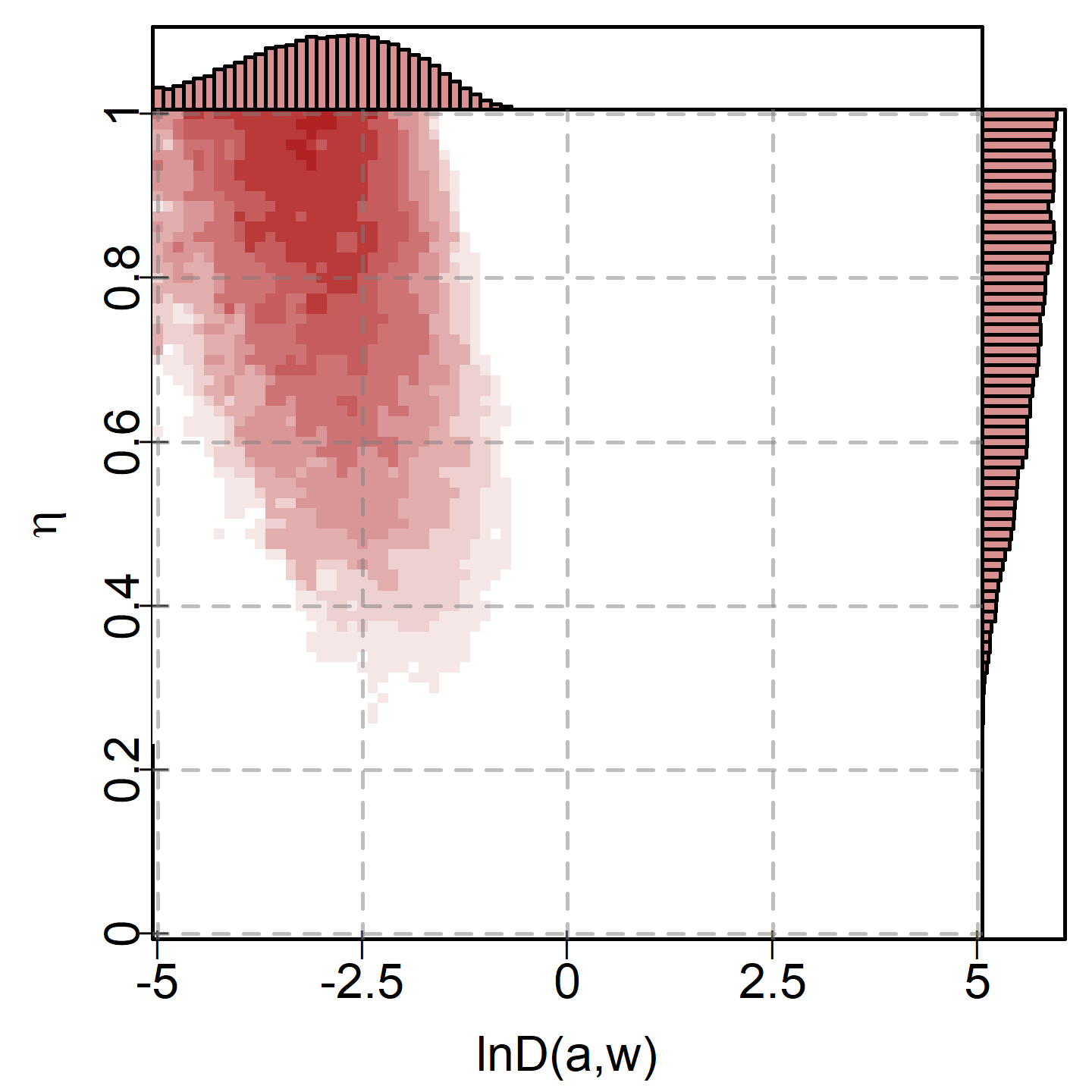} & \includegraphics[scale = 0.9]{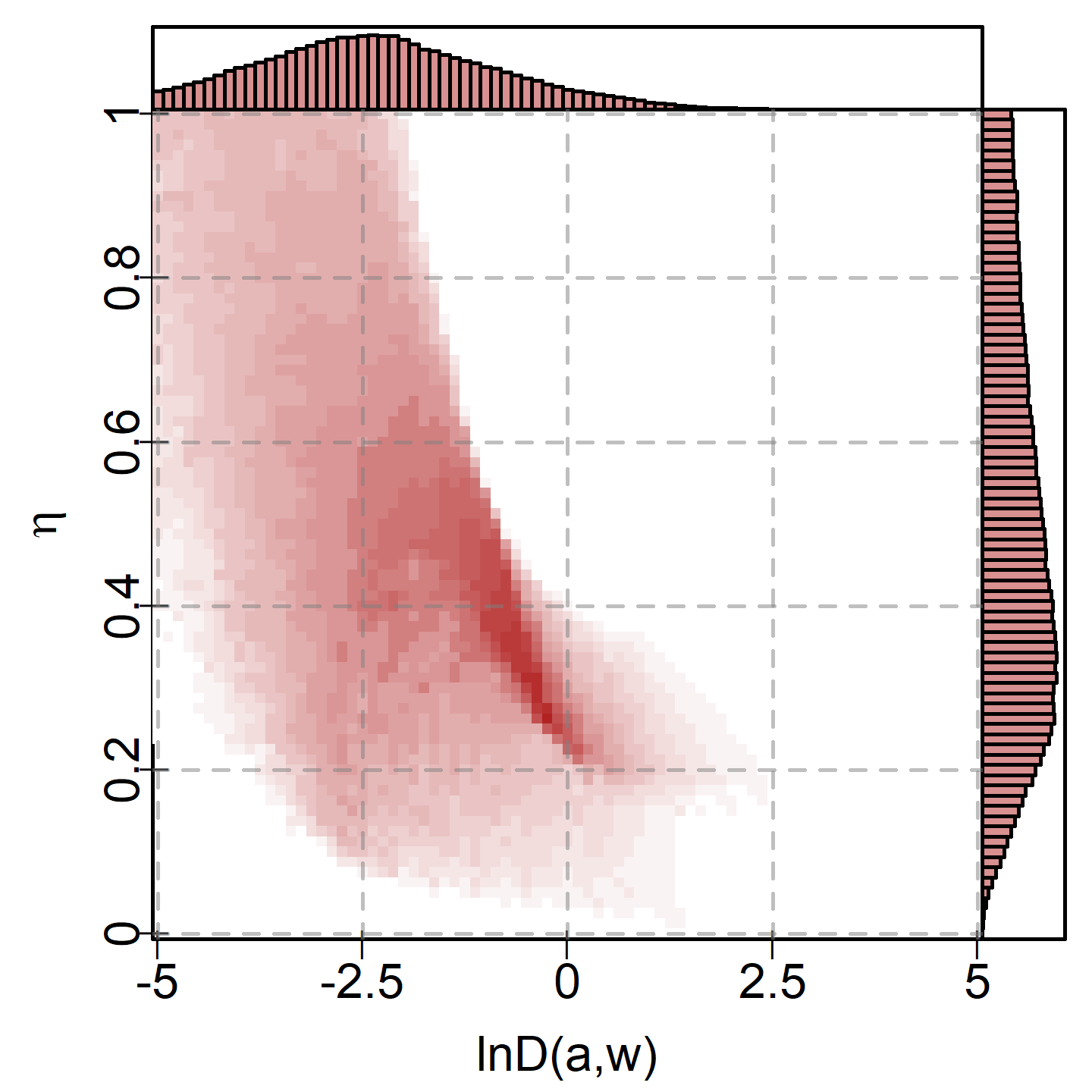} \\
      {\footnotesize a) $\ln \kappa = 8$} & {\footnotesize b) $\kappa = 2$} \\
      \includegraphics[scale = 0.9]{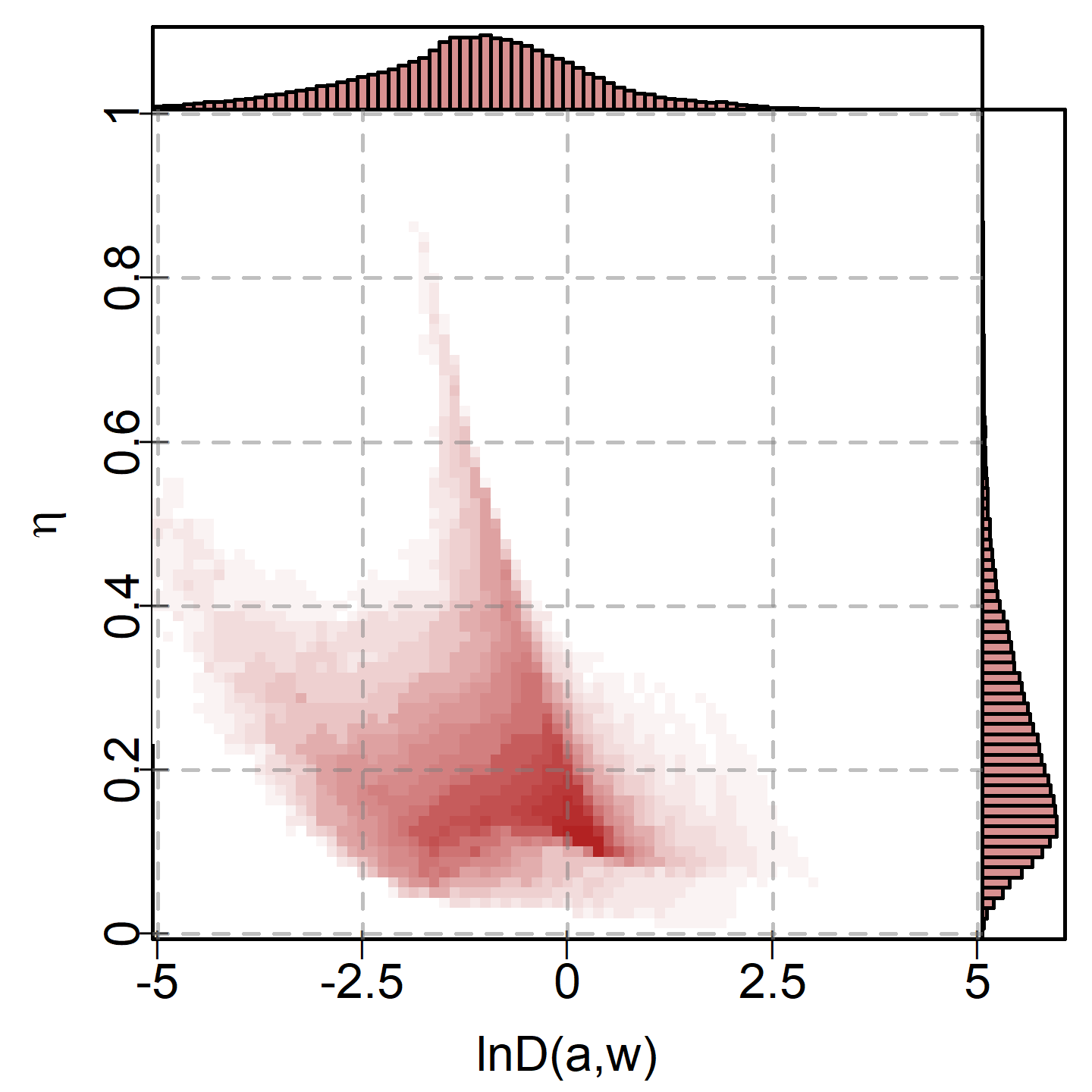} & \includegraphics[scale = 0.9]{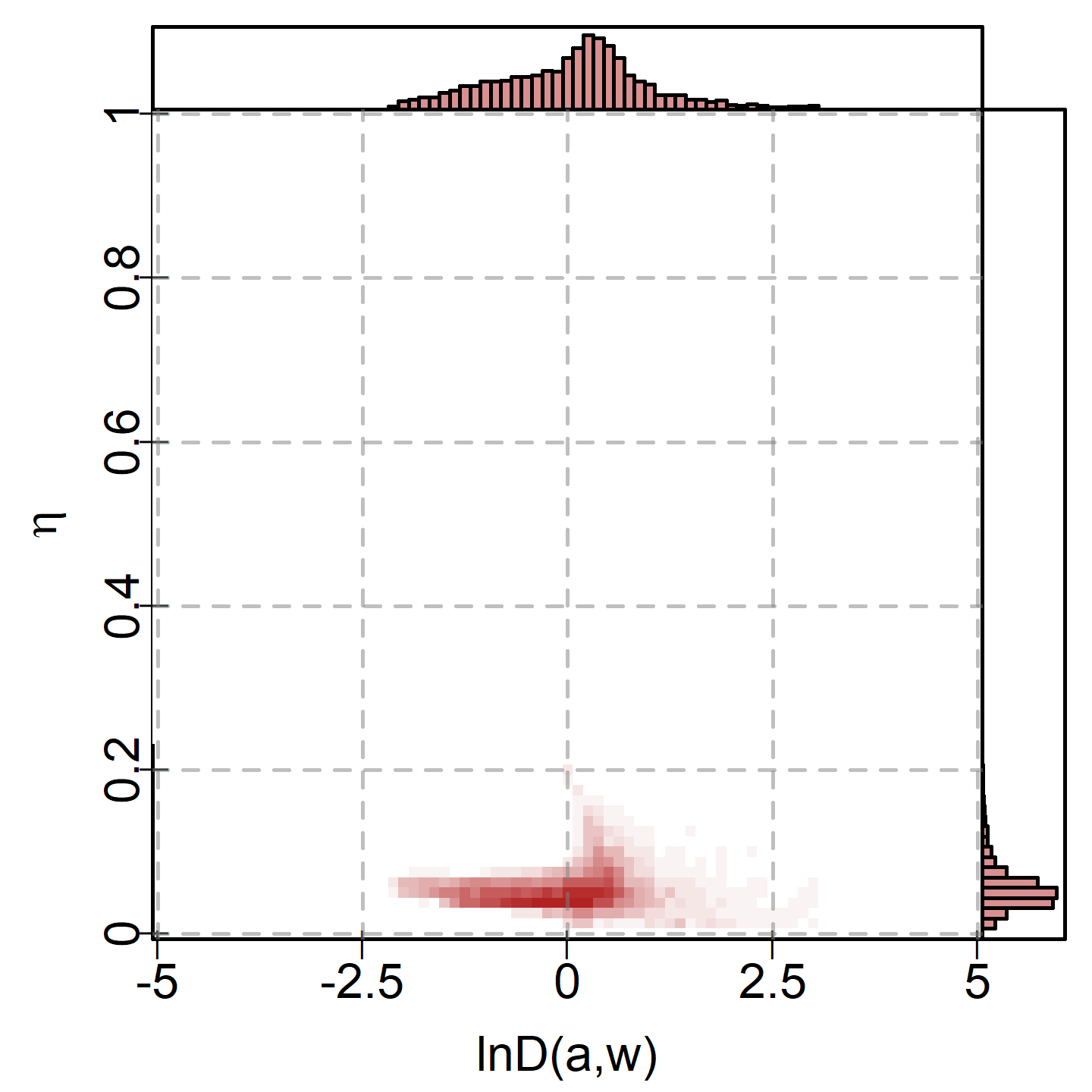} \\
      {\footnotesize c) $\ln \kappa = -1$} & {\footnotesize d) $\ln \kappa = -4$}
    \end{tabular}
    \caption{Likely alternatives for a range of $\kappa$ values. $w$ was chosen uniformly for these images, leading to worse coverage of $D(a,w)$ but more appropriate coverage of $w$.}\label{MC:fig:kappaAlternativesUnifW}
  \end{center}
\end{figure}

There are some noteworthy differences between this and the first set of heatmaps. The bias towards smaller proportions and stronger evidence in the first is quite clear when it is compared to the second, which generally shows similar shapes but more evenly distributed saturation across this shape. This leads to changes in the regions suggested for a particular $\kappa_{\min}$, but these are typically minor. The biggest difference occurs for large $\kappa$, where the bias towards small values in Figure \ref{MC:fig:kappaAlternatives} obscures all of the internal variation in the middle top that can be seen in Figure \ref{MC:fig:kappaAlternativesUnifW}.

Without these plots, an analyst would be left trying to identify alternatives from a density estimate. Besides showing comparable information about the prevalence of evidence to a density estimate in the histogram along the right side of the plot, these plots of plausible alternatives give information about likely strengths of evidence and regions for the combination of both. By leveraging the links between the centrality quotient, $\kappa$, and the distribution and strength of evidence in $\ve{p}$, these maps provide richer and clearer information.

\clearpage

\subsection{Selecting a subset of tests}

Perhaps the most important part of the alternative heatmaps presented in Figures \ref{MC:fig:kappaAlternatives} and \ref{MC:fig:kappaAlternativesUnifW} are the histograms along the right margin that indicate the likely prevalence of evidence in the data. Once $\kappa_{\min}$ has been determined using a sweep of $\kappa$ values, and a plausible set of alternatives has been identified using these alternative heatmaps, the corresponding range of proportions can be used to identify a subset of tests of interest. If false positives are less problematic to analysis than false negatives, the upper bound of this range might be taken, with the other bound taken if the opposite is true. In either case, suppose the chosen proportion is $\prevalence^*$, then the $M\prevalence^*$ largest values of $F^{-1}(1 - p_i; \kappa_{\min})$ are the tests most contributing to the small value of $\chipool{\ve{p}}{\kappa_{\min}}$ and so are the tests of greatest interest that can be selected for further investigation.

\subsection{Centrality in other beta densities}

Until now, it was always assumed that $p$-values follow a non-increasing beta density when the null hypothesis is false. This is a reasonable assumption, many statistical tests have this property for the rejection rule thresholding the $p$-value at $\alpha$. Relaxing this assumption, however, allows an exploration into how $\chipool{\ve{p}}{\kappa}$ behaves for a broader variety of densities and whether centrality is still a useful concept under these other distributions of non-null $p$-values.

First, consider the case of a strictly increasing density under $H_4$. Whether or not this case is interesting is a matter of opinion as under the convention that small $p$-values are evidence against $H_0$ such a density produces even less evidence against $H_0$ than the null distribution itself. It would be reasonable to expect, then, that this case produces only very large $\chipool{\ve{p}}{\kappa}$ for all $\kappa$ values. Following the same procedure as Section \ref{chipool:noninc}, this expectation is tested for $Beta(1, 0.5)$, resulting in the curve and density displayed in Figure \ref{MC:fig:incrBetaQ}.

\begin{figure}[!h] 
  \begin{center}
    \begin{tabular}{cc}
      \includegraphics[scale = 1]{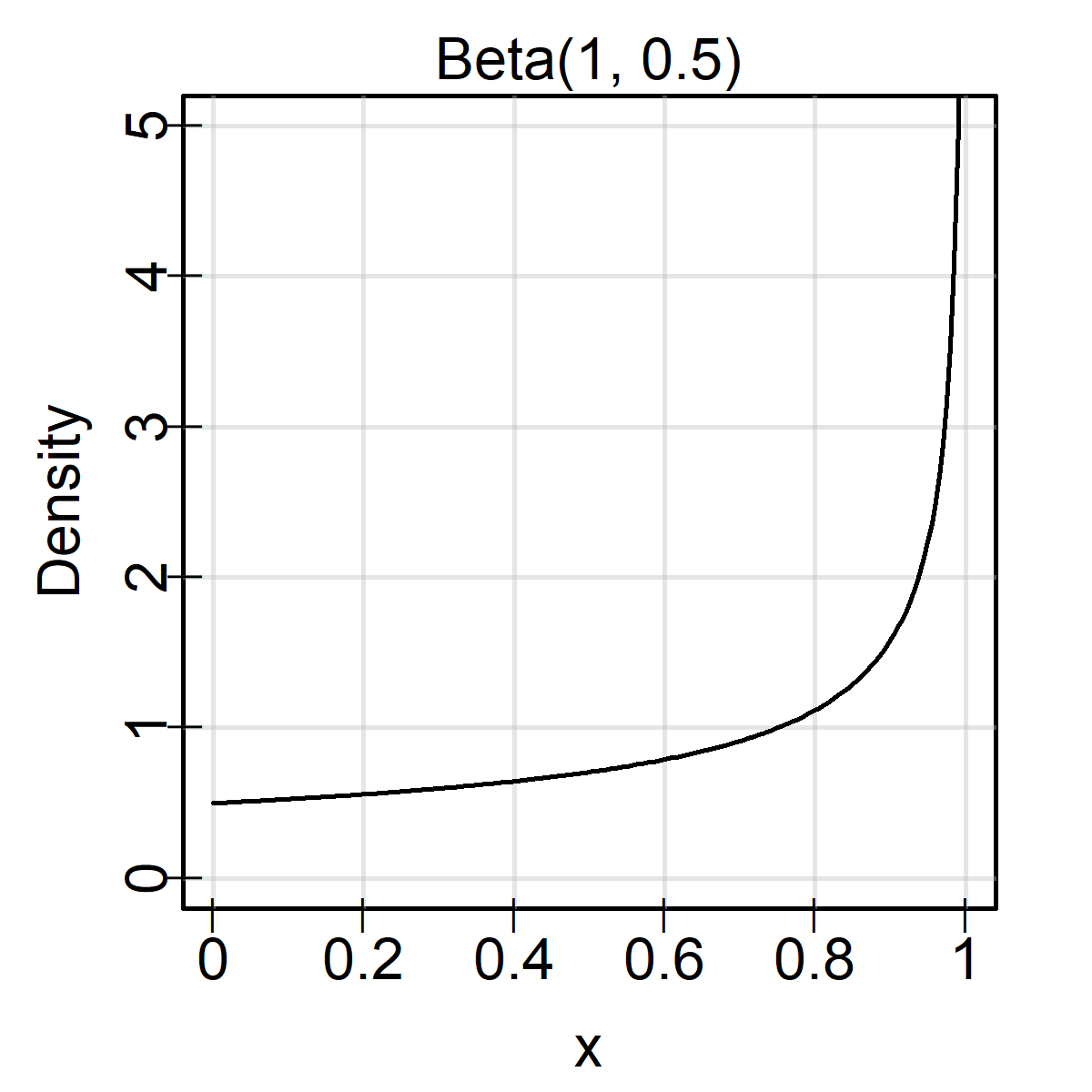} & \includegraphics[scale = 1]{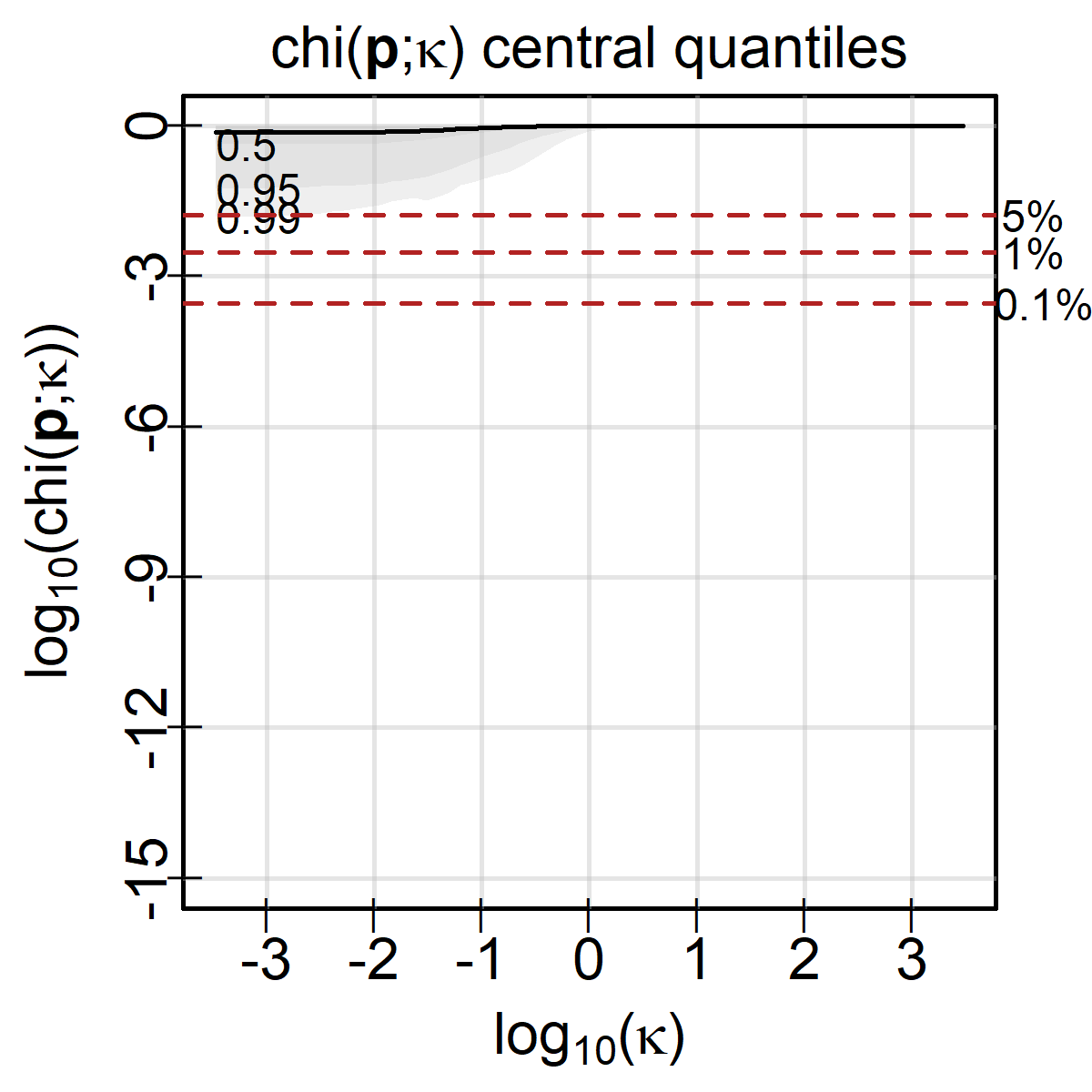} \\
      {\footnotesize (a)} & {\footnotesize (b)}
    \end{tabular}
    \caption{The (a) density and (b) central quantiles and median by $\kappa$ for $p$-values from a Beta(1, 0.5) distribution. The pooled $p$-value is generally large compared to the empirical curve minimum quantiles.}
    \label{MC:fig:incrBetaQ}
  \end{center}
\end{figure}

As expected, this setting gives only large $\chipool{\ve{p}}{\kappa}$ values for every $\kappa$. There is a slight dip to small $p$-values for small $\kappa$, likely a result of the small $p$-values that still occur for this density occasionally, but it barely crosses the null reference lines. Perhaps a more realistic case is a density that rarely, if ever, produces minimum $p$-values small enough to warrant rejection alone, but does tend to produce far more $p$-values less than 0.5 than expected under the null hypothesis. For such a setting, the previous investigations into centrality suggest that $\kappa_{\min}$ should be large. An example is $f = Beta(2,4)$ under $H_4$, shown in Figure \ref{MC:fig:lowbiasBetaQ}.

\begin{figure}[!h] 
  \begin{center}
    \begin{tabular}{cc}
      \includegraphics[scale = 1]{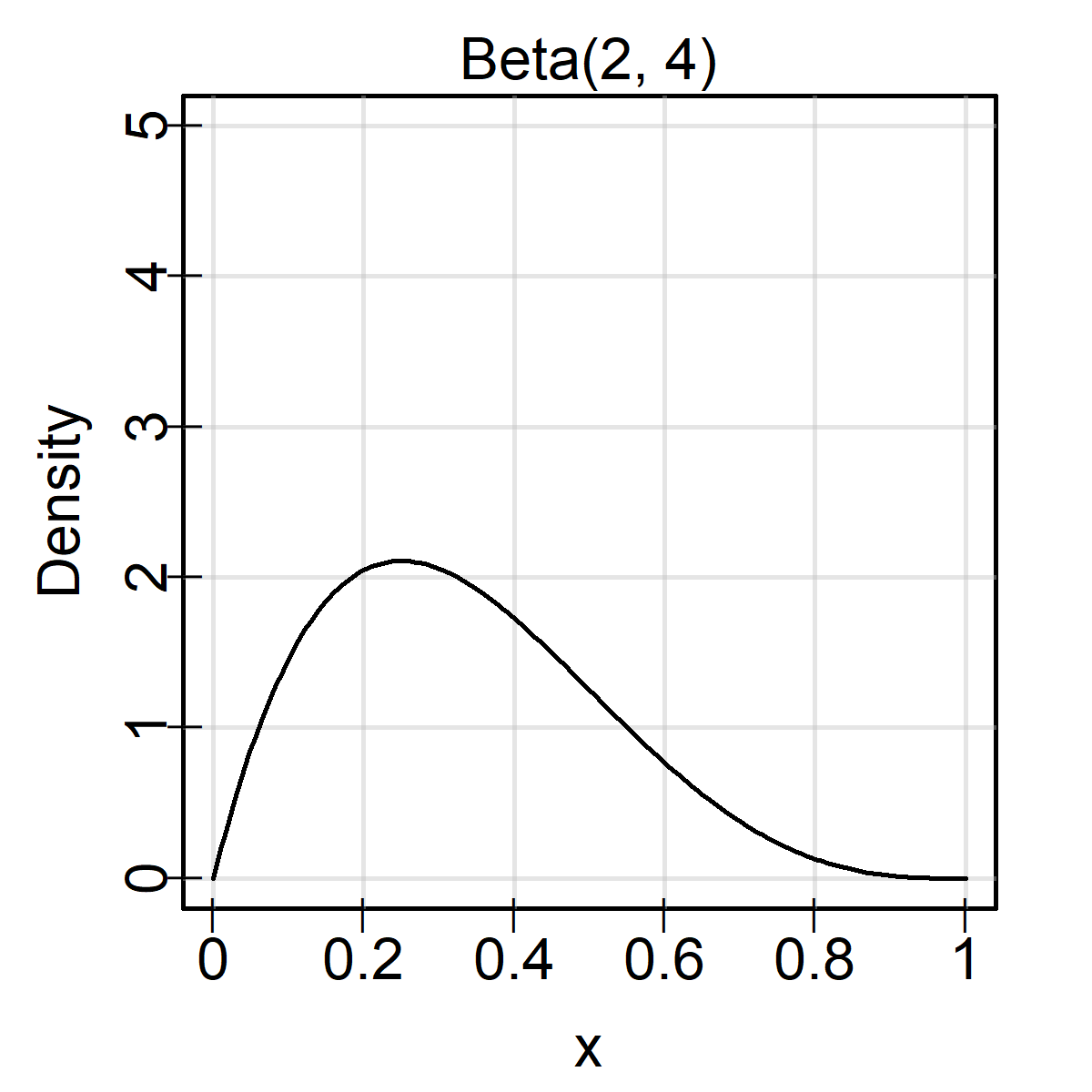} & \includegraphics[scale = 1]{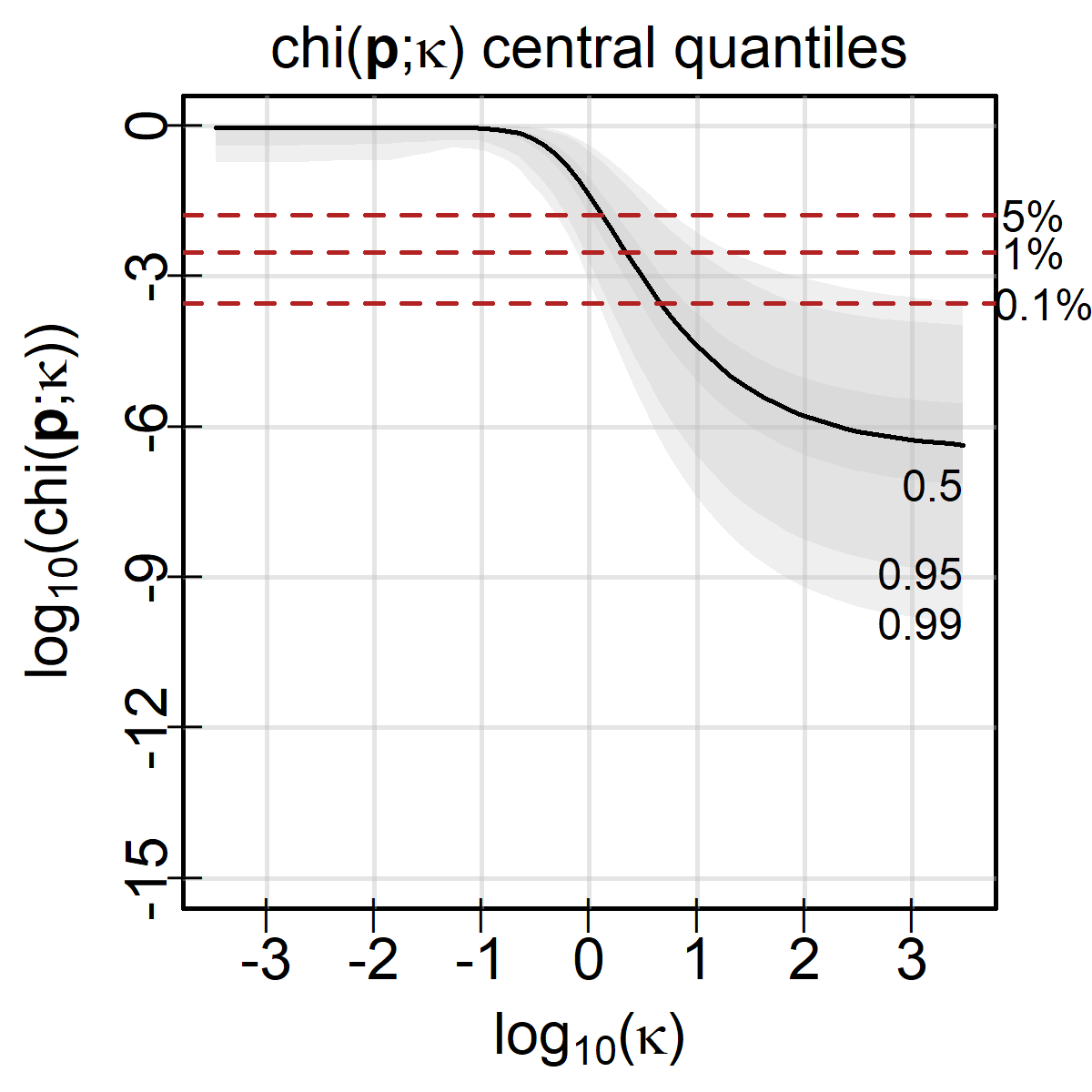} \\
      {\footnotesize (a)} & {\footnotesize (b)}
    \end{tabular}
    \caption{The (a) density and (b) central quantiles and median of a $p$-values following a Beta(2, 4) distribution. The absence of very small $p$-values and bias towards smaller ones means large $\kappa$ values are most powerful.}
    \label{MC:fig:lowbiasBetaQ}
  \end{center}
\end{figure}

Once again, this plot matches the expectation reasoned from centrality, despite the relaxation of the assumptions used to motivate centrality. Both of these examples suggest that the concepts of central and marginal rejection may have use beyond non-increasing densities, and provide a promising framework for future investigation.

\clearpage

\section{The \code{PoolBal} package}

There is no lack of packages available to pool $p$-values in \Rnsp. The most recent of these, \code{poolr} (\citealp{poolr}), lists 8 others all providing piecemeal coverage of every pooled $p$-value proposal\footnote{These are \cite{dewey2022metap} \cite{zhangetal2020tfisher}; \cite{wilson2019harmonic}; \cite{aggregation}; \cite{empbrown}; \cite{daietal2014modified}; \cite{schroderetal2011survcomp}; \cite{zhao2008gap}. Most of these packages cover a subset of pooling functions or implement adjustments for dependence rather than attempting to be the complete package for pooling $p$-values.}. Rather than re-implement the functionality provided by these packages, the \code{PoolBal} package aims primarily to support the evaluation of the central rejection level, marginal rejection level, and centrality quotient for these and any future packages which pool $p$-values. As they both allow some tuning of centrality, these core functions are supported by implementations of $\chipool{\ve{p}}{\kappa}$ and $\hrpool{\ve{p}}{\kappa}$ along with functions to evaluate the Kullback-Leibler divergence for general densities and compute it for the beta density in particular. This is meant to make the adoption of the framework provided in this work as simple as possible.

Briefly summarized, the functionality of \code{PoolBal} includes
\begin{description}
\item[\code{klDiv}, \code{betaDiv}:] compute the Kullback-Leibler divergence for arbitrary densities and the uniform to Beta case, respectively
\item[\code{findA}:] invert a given Kullback-Leibler divergence and most powerful test parameter $w$ to identify the unique Beta parameter $a$ that corresponds to this setting
\item[\code{pBetaH4}, \code{pBetaH3}:] helpers to generate $\ve{p}$ under under $H_4$ and $H_3$
\item[\code{estimatePc}, \code{estimatePrb}, \code{estimateQ}:] wrappers for \code{uniroot} from \code{base} that estimate the central rejection level, marginal rejection level at $b$, and centrality quotient for an arbitrary function
\item[\code{chiPool}:] an implementation of $\chipool{\ve{p}}{\kappa}$
\item[\code{chiPc}, \code{chiPr}, \code{chiQ}:] functions to compute the central rejection level, marginal rejection level, and centrality quotient of $\chipool{\ve{p}}{\kappa}$ using Equations (\ref{MC:eq:chiMethPc}), (\ref{MC:eq:chiMethPr}), and (\ref{MC:eq:chiMethC})
\item[\code{chiKappa}:] a wrapper for \code{uniroot} from \code{base} that inverts a given centrality quotient to give the $\kappa$ value in $\chipool{\ve{p}}{\kappa}$ with the corresponding quotient
\item[\code{hrStat}, \code{hrPool}:] compute $l_{HR}(\ve{p}; w)$ and $\hrpool{\ve{p}}{w}$ for $\ve{p}$, with the $p$-value determined empirically using simulated null data
\item[\code{hrPc}, \code{hrPr}, \code{hrQ}:] functions to compute the central rejection level, marginal rejection level, and centrality quotient of $\hrpool{\ve{p}}{\kappa}$ using simulation and \code{uniroot}
\item[\code{altFrequencyMat}:] function allowing access to a summarized version of the simulation results from Section \ref{chipool:regions}
  \item[\code{marHistHeatMap}:] function which generates heatmaps with marginal histograms, that is visualizations such as those in Figure \ref{MC:fig:kappaAlternativesUnifW}.
\end{description}
The package can be found on \href{https://github.com/Salahub/chi-pooling}{the author's GitHub}.

\section{Conclusion}

When presented with $M$ $p$-values from independent tests of hypotheses $H_{01}, \dots, H_{0M}$, a natural way to control the family-wise error rate (FWER) is by pooling these $p$-values using a function $g(\ve{p})$. If $g(\ve{p})$ is constructed using the sum of quantile transformations or the order statistics of the $p$-values, then the rejection rule $g(\ve{p}) \leq \alpha$ controls the FWER at $\alpha$. Selecting between the many possible $g(\ve{p})$ requires the choice of an alternative hypothesis from the telescoping set $H_1 \supset H_2 \supset H_3 \supset H_4$ in order to determine their powers against these alternatives. $H_3$ and $H_4$, though the most restrictive, still require the choice of $\prevalence$, the prevalence of non-null $p$-values in $\ve{p}$, and $f$, the distribution of these non-null values. An obvious choice for $f$ is the beta distribution restricted to be non-decreasing, as this biases non-null $p$-values lower than null $p$-values. By using $\prevalence$ and the Kullback-Leibler divergence of $f$ from the uniform distribution, both the prevalence and strength of non-null evidence can be measured.

If all the evidence is non-null, i.e. $\prevalence = 1$, the pooled $p$-value based on $l_w(\ve{p}) = w \sum_{i = 1}^M \ln p_i - (1 - w) \sum_{i = 1}^M \ln ( 1 - p_i )$ is uniformly most powerful (UMP) but is sensitive to the specification of its parameter $w \in [0,1]$. Incorrectly choosing this parameter, i.e. selecting a value that does not match the true generative distribution, costs power to reject the alternative hypothesis $H_4$. When $\prevalence \neq 1$, both the prevalence and strength of non-null evidence dictate the most powerful choice of $w$. Small values of $w$ are more powerful for weak evidence spread among all tests while large values are better at detecting strong evidence in a few tests.

This reflects a more universal pattern in pooled $p$-values and motivates a new paradigm for selecting and analyzing them. The marginal level of rejection at $\alpha$, the largest individual $p$-value that leads to rejection at $\alpha$ when all other $p$-values are 1, and the central rejection level at $\alpha$, the largest value simultaneously taken by all elements of $\ve{p}$ which still leads to rejection at $\alpha$, characterize this paradigm. By defining the central and marginal rejection level, a number of fundamental properties can be proven. Among them, the central rejection level of a pooled $p$-value satisfying some mild conditions is always greater than or equal to the marginal rejection level, with equality occurring only for $\tippool(\ve{p})$. This order allows a centrality quotient to be defined which summarizes the preference of a pooled $p$-value to diffuse or concentrated evidence with a value in $[0,1]$.

In order to control this quotient, a pooled $p$-value based on $\chi^2_{\kappa}$ quantile transformations was defined, $\chipool{\ve{p}}{\kappa}$. By choosing the degrees of freedom $\kappa \in [0, \infty)$, arbitrary control over the centrality of the pooled $p$-value is obtained. Increasing $\kappa$ raises the centrality quotient, and decreasing it drops the quotient. Furthermore, the limiting cases of $\kappa = 0$ and $\kappa \rightarrow \infty$ correspond to $\tippool(\ve{p})$, the minimum order statistic $p$-value, and $\stopool(\ve{p})$, the normal quantile transformation $p$-value. Both of these limiting pooled $p$-values are classic pooling functions which have been used and studied widely in the literature. $\chipool{\ve{p}}{\kappa}$ therefore provides a means to balance an important aspect of pooling $p$-values with a single parameter that has ready interpretation along its range. Comparing its power to $\hrpool{\ve{p}}{w}$ under $H_3$ and $H_4$, $\chipool{\ve{p}}{\kappa}$ loses less power than $\hrpool{\ve{p}}{w}$ with $w$ mis-specified. $\chipool{\ve{p}}{\kappa}$ is therefore more robust, and demonstrates that the central and marginal rejection paradigm is instructive to predict which version of $\chipool{\ve{p}}{\kappa}$ will be most powerful for a particular alternative hypothesis. $\chipool{\ve{p}}{\kappa}$ and the centrality quotient are both potent tools for pooling $p$-values to control the FWER.

\bibliographystyle{plainnat}
\renewcommand*{\bibname}{References} 
\bibliography{./fullbib.bib}

\end{document}